\documentclass[hyperref]{lmcs}
\pdfoutput=1

\usepackage{lastpage}
\renewcommand{\dOi}{14(3:1)2018}
\lmcsheading{}{1--\pageref{LastPage}}{}{}%
{Sep.~06,~2017}{Jul.~20,~2018}{}

\usepackage{pdfsync}

\usepackage{graphicx}

\usepackage{tikz}
\usetikzlibrary{decorations.pathreplacing,hobby,backgrounds,trees,arrows,automata,shapes,positioning,fit,patterns,calc,matrix}

\usepackage[english]{babel}

\usepackage{enumitem}

\setlist[1,enumerate]{leftmargin=*}
\setlist[2,enumerate]{label=$\alph*)$}
\setlist[1,itemize]{label=--}
\usepackage{booktabs}

\usepackage{xcolor}

\usepackage{enumitem}

\usepackage{xifthen}
\usepackage{xfrac}
\usepackage[defaultlines=2,all]{nowidow}
\usepackage[ugly]{nicefrac}
\usepackage{xspace}

\usepackage{dutchcal}

\usepackage[utf8]{inputenc}
\usepackage[T1]{fontenc}

\usepackage{csquotes}

\usepackage{calc}

\usepackage{titletoc}

\usepackage{amsmath,amsfonts,amssymb,amsthm}

\usepackage{hyperref}
\hypersetup{breaklinks,pdfencoding=auto,psdextra,bookmarksopen=true}
\usepackage{bookmark}

\bookmarksetup{
  open,
  numbered,
  addtohook={
    \ifnum\bookmarkget{level}=0
    \bookmarksetup{bold}
    \fi
    \ifnum\bookmarkget{level}=-1
    \bookmarksetup{bold}
    \fi
  }
}

\definecolor{ocre}{RGB}{40,130,80}
\definecolor{bookgreen}{RGB}{40,130,80}
\definecolor{bookblue}{RGB}{50,110,150}
\definecolor{bookred}{RGB}{180,15,47}

\newcommand\red[1]{\textcolor{bookred}{#1}}

\tikzstyle{char}=[anchor=mid,inner sep=0pt]
\tikzstyle{pebb}=[line width=1.5pt,->]
\tikzstyle{fiar}=[shorten >= 1pt,thick,->]
\tikzstyle{bag}=[inner sep=1pt]

\tikzstyle{non}=[inner sep=1pt]
\tikzstyle{lbox}=[rounded corners=5pt,draw=black!60,very thick,align=center]
\tikzstyle{wbox}=[rounded corners=5pt,align=center,minimum height=0.65cm,inner ysep=0pt]
\tikzstyle{gbox}=[rounded corners=5pt,fill=green!20,align=center,minimum height=0.65cm,inner ysep=0pt]
\tikzstyle{bbox}=[rounded corners=5pt,fill=blue!20,align=center,minimum height=0.65cm,inner ysep=0pt]
\tikzstyle{rbox}=[rounded corners=5pt,fill=red!20,align=center,minimum height=0.65cm,inner ysep=0pt]
\tikzstyle{ledg}=[draw=black,very thick]

\tikzstyle{linc}=[fill=white,draw=black,rotate=90,inner sep=2pt,very thick,circle]
\tikzstyle{linc2}=[fill=white,draw=black,rotate=0,inner sep=2pt,very thick,circle]
\tikzstyle{leg}=[draw,minimum width = 1.0cm,minimum height = 0.5cm]
\tikzstyle{tag}=[draw,fill=white,sloped,circle,inner sep=1pt]

\tikzstyle{word}=[draw,very thick,rectangle,rounded corners=2pt,anchor=west,minimum height=0.5cm]
\tikzstyle{word2}=[rectangle,rounded corners=2pt,anchor=west,minimum height=0.5cm]

\tikzstyle{snode}=[circle,draw=black!70,thick,align=center,minimum size = 0.5cm,inner sep =0pt]
\tikzstyle{bobox}=[rectangle,rounded corners=5pt,draw=black!70,thick,fill=orange!20]
\tikzstyle{triangle} =[regular polygon, regular polygon sides=3,rounded corners=5pt,draw=black!70,fill=gray!20]
\tikzstyle{lines} =[draw=black!70,thick]

\tikzstyle{stmono}=[inner sep=2pt]

\tikzset{every state/.style={draw=blue!50!green,very thick,fill=blue!50!green!20}}
\tikzset{statesub/.style={state,minimum size=1.3cm,inner sep=1pt}}
\tikzset{pattstate/.style={state,draw=red!50!yellow,line width=2pt,fill=red!50!yellow!20}}
\tikzset{pdotstate/.style={state,minimum size=0.75cm,inner sep=0.5pt,draw=red!50!yellow,line
    width=2pt,dashed,fill=red!50!yellow!20}}
\tikzstyle{trans}=[shorten >= 1pt,thick,->]

\makeatletter
\tikzstyle{initial by arrow}=   [after node path=
{
  {
    [to path=
    {
      [->,double=none,shorten >= 1pt,thick,every initial by arrow]
      ([shift=(\tikz@initial@angle:\tikz@initial@distance)]\tikztostart.\tikz@initial@angle)
      node [shape=coordinate,anchor=\tikz@initial@anchor,draw] {\tikz@initial@text}
      -- (\tikztostart)}]
    edge ()
  }
}]
\makeatother

\makeatletter
\tikzstyle{accepting by arrow}=   [after node path=
{
  {
    [to path=
    {
      [->,double=none,shorten >= 1pt,thick,every accepting by arrow]
      (\tikztostart) --
      ([shift=(\tikz@accepting@angle:\tikz@accepting@distance)]\tikztostart.\tikz@accepting@angle)
      node [shape=coordinate,anchor=\tikz@accepting@anchor,draw] {\tikz@accepting@text}
    }
    ]
    edge ()
  }
}]
\makeatother

\newcommand\mvline[3][]{
  \pgfmathtruncatemacro\hc{#3-1}
  \draw[#1]({$(#2-1-#3)!.5!(#2-1-\hc)$} |- #2.north) -- ({$(#2-1-#3)!.5!(#2-1-\hc)$} |- #2.south);
}
\newcommand\mhline[3][]{
  \pgfmathtruncatemacro\hc{#3-1}
  \node[fit=(#2-#3-1),inner sep=0pt](R){};
  \node[fit=(#2-\hc-1),inner sep=0pt](L){};
  \node (K) at ($(R)!0.5!(L)$) {};
  \draw[#1] (K -| #2.west) -- (K -| #2.east);
}

\newcommand{\As}{\ensuremath{\mathcal{A}}\xspace}
\newcommand{\Bs}{\ensuremath{\mathcal{B}}\xspace}
\newcommand{\Cs}{\ensuremath{\mathcal{C}}\xspace}
\newcommand{\Ds}{\ensuremath{\mathcal{D}}\xspace}

\newcommand{\Is}{\ensuremath{\mathcal{I}}\xspace}

\newcommand{\Ps}{\ensuremath{\mathcal{P}}\xspace}

\newcommand{\Gb}{\ensuremath{\mathbf{G}}\xspace}
\newcommand{\Hb}{\ensuremath{\mathbf{H}}\xspace}

\newcommand{\Kb}{\ensuremath{\mathbf{K}}\xspace}
\newcommand{\Lb}{\ensuremath{\mathbf{L}}\xspace}

\newcommand{\Qb}{\ensuremath{\mathbf{Q}}\xspace}

\newcommand{\fo}{\ensuremath{\textup{FO}}\xspace}
\newcommand{\fow}{\mbox{\ensuremath{\fo(<)}}\xspace}

\newcommand{\mso}{\ensuremath{\textup{MSO}}\xspace}

\newcommand{\fod}{\ensuremath{\fo^2}\xspace}

\newcommand{\sic}[1]{\ensuremath{\Sigma_{#1}}\xspace}

\newcommand{\sici}{\sic{i}}

\newcommand{\sicu}{\sic{1}}

\newcommand{\sicd}{\sic{2}}

\newcommand{\sict}{\sic{3}}

\newcommand{\pic}[1]{\ensuremath{\Pi_{#1}}\xspace}

\newcommand{\picu}{\pic{1}}

\newcommand{\picd}{\pic{2}}

\newcommand{\pict}{\pic{3}}

\newcommand{\bsc}[1]{\ensuremath{\Bs\Sigma_{#1}}\xspace}

\newcommand{\bsci}{\bsc{i}}

\newcommand{\bscu}{\bsc{1}}

\newcommand{\bscd}{\bsc{2}}

\newcommand{\bsct}{\bsc{3}}

\newcommand{\ptp}[1]{\ensuremath{#1\textup{-}\textup{PT}}\xspace}
\newcommand{\kpt}{\ptp{k}}

\newcommand{\at}{\ensuremath{\textup{AT}}\xspace}

\newcommand{\bool}[1]{\ensuremath{\mathit{Bool}(#1)}\xspace}

\newcommand{\eqfod}[1]{\ensuremath{\cong_{#1}}\xspace}
\newcommand{\keqfod}{\eqfod{k}}

\newcommand{\nfa}{{{NFA}}\xspace}
\newcommand{\nfas}{{\textup{NFAs}}\xspace}

\newcommand{\vari}{\kl[quotienting]{quotienting} \kl[Boolean algebra]{Boolean algebra}\xspace}

\newcommand{\varis}{\kl[quotienting]{quotienting} \kl[Boolean algebras]{Boolean algebras}\xspace}

\newcommand{\pvari}{\kl[quotienting]{quotienting} \kl[lattice]{lattice}\xspace}

\newcommand{\pvaris}{\kl[quotienting]{quotienting} \kl[lattices]{lattices}\xspace}

\newcommand{\imprint}{imprint\xspace}
\newcommand{\imprints}{imprints\xspace}
\newcommand{\Imprint}{Imprint\xspace}
\newcommand{\Imprints}{Imprints\xspace}

\newcommand{\tame}{multiplicative\xspace}

\newcommand{\Ratms}{Rating maps\xspace}
\newcommand{\Ratm}{Rating map\xspace}
\newcommand{\ratms}{rating maps\xspace}
\newcommand{\ratm}{rating map\xspace}

\newcommand{\Nice}{Nice\xspace}
\newcommand{\nice}{nice\xspace}

\newcommand{\mratm}{multiplicative rating map\xspace}
\newcommand{\mratms}{multiplicative rating maps\xspace}
\newcommand{\Mratm}{Multiplicative rating map\xspace}
\newcommand{\Mratms}{Multiplicative rating maps\xspace}

\newcommand{\prin}[2]{\ensuremath{\Is[#1](#2)}\xspace}
\newcommand{\lprin}[1]{\prin{\rho_\Lb}{#1}}

\newcommand{\itriv}[1]{\ensuremath{\Is_{\mathit{triv}}[#1]}\xspace}

\newcommand{\ptriv}[1]{\ensuremath{\Ps_{\mathit{triv}}[#1]}\xspace}

\newcommand{\opti}[2]{\ensuremath{\Is_{#1}[#2]}\xspace}
\newcommand{\copti}[1]{\opti{\Cs}{#1}}

\newcommand{\popti}[2]{\ensuremath{\Ps_{#1}[#2]}\xspace}
\newcommand{\pcopti}[1]{\popti{\Cs}{#1}}

\newcommand{\bsuopti}{\opti{\bscu}{\rho}}

\newcommand{\ptopti}{\opti{\bscu}{\rho}}

\newcommand{\fodopti}{\opti{\fod}{\rho}}

\newcommand{\foopti}{\opti{\fo}{\rho}}

\newcommand{\suopti}{\popti{\sicu}{\alpha,\rho}}

\newcommand{\sdopti}{\popti{\sicd}{\alpha,\rho}}

\newcommand{\inv}{\ensuremath{^{-1}}\xspace}

\newrobustcmd{\cont}[1]{\ensuremath{\kl[cont]{\bf alph}(#1)}\xspace}

\newcommand{\dclos}{\ensuremath{\mathord{\downarrow}}\xspace}

\newcommand{\nat}{\ensuremath{\mathbb{N}}\xspace}
\newcommand{\efgame}{Ehrenfeucht-Fra\"iss\'e\xspace}

\newrobustcmd{\fullcont}[1]{{{#1}^{\kl[fullcont]\circledast}}}

\newrobustcmd{\quotienting}{\kl[quotienting]{quotienting}\xspace}
\newrobustcmd{\quotientingintro}{quotienting\xspace}

\usepackage{amsthm}
\theoremstyle{plain}
\newtheorem{theorem}[thm]{Theorem}
\newtheorem{megadef}[thm]{Definition}
\newtheorem{corollary}[thm]{Corollary}
\newtheorem{proposition}[thm]{Proposition}
\newtheorem{lemma}[thm]{Lemma}
\newtheorem{fct}[thm]{Fact}
\newtheorem{example}[thm]{Example}
\newtheorem{remark}[thm]{Remark}
\newtheorem*{claim}{Claim}

\usepackage[notion,quotation,scope,diagnose line=true,paper]{knowledge}

\IfKnowledgePaperModeTF
{  \knowledgestyle{notation}{color=black}
  \knowledgestyle{intro notation}{color=black}
  \knowledgestyle{definition}{color=black}
  \knowledgestyle{intro definition}{color=black,emphasize}
}{  \knowledgestyle{notation}{color=black}
  \knowledgestyle{intro notation}{color=black}
  \knowledgestyle{definition}{color=black}
  \knowledgestyle{intro definition}{color=black,emphasize}
}

\knowledge{fullcont}{notion,style=notation}
\knowledge{cont}{notion,style=notation}
\knowledge{\at}{notion,style=notation}
\knowledge{\prin{\rho}{\Kb}}[\prin{\rho}{\Kb'}]{notion,style=notation}
\knowledge{\itriv{L,\rho}}{notion,style=notation}
\knowledge{\rho_\Lb}[\rho_{\Lb}]{notion,style=notation}
\knowledge{\opti{\Cs}{L,\rho}}{notion,style=notation}
\knowledge{\opti{\Cs}{\rho}}[\copti{\rho}]{notion,style=notation}

\knowledge{hemiring}[hemirings|Hemirings]{notion,style=definition,intro style=intro definition}
\knowledge{semiring}[semirings|Semiring|Semirings]{notion,style=definition}
\knowledge{separating}[Separating]{notion,style=definition}
\knowledge{cover}[covers|Cover|Covers]{notion,style=definition}
\knowledge{\Cs-cover}[\Cs-covers|\Cs-Cover|\Cs-Covers|\sicu-cover|\sicu-covers|\fod-cover|\fod-covers]{notion,style=definition}
\knowledge{Membership problem}[membership|Membership|\Cs-membership|\Cs-membership problem|\Cs-membership problems|membership problem|membership~problem|membership problems|Membership problems|membership algorithm|membership algorithms]{notion,style=definition}
\knowledge{Separation problem}[separation|Separation|\Cs-separation|\Cs-separation problem|\Cs-separation problems|separation problem|separation~problem|separation problems|Separation problems]{notion,style=definition}
\knowledge{Covering problem}[covering|\emph{covering}|Covering|\Cs-covering|\Cs-covering problem|\fod-covering|\fod-covering problem|\sicu-covering|\sicu-covering problem|\Cs-covering problems|covering problem|covering problems|Covering problems]{notion,style=definition}
\knowledge{Universal cover}[universal cover|universal \Cs-cover|Universal covers|universal covers|universal \Cs-covers]{notion,style=definition}
\knowledge{Universal covering problem}[universal covering|Universal covering|universal \Cs-covering|universal \Cs-covering problem|Universal \Cs-covering|\emph{universal} \Cs-covering problem|\emph{universal} \Cs-covering|universal \bscu-covering]{notion,style=definition}
\knowledge{rating set}[rating sets|Rating sets]{notion,style=definition}
\knowledge{\ratm}[\ratms|\Ratm|\Ratms]{notion,style=definition}
\knowledge{\mratm}[\mratms|\Mratm|\Mratms|\tame]{notion,style=definition}
\knowledge{\imprint}[\Imprint|\Imprints|\imprints|$\rho$-\imprint|$\rho$-\imprints]{notion,style=definition}
\knowledge{trivial $\rho$-\imprint}[trivial $\rho$-\imprint on $L$]{notion,style=definition}
\knowledge{\nice}[\Nice]{notion,style=definition}
\knowledge{optimal}[Optimal|optimality]{notion,style=definition}
\knowledge{first-order}[first-order logic|First-order logic|\fo]{notion,style=definition}
\knowledge{quotienting}[Quotienting|\quotienting|\quotientingintro]{notion,style=definition}
\knowledge{lattice}[Lattice|lattices]{notion,style=definition}
\knowledge{Boolean algebra}[Boolean algebras]{notion,style=definition}
\knowledge{canonical preorder}[Canonical preorder|canonical preorders|Canonical preorders]{notion,style=definition}
\knowledge{stratification}[Stratification|stratifications|Stratifications]{notion,style=definition}
\knowledge{stratum}[strata]{notion,style=definition}

\begin{document}

\ExplSyntaxOn
\bool_new:N\kl_cs_bool
\ExplSyntaxOff

\title{The Covering Problem}

\author{Thomas Place and Marc Zeitoun}

\address{LaBRI, University of Bordeaux, France}

\email{firstname.lastname@labri.fr}
\thanks{Supported by the DeLTA project (ANR-16-CE40-0007).}



\begin{abstract}
  An important endeavor in computer science is to precisely understand the expressive power of descriptive formalisms (such as fragments of first-order logic) over discrete structures (such as finite words, trees or graphs). Of course, the term ``understanding'' is not a precise mathematical notion. Therefore, carrying out this investigation requires a concrete objective to capture this understanding. In the literature, the standard choice for this objective is the \emph{"membership problem"}, whose aim is to find a procedure deciding whether an input language can be defined in the formalism under investigation. This approach was cemented as the ``right'' one by the seminal work of Schützenberger, McNaughton and Papert on ``first-order logic over finite words'', and has been in use since then.

  Unfortunately, "membership" questions are hard: for several fundamental formalisms, researchers have failed in this endeavor despite decades of investigation. In view of recent results on one of the most famous open questions, namely membership for all levels in the quantifier alternation hierarchy of "first-order logic", an explanation may be that "membership" is too restrictive as a setting. Indeed, these new results were obtained by considering \emph{more general} problems than "membership", taking advantage of the increased flexibility of the enriched mathematical setting. Investigating such new problems opened a promising research avenue, which permitted to solve "membership" for natural fragments of first-order logic. However, many of these problems are \emph{ad hoc}: for each fragment, the solution relies on a specific one. A unique new problem replacing "membership" as the right one is still missing.

  The main contribution of this paper is a suitable candidate to play this role: the \hbox{"covering problem"}. We motivate this problem with several arguments. First, it admits an elementary set theoretic formulation, similar to "membership". Second, we are able to reexplain or generalize all known results with this problem. Third, we develop a mathematical framework as well as a methodology tailored to the investigation of this problem. At last, for each class admitting a decidable "membership", we are able to instantiate our methodology to solve this more general problem. In particular, this yields \emph{constructive} solutions to "membership". We illustrate our approach with algorithms solving "covering" (hence also "membership" and generalizations thereof, such as the problem called "separation") for some classical fragments of first-order logic: level~1 in the quantifier alternation hierarchy of first-order logic, which consists of the so called \emph{piecewise testable} languages, the well-known 2-variable fragment of first-order logic, and level~$\frac12$ in the quantifier alternation hierarchy of first-order logic.
\end{abstract}

\maketitle
\section{Introduction}
One of the most successful applications of the notion of regularity in computer science is the investigation of logics on discrete structures, such as words or trees. The story began in the 60s when Büchi~\cite{BuchiMSO}, Elgot~\cite{ElgotMSO} and Trakhtenbrot~\cite{TrakhMSO} proved that the \emph{regular languages} of finite words are exactly those that can be defined in monadic second order logic (\mso). This was later pushed to infinite words~\cite{BuchiMSOInf}, to finite or infinite trees~\cite{Thatcher&Wright:1968,rabin1969}, to labeled countable linear orders~\cite{ccp11} and even to graphs of bounded tree width~\cite{buchigraphs} or linear cliquewidth~\cite{def_cliquewidth18}. Such connections not only testify to the robustness of the notion of regularity. Indeed, in the context of finite words, the connection was further exploited to investigate the expressive power of important \emph{fragments} of \mso, by relying on a decision problem associated to any such fragment: the "membership problem". This problem just asks whether the fragment under investigation forms a recursive class, \emph{i.e.}, its statement is as follows: given a regular language as input, decide whether it is definable by a sentence of the fragment.

\medskip
Obtaining "membership algorithms" is difficult. An oft-told and still open example is to decide the most natural fragment of \mso, namely "first-order logic"~("\fo"), on finite binary trees. On finite words (a much simpler structure than binary trees), Schützenberger, McNaughton and Papert~\cite{sfo,mnpfo} settled this question in the 70s. Their result was highly influential: it was often revisited~\cite{wfo,DGfo,Colcombet11,pingoodref}, and it paved the way to a series of results of the same nature. A famous example is Simon's Theorem~\cite{simonthm}, which yields an algorithm for the first level of the quantifier alternation hierarchy of "\fo". Other prominent examples are fragments of "\fo" where the linear order on positions is replaced by the successor relation~\cite{bslt73,mcnlt74,zalclt72,twltt85} or consider the two-variable fragment of "first-order logic"~\cite{twfodeux}. The relevance of this approach is nowadays validated by a wealth of results.

\medskip
The reason for this success is twofold. First, these results cemented "membership" as the ``right'' question: a solution conveys a deep intuition on the investigated logic. In particular, most results include a \emph{generic method for building a canonical sentence} witnessing "membership" of an input language if it is expressible in the logic. Second, Schützenberger's solution established a suitable framework and a methodology for solving "membership problems". This methodology is based on a canonical, finite and computable algebraic abstraction of a regular language: the \emph{syntactic monoid}. The core of the approach is to translate the semantic question (\emph{is the language definable in the fragment?}) into a purely syntactical, easy question to be tested on the syntactic monoid (\emph{does the syntactic monoid satisfy some~equation?}).

\medskip
Unfortunately, this methodology seems to have reached its limits for the hardest questions. An emblematic example is the \emph{quantifier alternation hierarchy of "first-order logic"}, which classifies sentences according to the number of alternations between $\exists$ and $\forall$ quantifiers in their prenex normal form. A sentence is \sici if its prenex normal form has $(i-1)$ alternations and starts with a block of existential quantifiers. A sentence is \bsci if it is a Boolean combination of \sici sentences. Obtaining "membership algorithms" for all levels in this hierarchy is a major open question, which has received a lot of attention (see~\cite{Weil-ConcatSurvey1989,Thomas:Languages-automata-logic:1997:a,Pin_1995,Pin_1997,pinbridges,Pin-ThemeVar2011,PZ:Siglog15,Pin:WSPC16} for details and a complete bibliography). However, progress on this question has been slow: until recently, only the lowest levels were solved, namely \sicu~\cite{arfi87,pwdelta}, \bscu~\cite{simonthm} and \sicd~\cite{arfi87,pwdelta}.

\medskip
It took years to solve higher levels. Recently, "membership algorithms" were obtained for the levels \sict~\cite{pzqalt,pzqaltj}, \bscd~\cite{pzqalt,pzqaltj} and \sic{4}~\cite{pseps3,pseps3j}. This was achieved by introducing new ingredients into Schützenberger's methodology: problems that are \emph{more general than "membership"}. For each of these results, the strategy is the same:
\begin{itemize}
\item First, a well-chosen \emph{more general} problem is solved for a \emph{lower} level in the hierarchy.
\item Then, this knowledge is \emph{turned} into a "membership algorithm" for the level under investigation.
\end{itemize}
Let us illustrate what we mean by ``more general problem'' by presenting the simplest of them: the \emph{"\Cs-separation problem"} (where $\Cs$ is a class of regular languages). This problem takes \emph{two} languages as input---rather than just one for "membership"---and asks whether there exists a third one which:
\begin{itemize}
\item belongs to $\Cs$,
\item contains the first language, and
\item is disjoint from the second.
\end{itemize}
It is easy to see why this problem generalizes "membership": a language belongs to a class~$\Cs$ if and only if it is $\Cs$-separable from its complement.  Being more general, such problems are also more difficult than "membership". However, this generality also makes them more rewarding in the insight they provide on the investigated logic. This motivated a series of papers on the "separation problem"~\cite{pzfo,cmmptsep,pvzmfcs13,pvzltt,pzsucc}, which culminated in the three results above~\cite{pzqalt,pzqaltj,pseps3,pseps3j}. However, while this avenue of research is very promising, it presently suffers \textbf{three major~flaws}:

\begin{enumerate}
\item The problems considered up until now form a jungle: each particular result relies actually not on "separation" itself, but rather on a specific \emph{ad hoc} generalization of this problem. As an illustration, the results of~\cite{pzfo,pzqalt,pzqaltj,pseps3,pseps3j} rely on \emph{three} different such problems. \item Among the problems that were investigated, "separation" is the only one that admits a simple and generic set-theoretic definition (which is why it is favored as an example). On the other hand, for all other problems, the definition requires to introduce additional concepts, such as semigroups and \efgame games.
\item In contrast to "membership" solutions, the solutions that have been obtained for these more general problems are \emph{non-constructive}. For example, most of the solutions for "separation" do not include a generic method for building a separator language, when it exists, because the algorithms are designed around the idea of witnessing that the two inputs are \emph{not} separable.
\end{enumerate}

\bigskip

\noindent {\bf Contributions.} Our objective in this paper is to address each of these three issues. Our first contribution is the definition of a \emph{single general} problem, the ``\emph{"covering problem"}'', which
\begin{enumerate}[label=$\alph*)$]
\item\label{it:covergen} encompasses \emph{all variations} of  the "separation problem" introduced so far to solve "membership",
\item\label{it:coversimple} enjoys a simple, language-theoretic formulation (just as "membership" and "separation").
\end{enumerate}
This already addresses the first two issues. The second contribution is a \emph{framework} and a \emph{methodology} for solving this new problem, which were lacking for the "separation problem". Finally, we illustrate this methodology with the presentation of algorithms solving the "covering problem" for several important fragments of "first-order logic". Naturally, these algorithms are based on the general methodology developed as the second contribution, and they yield constructive solutions for the "separation problem" as a byproduct, which addresses the third issue.

\medskip
Let us review these contributions in more details. As explained, the first one is to define a new problem, which we call \emph{"covering"}, satisfying Items~\ref{it:covergen} and \ref{it:coversimple} above, in order to gain afterwards a methodology for solving "separation" in a constructive~way.

\medskip\noindent\emph{\bfseries First step: Extending "separation" to inputs that are sets.}
We start with a simple observation: we already have in hand \emph{two orthogonal} generalizations of "membership". The first is "separation", that we aim at extending even further. The second is the straightforward but powerful generalization introduced by Schützenberger, with precisely a similar motivation as ours: setting up a methodology for solving "membership". In order to define "covering", a natural move is to combine both generalizations.

Before proceeding, let us recall Schützenberger's key idea: for testing whether a language~$L$ belongs to a fragment, one should not consider $L$ \emph{alone}. Instead, one should test whether \emph{all} languages recognized by its syntactic monoid belong to the fragment. This seemingly more demanding problem is in fact equivalent to "membership" when the fragment enjoys some mild properties: if $L$ belongs to the fragment, then so does any language recognized by its syntactic monoid. The motivation and the payoff for considering such \emph{input sets} is that they have a nice algebraic structure, which can be leveraged to develop inductive arguments in order to successfully design "membership algorithms". While simple, this idea is the core of most classical "membership algorithms".

The definition of the \emph{"covering problem"} builds on this idea: it generalizes "separation" to an input that is a \emph{set} of languages rather than just a pair. Thus, "covering" is a (strict) generalization of "separation" to an arbitrary number of input languages.

\medskip\noindent\emph{\bfseries Second step: "Separation" as an approximation problem.} Carrying out this idea is not immediate, however: there is a discrepancy between the generalization for "membership" and that for "separation". Indeed, extending "membership" to a set of languages is obvious: simply solving "membership" for all languages of the set is enough for developing inductive arguments. A similar naive generalization for "separation" would be, given a finite set of languages, to test whether each pair of languages from the set is separable by a language in the fragment.  Unfortunately, this turns out to be inadequate, because answering separation for all pairs of languages from a set is too weak to provide enough information.

A solution is to think of "separation" as an (over-)approximation problem. Given two input languages $L_1$ and $L_2$, it asks for a ``good approximation'' of $L_1$ (definable in the fragment under investigation) while $L_2$ serves as a quality measure---good approximations are those which do not intersect it. This point of view is amenable to generalization: in the \emph{"covering problem"}, our inputs are pairs $(L_1,\Lb_2)$ where $L_1$ is a single language and $\Lb_2$ is a finite \emph{set} of languages. The objective is still to approximate $L_1$ while $\Lb_2$ specifies what are the good approximations. Specifically, "covering" asks for a finite set of languages \Kb (all belonging to the fragment under investigation) such that,
\begin{enumerate}
\item The union of all languages in \Kb includes $L_1$: \Kb is a cover of $L_1$ (hence the name ``"covering"'').
\item No language in \Kb intersects all languages in $\Lb_2$.
\end{enumerate}
In particular, the original separation problem is the special case of "covering" when the input set $\Lb_2$ is a singleton.

\medskip\noindent\emph{\bfseries Third step: Abstracting the quality measure.} In the covering problem, the input is made of two objects playing different roles: we have a language $L_1$ that needs to be covered and a set of languages~$\Lb_2$ that serves as a quality measure specifying suitable covers. It is cleaner to separate\footnote{No pun intended.} these roles. For this reason, we define \emph{"\ratms"}, whose purpose is exclusively to evaluate the quality of a cover. This has two advantages: first, this makes it easier to pinpoint the hypotheses that we need on the set of languages and on the "\ratm". Second, it simplifies the~notation.

Our algorithms apply to inputs and "\ratms" satisfying some mild assumptions. We will show that one can always effectively reduce any input to such a special one.

\medskip\noindent\emph{\bfseries Benefits of the "covering problem".}
The main advantage of the covering problem is that it comes with a generic framework and a generic methodology designed for solving it.  This framework is our second contribution. It generalizes the original framework of Schützenberger for "membership" in a natural way and lifts all its benefits to a more general setting. In particular, we recover \emph{constructiveness}: a solution to the covering problem associated to a particular fragment yields a generic way for building an actual "optimal" "cover" of the input set. Furthermore, its definition is modular: the covering problem is designed so~that it can easily be generalized to accommodate future needs: while we use in this paper specific \ratms, the definition allows much more freedom (see~\cite{pseps3,pseps3j}).

Finally, the relevance of our new framework is supported by the fact that we are able to obtain covering algorithms for the fragments that were already known to enjoy a decidable "separation problem". In contrast to the previous algorithms, these more general ones are presented within a single unified framework. This is our third contribution. We present actual covering algorithms for five particular logics: "first-order logic" ("\fo"), two-variable "\fo" (\fod) and three logics within the quantifier alternation hierarchy of "\fo" (\sicu, \bscu and~\sicd). We also illustrate our proof techniques for three of these cases, $\sicu$, $\bscu$ and $\fod$.  As explained, the payoff is that we obtain \emph{effective} solutions to the covering problem. Hence, we obtain an effective method for building separators for the weaker "separation~problem".

\medskip
\noindent {\bf Organization.} We define the covering problem in Section~\ref{sec:covering}. We then devote Sections~\ref{sec:opti} and~\ref{sec:bgen:tame} to the presentation of our general framework designed for tackling the covering problem. In Section~\ref{sec:genba}, we then use them to design a general approach for handling covering in the restricted case of classes that are "Boolean algebras". We illustrate this approach with two detailed examples. In Section~\ref{sec:bsigma}, we investigate the fragment \bscu in the quantifier alternation hierarchy of "first-order logic". Then, in Section~\ref{sec:fo2}, we consider two-variable "first-order logic": \fod. Finally, we generalize our approach to handle covering for any lattice in Section~\ref{sec:genlatts}. We illustrate this generalized approach with a simple example in Section~\ref{sec:sigma1}: the fragment \sicu in the quantifier alternation hierarchy of "first-order logic".

\smallskip
\noindent  This paper is the full version of \cite{pzcovering}.

\section{Preliminary definitions}
\label{sec:prelims}
In this section, we present the standard terminology needed to formulate our results. Specifically, we define of classes of languages and their properties. Moreover, we introduce the standard "membership" and "separation" problems  (which we shall generalize with the "covering problem" in the next section).

\subsection{Finite words and classes of languages}
Throughout the whole paper, we fix a finite alphabet~$A$ and work with words over $A$. We denote by $A^*$ the set of all finite words over $A$. We let $\varepsilon$ be the empty word, and $A^{+}$ be the set $A^{*}\setminus\{\varepsilon\}$ of all nonempty words over~$A$.

Given a word $w \in A^*$, we denote by\AP\phantomintro{cont} $\cont{w}$ the set of letters appearing in $w$, that is, the least set $B\subseteq A$ such that $w\in B^*$. We say that $\cont{w}$ is the \emph{alphabet} of $w$. Finally, for $B\subseteq A$, we write\AP\phantomintro{fullcont}~$\fullcont{B}$ for the set of words whose alphabet is \emph{exactly} $B$, that is,
\[
  \fullcont{B} = \{w \in A^* \mid \cont{w} = B\}.
\]
Observe  that $\fullcont{B}\subseteq B^*$, and that $\fullcont{B}\subsetneq B^*$ when $B\neq\emptyset$.

A \emph{language (over $A$)} is a subset of $A^*$. Furthermore, a \emph{class of languages} \Cs is simply a set of languages over $A$.

\begin{remark}
  When it is important to consider several alphabets, a class of languages is usually defined as a function that maps a finite alphabet $A$ to a set of languages $\Cs(A)$ over~$A$. However, we adopt a simpler terminology in this paper, since we do not need to deal with several alphabets.
\end{remark}

All classes that we consider in the paper satisfy robust properties. We present them now. We say that a class \Cs of languages is \AP a ""lattice"" if it contains $\emptyset$ and $A^*$ and is closed under union and intersection. \AP A ""Boolean algebra"" is a "lattice" that is additionally closed under complement. Finally, given a language $L\subseteq A^*$ and a word $u\in A^*$, the left quotient $u^{-1}L$ of $L$ by $u$ is the language
\[
  u^{-1}L \stackrel{\text{def}}{=}\{w\in A^*\mid uw\in L\}.
\]
The right quotient $Lu^{-1}$ of $L$ by $u$ is defined symmetrically. A class \Cs is \AP ""\quotientingintro"" when it is closed under taking (left and right) quotients by words of~$A^*$. In the paper, all classes that we consider are at least "lattices".

\begin{example}\label{ex:at}
  Let   \at be the class of languages consisting of all Boolean combinations of languages~$B^*$, for some sub-alphabet $B \subseteq A$. Here, ``\at'' stands for ``alphabet testable'': a language is in \at when membership of a word in this language depends only on the set of letters occurring in the word. It is straightforward to verify that \at is a {\bf finite} \vari, which will serve as an important example in the paper.
\end{example}

Furthermore, we are only interested in regular languages, \emph{i.e.}, the classes that we consider in the paper contain regular languages only. These are the languages that can be equivalently defined by nondeterministic finite automata, finite monoids or monadic second-order logic. In the paper, we shall use the definitions based on automata and monoids, which we briefly recall below.

\medskip
\noindent
{\bf Automata.} A nondeterministic finite automaton (\nfa) is a tuple $\As = (Q,I,F,\delta)$, where $Q$ is a finite set of states, $I$ (resp. $F$) is the set of initial (resp.\ final) states, and $\delta \subseteq Q \times A \times Q$ is a set of transitions.  For such an \nfa and two states $q,r \in Q$, we shall write $L_{q,r}\stackrel{\text{def}}=\{w\in A^* \mid q \xrightarrow{w} r\}$ for the language of words labeling a run from state $q$ to state $r$. It is well-know that a language $L$ is regular when it is recognized by some \nfa \As, \emph{i.e.}, $L$ is the union of all languages $L_{q,r}$ with $q \in I$ and $r \in F$.

\medskip
\noindent
{\bf Semigroups and monoids.} A \emph{semigroup} is a set $S$ endowed with a binary associative operation $(s,t)\mapsto s\cdot t$. We also write $st$ instead of $s\cdot t$. An idempotent of a semigroup $S$ is an element $e \in S$ such that $ee = e$. It is folklore that for any \emph{finite} semigroup $S$, there exists a natural number $\omega(S)$ (denoted by $\omega$ when $S$ is understood from the context) such that for any $s \in S$, the element $s^\omega$ is idempotent.

A \emph{monoid} is a semigroup having a neutral element $1_S$, \emph{i.e.}, such that $1_S\cdot s=s\cdot 1_S=s$ for every element $s$ of the monoid. In particular, $A^{+}$ is a semigroup (the binary operation is the concatenation of words) and $A^{*}$ is a monoid, with $\varepsilon$ as the neutral element. An \emph{ordered monoid} is a monoid $M$ together with an order relation $\leq$ on $M$ which is compatible with the multiplication of $M$, that is, such that $s\leq s'$ and $t\leq t'$ imply $ss'\leq tt'$.

A \emph{morphism} between two monoids $M,M'$ is a map $\alpha: M \to M'$ such that $\alpha(1_M) = 1_{M'}$ and for all $s,t\in M$, we have $\alpha(st)=\alpha(s)\alpha(t)$. It is well known that a language $L$ is regular if and only if there exists a morphism from $A^{*}$ into a finite monoid such that membership of any word in~$L$ is determined by its image under this morphism.

\subsection{The \kl{membership} and \kl{separation problems}}
As announced above, in the paper, we only work with classes of regular languages. Usually, such a class \Cs is associated to a syntax: the languages in \Cs are those which can be described by at least one representation in this syntax (see the example of "first-order logic" below). When we have such a class of languages in hand, the most basic question is whether, for a regular language given as input, one can test membership of the language in the class \Cs. In other words, we want to determine whether there exists an algorithm that decides when this input language admits a description in the given syntax. The corresponding decision problem is called \emph{"\Cs-membership"} (or "membership" for \Cs).

\begin{megadef}[\AP""Membership problem"" for \Cs]\leavevmode\\
  \begin{tabular}{ll}
    {\bf Input:}    & A regular language $L$.\\
    {\bf Question:} & Does $L$ belong to \Cs?
  \end{tabular}
\end{megadef}

Recent solutions to the "membership" problem actually consider a more general problem, the \emph{"\Cs-separation problem"} (or "separation" problem for \Cs). This is the following decision problem:
\begin{megadef}[\AP""Separation problem"" for \Cs]$\quad$\\
  \begin{tabular}{ll}
    {\bf Input:}    & Two regular languages $L_1$ and $L_2$.\\
    {\bf Question:} & Does there exist a language $K$ from \Cs such that $L_1\subseteq K$ and $K\cap L_2=\emptyset$?
  \end{tabular}
\end{megadef}
We say that a language $K$ such that $L_1\subseteq K$ and $K\cap L_2=\emptyset$ is a \emph{separator} of $(L_1,L_2)$. Observe that since regular languages are closed under complement, there is a straightforward reduction from "membership" to "separation". Indeed, an input language $L$ belongs to \Cs when it can be \Cs-separated from its complement.

As we explained in the introduction, we shall not work directly with these two problems in the paper. Instead, we consider the more general \emph{"covering problem"}, which we define in the next section.

\AP
\subsection{\emph{\reintro{First-order logic}} and quantifier alternation} Most examples of classes that we shall consider in the paper are taken from logic. Here, we briefly recall the definition of "first-order logic" over words and its quantifier alternation hierarchy.
\phantomintro{First-order logic}

One may view a finite word as a logical structure composed of a linearly ordered sequence of positions labeled over $A$. In "first-order logic" ("\fo"), one may use the following predicates:

\begin{enumerate}
\item For each letter $a \in A$, a unary predicate $P_a$ which selects positions labeled with an ``$a$'',
\item A binary predicate ``$<$'' for the (strict) linear order between the positions.
\end{enumerate}
A language $L$ is said to be \emph{"first-order" definable} when there exists an "\fo" sentence $\varphi$ such that $L = \{w  \mid w \models \varphi\}$. One also denotes by "\fo" the class of all "first-order" definable languages. It is folklore that "\fo" is a \vari.

\medskip

We shall also consider the quantifier alternation hierarchy of "\fo". It is natural to classify "first-order" sentences by counting their number of quantifier alternations. Let $n \in \nat$. We say that an "\fo" sentence is \sic{n} (resp. \pic{n}) when its prenex normal form has either,
\begin{itemize}
\item \emph{exactly} $n -1$ quantifier alternations (\emph{i.e.}, exactly $n$ blocks of quantifiers) starting with an $\exists$ (resp.\ $\forall$), or
\item \emph{strictly less} than $n -1$ quantifier alternations (\emph{i.e.}, strictly less than $n$ blocks of	quantifiers).
\end{itemize}
For example, a formula whose prenex normal form is
\[
  \forall x_1 \exists x_2 \forall x_3 \forall x_4
  \ \varphi(x_1,x_2,x_3,x_4) \quad \text{(with $\varphi$ quantifier-free)}
\]
\noindent
is \pic{3}. In general, the negation of a \sic{n} sentence is not a \sic{n} sentence (it is a \pic{n} sentence). Hence it is relevant to define \bsc{n} sentences as the Boolean combinations of \sic{n} sentences. As for "\fo", we use \sic{n}, \pic{n} and \bsc{n} to denote the corresponding classes of languages. This yields an infinite hierarchy of classes of languages, as presented in \figurename~\ref{fig:hiera}.

\tikzstyle{non}=[inner sep=1pt]
\tikzstyle{tag}=[draw,fill=white,sloped,circle,inner sep=1pt]
\begin{figure}[!htb]
  \centering
  \begin{tikzpicture}

    \node[non] (s1) at (1.0,-0.8) {\sicu};
    \node[non] (p1) at (1.0,0.8) {\picu};
    \node[non] (b1) at (2.5,0.0) {\bscu};

    \node[non] (s2) at ($(b1)+(1.5,-0.8)$) {\sicd};
    \node[non] (p2) at ($(b1)+(1.5,0.8)$) {\picd};
    \node[non] (b2) at ($(b1)+(3.0,0.0)$) {\bscd};

    \node[non] (s3) at ($(b2)+(1.5,-0.8)$) {\sict};
    \node[non] (p3) at ($(b2)+(1.5,0.8)$) {\pict};
    \node[non] (b3) at ($(b2)+(3.0,0.0)$) {\bsct};

    \node[non] (s4) at ($(b3)+(1.5,-0.8)$) {\sic{4}};
    \node[non] (p4) at ($(b3)+(1.5,0.8)$) {\pic{4}};

    \draw[thick] (s1.east) to [out=0,in=-90] node[tag] {\scriptsize
      $\subsetneq$} (b1.-120);
    \draw[thick] (p1.east) to [out=0,in=90] node[tag] {\scriptsize
      $\subsetneq$} (b1.120);

    \draw[thick] (b1.-60) to [out=-90,in=180] node[tag] {\scriptsize
      $\subsetneq$} (s2.west);
    \draw[thick] (b1.60) to [out=90,in=-180] node[tag] {\scriptsize
      $\subsetneq$} (p2.west);
    \draw[thick] (s2.east) to [out=0,in=-90] node[tag] {\scriptsize
      $\subsetneq$} (b2.-120);
    \draw[thick] (p2.east) to [out=0,in=90] node[tag] {\scriptsize
      $\subsetneq$} (b2.120);

    \draw[thick] (b2.-60) to [out=-90,in=180] node[tag] {\scriptsize
      $\subsetneq$} (s3.west);
    \draw[thick] (b2.60) to [out=90,in=-180] node[tag] {\scriptsize
      $\subsetneq$} (p3.west);
    \draw[thick] (s3.east) to [out=0,in=-90] node[tag] {\scriptsize
      $\subsetneq$} (b3.-120);
    \draw[thick] (p3.east) to [out=0,in=90] node[tag] {\scriptsize
      $\subsetneq$} (b3.120);

    \draw[thick] (b3.-60) to [out=-90,in=180] node[tag] {\scriptsize
      $\subsetneq$} (s4.west);
    \draw[thick] (b3.60) to [out=90,in=-180] node[tag] {\scriptsize
      $\subsetneq$} (p4.west);

    \draw[thick,dotted] ($(s4.east)+(0.1,0.0)$) to
    ($(s4.east)+(0.6,0.0)$);
    \draw[thick,dotted] ($(p4.east)+(0.1,0.0)$) to
    ($(p4.east)+(0.6,0.0)$);

  \end{tikzpicture}
  \caption{Quantifier Alternation Hierarchy}
  \label{fig:hiera}
\end{figure}
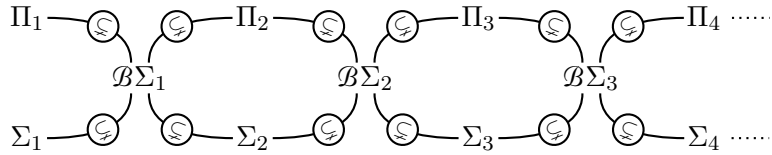

It is folklore that all classes \sic{n} and \pic{n} in the hierarchy are \pvaris (but not "Boolean algebras") and that all classes \bsc{n} are \varis.

\section{The \kl{covering problem}}
\label{sec:covering}
In this section, we first define the "covering problem" and establish the connection with separation. Next, we outline the steps that we shall take in the next sections for devising a general approach to this problem.

\subsection{Preliminary definitions} Unlike the "membership problem" but as the "separation problem", the "covering problem" takes \textbf{two} different objects as input. The first one is a single language $L \subseteq A^*$. The second one is a \emph{finite multiset of languages} $\Lb = \{L_1,\dots,L_n\}$. Note that we speak of multisets here for the sake   of allowing several copies of the same language in \Lb.

\begin{remark}
  Using multisets of languages is not mandatory but it is natural. Indeed, each language is given by a recognizer (typically, an \nfa or a monoid morphism). Since two distinct recognizers may define the same language, our input is indeed a multiset of languages. Another important point is that considering multisets is harmless. If\/ $\Lb_1$ and $\Lb_2$ are distinct multisets for the same underlying set of languages, then the "covering problems" for instances $\Lb_1$ and $\Lb_2$ will be equivalent.
\end{remark}

Consider some class \Cs. Given an input language $L$ and an input finite multiset of languages \Lb, the "\Cs-covering problem" asks whether there exists a \emph{"\Cs-cover" of $L$ which is "separating" for \Lb}. Let us first define what these notions mean.

\medskip
\noindent
{\bf "Covers".} Consider some language $L \subseteq A^*$. A \AP""cover"" of $L$ is just a \textbf{finite} set of languages \Kb such that:
\[
  L \subseteq \bigcup_{K \in \Kb} K.
\]
We shall often look for "covers" of the universal language $A^*$. Indeed, this special case suffices when the investigated class \Cs is a Boolean algebra (we discuss this point in Section~\ref{sec:genba}). Such a "cover" will be called a \AP""universal cover"".

\medskip
\noindent
{\bf "Separating" "covers".} Consider a finite multiset of languages \Lb and a set \Kb of languages. We say that~\Kb is \AP""separating"" for \Lb when the following property holds:
\[
  \text{For all $K \in \Kb$, there exists $L \in \Lb$ such that $K \cap L' = \emptyset$}.
\]
In other words, \Kb is "separating" for \Lb when no $K \in \Kb$ intersects each of the languages in \Lb. Note that while this definition makes sense for any set of languages \Kb, we are mainly interested in the case when \Kb is a "cover" of some other language $L$. We illustrate this definition in \figurename~\ref{fig:lattice:separatingcover}.

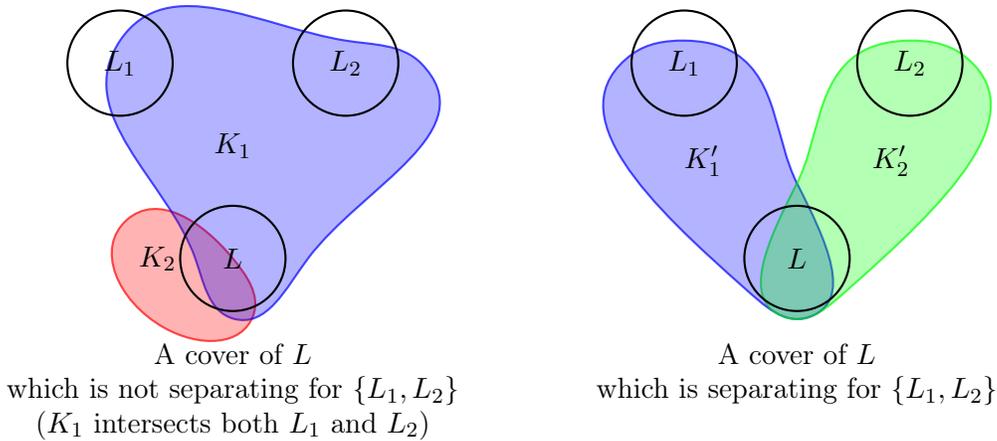
\begin{figure}[!ht]
  \begin{center}
    \begin{tikzpicture}
      \draw[thick] (0,0) circle (0.7cm) node (l1) {$L_1$};
      \draw[thick] (3,0) circle (0.7cm) node (l2) {$L_2$};
      \draw[thick] ($(0,0)!1!-60:(3,0)$) circle (0.7cm) node (l3) {$L$};

      \coordinate (m1) at ($(l1)+(1.0,0.0)$);
      \coordinate (m2) at ($(l2)+(1.0,0.0)$);
      \coordinate (m3) at ($(l3)+(1.0,0.0)$);

      \begin{pgfonlayer}{background}
        \def\fstlang{($(l3)$) to[closed,curve
          through={
            ($(l3)!1.5!160:(m3)$)
            .. ($(l3)!1.3!225:(m3)$)
            .. ($(l3)!0.8!290:(m3)$)
          }]
          ($(l3)$)}

        \def\seclang{($(l3)!0.8!-90:(m3)$) to[closed,curve
          through={
            ($(l3)!0.5!180:(m3)$)
            .. ($(l1)!0.3!90:(m1)$)
            .. ($(l2)!0.3!90:(m2)$)
            .. ($(l2)!1!0:(m2)$)
            .. ($(l3)!1!0:(m3)$)
          }]
          ($(l3)!0.8!-90:(m3)$)}

        \begin{scope}[draw opacity=0.7,fill opacity=0.3]
          \draw[thick,red] \fstlang;
          \draw[thick,blue] \seclang;
          \fill[red] \fstlang;
          \fill[blue] \seclang;
        \end{scope}

        \node at ($(l3)-(1.0,0.0)$) {$K_2$};
        \node at ($(l3)!1/2!30:(l2)$) {$K_1$};

        \node[align=center,anchor=north] at ($(l3)-(0,1.0)$) {A "cover" of $L$\\ which is not "separating" for $\{L_1,L_2\}$ \\($K_1$ intersects both $L_1$ and $L_2$)};

      \end{pgfonlayer}

      \begin{scope}[xshift=7.5cm]

        \draw[thick] (0,0) circle (0.7cm) node (l1) {$L_1$};
        \draw[thick] (3,0) circle (0.7cm) node (l2) {$L_2$};
        \draw[thick] ($(0,0)!1!-60:(3,0)$) circle (0.7cm) node (l3) {$L$};

        \coordinate (m1) at ($(l1)+(1.0,0.0)$);
        \coordinate (m2) at ($(l2)+(1.0,0.0)$);
        \coordinate (m3) at ($(l3)+(1.0,0.0)$);

        \begin{pgfonlayer}{background}
          \def\fstlang{($(l1)!1.3!160:(m1)$) to[closed,curve
            through={
              ($(l1)!1!90:(m1)$)
              .. ($(l2)!1!90:(m2)$)
              .. ($(l2)!1.3!20:(m2)$)
              .. ($(l2)!1/2!(l1)$)
            }]
            ($(l1)!1.3!160:(m1)$)}
          \def\seclang{($(l3)!0.8!280:(m3)$) to[closed,curve
            through={
              ($(l3)!0.8!210:(m3)$)
              .. ($(l1)!0.7!170:(m1)$)
              .. ($(l1)!0.3!90:(m1)$)
              .. ($(l1)!0.7!10:(m1)$)
              .. ($(l3)!1!90:(m3)$)
            }]
            ($(l3)!0.8!280:(m3)$)}

          \def\trdlang{($(l3)!0.8!-100:(m3)$) to[closed,curve
            through={
              ($(l3)!0.8!-30:(m3)$)
              .. ($(l2)!0.7!10:(m2)$)
              .. ($(l2)!0.3!90:(m2)$)
              .. ($(l2)!0.7!170:(m2)$)
              .. ($(l3)!1!90:(m3)$)
            }]
            ($(l3)!0.8!-100:(m3)$)}

          \begin{scope}[draw opacity=0.7,fill opacity=0.3]
            \draw[thick,blue] \seclang;
            \draw[thick,green] \trdlang;
            \fill[blue] \seclang;
            \fill[green] \trdlang;
          \end{scope}

          \node at ($($(l3)!1/2!(l1)$)-(0.5,0.0)$) {$K'_1$};
          \node at ($($(l2)!1/2!(l3)$)+(0.5,0.0)$) {$K'_2$};
        \end{pgfonlayer}

        \node[align=center,anchor=north] at ($(l3)-(0,1.0)$) {A "cover" of $L$\\ which is "separating" for $\{L_1,L_2\}$};
      \end{scope}
    \end{tikzpicture}
  \end{center}
  \caption{Two "covers" of $L$. The right one is "separating" for $\{L_1,L_2\}$ and the left one is not}
  \label{fig:lattice:separatingcover}
\end{figure}

A simple observation is that for any language $L$ and any multiset of languages \Lb, there exists a "cover" of $L$ which is "separating" for \Lb if and only if the intersection between $L$ and all languages in \Lb is empty. This generalizes the fact that two languages are separable if and only if they are disjoint.

\begin{lemma} \label{lem:whenseparating}
  Let $L$ be a language and let \Lb be a finite multiset of languages. There exists a "cover" of~$L$ which is "separating" for \Lb if and only the following condition is satisfied:
  \[
    L \cap \bigcap_{L' \in \Lb} L' = \emptyset.
  \]
\end{lemma}

\begin{proof}
  Assume first that there exists a "cover" \Kb of $L$ which is "separating" for \Lb. Since \Kb is a "cover" of~$L$, we have $L \subseteq \bigcup_{K \in \Kb} K$,
  and therefore,
  \[
    L \cap \bigcap_{L' \in \Lb} L' \subseteq \bigcup_{K \in \Kb} \left(K \cap \bigcap_{L' \in \Lb} L'\right).
  \]
  Moreover, since \Kb is "separating", for any $K \in \Kb$, there exists $L' \in \Lb$ such that $K \cap L' = \emptyset$. It follows that $\bigcup_{K \in \Kb} \left(K \cap \bigcap_{L' \in \Lb} L'\right)= \emptyset$ and we conclude that $L \cap \bigcap_{L' \in \Lb} L' = \emptyset$.

  Conversely, assume that $L \cap {\bigcap_{L' \in \Lb} L'} = \emptyset$.  Consider the following equivalence relation defined over words of $L$: $u,v \in L$  are equivalent when $u \in L' \Leftrightarrow v \in L'$ for all $L' \in \Lb$. We let \Kb be the partition of $L$ induced by this equivalence. Clearly, \Kb is a "cover" of $L$. Moreover, one may verify that it is "separating" for \Lb since we have $L \cap \bigcap_{L' \in \Lb} L' = \emptyset$.
\end{proof}

\begin{remark}\label{rem:coveringeasier}
  For multisets \Lb whose size is at least $2$, finding a "cover" of $L$ which is "separating" for \Lb is less demanding than finding separators for all pairs of languages $(L,L')$ where $L' \in \Lb$. For example, consider the alphabet $A = \{a,b,c\}$ and let $L = a^+ + b^+$, $L_1 = b^+ + c^+$ and $L_2 = c^+ + a^+$. It is impossible to separate the pairs $(L,L_1)$ and $(L,L_2)$ (as they pairwise intersect). However, $\{a^*,b^*\}$ is a "cover" of $L$ which is "separating" for $\{L_{1},L_{2}\}$.
\end{remark}

Naturally, as for separation, the "covering problem" restricts the set of allowed "covers" with a predefined class \Cs: we look for "separating" "covers" which are made of languages belonging to \Cs.

\subsection{The problem}

We may now state the "covering problem" for regular languages. As for "separation" and "membership", it depends on a class \Cs of languages that restricts the set of possible "covers". Given a language $L$, a \AP""\Cs-cover"" of\/ $L$ is a "cover" \Kb of $L$ such that all languages $K \in \Kb$ belong to \Cs. Finally, if \Lb is a finite multiset of languages, we say that the pair $(L,\Lb)$ is \emph{\Cs-coverable} when there exists a "\Cs-cover" of $L$ which is "separating" for \Lb (when \Lb is clear from the context, we will simply say that such a "cover" is "separating"). The "covering problem" is a follows.

\begin{megadef}[\AP""Covering problem"" for \Cs]$\quad$\\
  \begin{tabular}{ll}
    {\bf Input:}    & A regular language $L$ and a finite multiset of regular languages $\Lb$.\\
    {\bf Question:} & Is $(L,\Lb)$ \Cs-coverable?
  \end{tabular}
\end{megadef}

There are two stages when solving the "covering problem" for a given class \Cs.
\begin{enumerate}
  \itemAP \emph{Stage One}: find an algorithm which \emph{decides} the "covering problem" for \Cs (such an algorithm is called a \emph{"covering" algorithm} for \Cs).
\item \emph{Stage Two}: find an algorithm that actually \emph{computes} representations for languages in a "separating" "\Cs-cover" when it exists (\emph{i.e.}, when the answer to the question of Stage~1 is ``yes'').
\end{enumerate}

Let us formally connect "\Cs-covering" with "\Cs-separation": it is more general (provided that the class \Cs is closed under union). Specifically, "separation" is the special case of "covering" when the multiset $\Lb$ is a singleton. While simple, this connection is important: many "separation" algorithms in the literature are actually based on "covering".

\begin{theorem} \label{thm:covsep}
  Let \Cs be a class closed under union and let $L_1,L_2$ be two languages. The following properties are equivalent:
  \begin{enumerate}[ref=\arabic*)]
  \item $L_1$ is \Cs-separable from $L_2$.
  \item $(L_1,\{L_2\})$ is \Cs-coverable.
  \end{enumerate}
\end{theorem}

\begin{proof}
  Assume first that $L_1$ is \Cs-separable from $L_2$. Then there exists $K \in \Cs$ such that $L_1 \subseteq K$ and $L_2 \cap K = \emptyset$. It follows that $\{K\}$ is a "separating" "\Cs-cover" of $(L_1,\{L_2\})$.
  
  Conversely, assume that $(L_1,\{L_2\})$ is \Cs-coverable and let \Kb be a "separating" "\Cs-cover". We let $K = \bigcup_{K' \in \Kb} K'$. Clearly, $K \in \Cs$ by closure under union. Since \Kb is a "cover" of $L_1$, we have $L_1 \subseteq K$ and since \Kb is "separating" for $\{L_2\}$, we have $K' \cap L_2 = \emptyset$ for all $K' \in \Kb$. Thus, $K \cap L_2 = \emptyset$, which means that $K \in \Cs$ separates $L_1$ from $L_2$.
\end{proof}

\subsection{A special case: universal covering}

Clearly, the definition of \Cs-covering makes sense for any class \Cs. However, it turns out that when \Cs is a Boolean algebra, it suffices to consider a special case which has one less parameter. Consequently, handling covering will be simpler for Boolean algebras.

We call this restriction \emph{universal covering}. This weaker problem corresponds to the special case of inputs $(L,\Lb)$ when the language $L$ that needs to be covered is the universal language $A^*$. More precisely, given a class \Cs restricting the set of possible covers, the \emph{universal covering problem} for \Cs is as follows.
\begin{megadef}[\AP""Universal covering problem"" for \Cs]$\quad$\\
  \begin{tabular}{ll}
    {\bf Input:}    & A finite multiset of regular languages \Lb.\\
    {\bf Question:} & Is $(A^*,\Lb)$ \Cs-coverable?
  \end{tabular}
\end{megadef}
Therefore, the problem asks whether there exists a \emph{universal cover} (\emph{i.e.}, a cover of $A^*$) which is separating for the input multiset \Lb. Intuitively, this restriction is simpler than the full problem since it takes only one object as input rather than two. For the sake of simplifying the presentation, when we consider universal covering, we shall often omit $A^*$ and say that \Lb is \emph{\Cs-coverable} to indicate that $(A^*,\Lb)$ is \Cs-coverable. 

As announced, when \Cs is a Boolean algebra, the full covering problem reduces to this special case. Consequently, when working with a Boolean algebra, one should not consider the full covering problem. Instead, working with universal covering (which has one less parameter) is simpler. We prove this in the following proposition.

\begin{proposition}
  \label{prop:bgen:eqpoint}
  Let $L$ be a language and $\Lb$ be a finite multiset of languages. Given any Boolean algebra \Cs, the two following properties are equivalent:
  \begin{enumerate}
  \item\label{it:eqpoint:1} $(L,\Lb)$ is \Cs-coverable.
  \item\label{it:eqpoint:2} $\{L\} \cup \Lb$ is \Cs-coverable.
  \end{enumerate}
\end{proposition}

\begin{proof}
  We first prove the implication $\ref{it:eqpoint:1}\Rightarrow\ref{it:eqpoint:2}$. Assume that $(L,\Lb)$ is \Cs-coverable and let \Kb be a "\Cs-cover" of $L$ which is "separating" for $\Lb$. Our goal is to find a "universal \Cs-cover", which is "separating" for $\{L\} \cup \Lb$. Let $K'$ be the following language:
  \[
    K' = A^* \setminus \left(\bigcup_{K \in \Kb} K \right).
  \]
  Note that since \Cs is a Boolean algebra, we have $K' \in \Cs$. Let $\Kb' = \Kb \cup \{K'\}$. Clearly, $\Kb'$ is a "universal \Cs-cover". We show that it is "separating" for $\{L\} \cup \Lb$ which concludes the proof for this direction. Given $K \in \Kb'$, either $K \in \Kb$, or $K=K'$. In the first case, we get $H \in \Lb$ such that $K \cap H = \emptyset$ since \Kb is "separating" for $\Lb$. Otherwise, when $K = K'$, we have $K'\cap L = \emptyset$ by definition of $K'$, since \Kb is a "cover" of $L$. This concludes the proof for this direction.

  Conversely, assume that $\{L\} \cup \Lb$ is \Cs-coverable and let \Kb be a "universal \Cs-cover" which is "separating" for $\{L\} \cup \Lb$. We define:
  \[
    \Kb' = \{K \in \Kb \mid K \cap L \neq \emptyset\}.
  \]
  We claim that $\Kb'$ is a "\Cs-cover" of $L$ which is "separating" for $\Lb$. Indeed, we know that $L \subseteq \bigcup_{K \in \Kb'} K$ since \Kb is a "cover" of $A^*$ and $\Kb'$ contains all languages in \Kb that intersect $L$. Moreover, since \Kb is "separating" for $\{L\} \cup \Lb$, we know that for any $K \in \Kb$, either $K \cap L = \emptyset$ or $L\cap H = \emptyset$ for some $H \in \Lb$. For the languages $K \in \Kb'$, we know by definition that $K \cap L \neq \emptyset$. Thus, there exists $H \in \Lb$ such that $K\cap H = \emptyset$. This concludes the proof.
\end{proof}

\subsection{A framework for the \kl{covering problem}}

Now that "covering" is defined, we need to explain the benefits of considering this problem rather than just separation. We do so by presenting a general framework whose purpose is to obtain covering algorithms for actual classes of languages. This approach is designed with both stages of the problem in mind: finding a decision algorithm and constructing "separating" "covers" when they exist. Consider some "lattice" \Cs. Our approach to "\Cs-covering" is obtained by combining three independent key ideas that we describe now.

\begin{enumerate}[ref=(\arabic*)]
\item\label{key:bgen:1} When trying to solve "\Cs-covering" for some input pair $(L,\Lb)$, we view the multiset \Lb as a \emph{quality measure}: it is used to evaluate the quality of "\Cs-covers" of $L$. In other words, we are browsing "\Cs-covers" of $L$ in search for one which is good enough with respect to \Lb. This point of view allows us to reformulate "\Cs-covering" as a \emph{\bfseries computational problem}, rather than a decision problem. One wants to build an object that always exists regardless of whether $(L,\Lb)$ is \Cs-coverable or not: a "\Cs-cover" of $L$ which is \emph{optimal for \Lb}. The main property of this object is that for any subset $\Hb$ of \Lb, this optimal \Cs-cover of $L$ is separating for $\Hb$ if and only if $(L,\Hb)$ is \Cs-coverable. Therefore, having it in hand is enough to solve "\Cs-covering" for all subsets of \Lb (including \Lb~itself).

\item\label{key:bgen:1bis} The second key idea is to generalize this computational problem to get a \emph{\bfseries generic} computational problem. This problem is parametrized by a new object that we name ``"\ratm"''. "\Ratms" are algebraic objects that one may use to measure the quality of an arbitrary \Cs-cover, and that abstract the multiset \Lb (which was also used in the "\Cs-covering problem" as a quality measure of a cover). The general problem asks to build a "\Cs-cover" of some language $L$ which is as good as possible with respect to a given \ratm $\rho$: a \emph{$\rho$-"optimal" "universal \Cs-cover"}. The approach described above for "\Cs-covering" with input $(L,\Lb)$ is just the instance of this abstract problem for a particular "\ratm" that one may build from \Lb. Generalizing the problem makes the presentation simpler, underlines the important hypotheses and yields elegant "covering"~algorithms.  

\item\label{key:bgen:2} The third key idea exploits the crucial fact that our inputs for the "\Cs-covering problem" are made of \emph{\bfseries regular} languages. In particular, this means that in the above computational problem, we may restrict ourselves to a class of "\ratms" having special properties. We call \emph{"\mratms"} these enhanced "\ratms". Our algorithms crucially exploit their properties.
\end{enumerate}

We detail theses two key ideas in Section~\ref{sec:opti} and~\ref{sec:bgen:tame}. First, we define \ratms and explain how they relate to the covering problem in Section~\ref{sec:opti}: this is our first two key ideas.  We then define the special class of \mratms in Section~\ref{sec:bgen:tame}: this our third key idea.

We then summarize the notions introduced in these two sections to outline our general methodology for tackling "\Cs-covering". We actually present two methodologies. The first one is designed to accommodate the restricted "universal \Cs-covering problem" which is simpler to handle (and is equivalent to full "covering" when \Cs is a "Boolean algebra" by Proposition~\ref{prop:bgen:eqpoint}). We present it in Section~\ref{sec:genba} and illustrate it with two examples: level \bscu of the quantifier alternation hierarchy of first-order logic in Section~\ref{sec:bsigma} and two-variable first-order logic in Section~\ref{sec:fo2}. Then, we present a generalized methodology designed for tackling the full "\Cs-covering problem" in Section~\ref{sec:genlatts}. We illustrate it with an example in Section~\ref{sec:sigma1}: level \sicu in the quantifier alternation hierarchy of first-order logic.

\begin{remark}
  Both methodologies apply to any class \Cs which is a \emph{\pvari of regular languages}. Actually, most notions involved in our framework make sense for any "lattice" \Cs. However, we need \Cs to be a \pvari of regular languages to use a crucial result (namely Lemma~\ref{lem:bgen:optsemi}).
\end{remark}

\begin{remark}
  It is important to keep in mind that the purpose of this methodology is to provide the right framework to tackle "\Cs-covering problems". On the other hand, they do {\bf not} yield "\Cs-covering" algorithms ``for free''. Once a \pvari \Cs is fixed, getting a "\Cs-covering" algorithm using our methodology still requires a lot of work \emph{specific to \Cs}. This is illustrated by the detailed examples that we present.
\end{remark}

\section{\Ratms and optimal covers}
\label{sec:opti}
This section details the main ingredient in our approach to "\Cs-covering". Given an input pair $(L,\Lb)$, we view the finite multiset \Lb as a quality measure for evaluating "\Cs-covers" of $L$. Our objective is to build such a "\Cs-cover" of $L$ which is ``"optimal" for this measure''. The main point here is that this object always exists (regardless of whether $(L,\Lb)$ is \Cs-coverable) and it is "separating" for any subset of \Lb which is \Cs-coverable.

We shall actually work within a more general framework and consider a generic computational problem. It asks to build a "\Cs-cover" of some input language $L$ that is "optimal" with respect to a parameter that we name a \emph{"\ratm"}. We shall then prove that for any finite multiset of languages \Lb, one may define a special "\ratm" $\rho_\Lb$ such that the approach outlined above for "\Cs-covering" with input \Lb corresponds to building a "\Cs-cover" which is "optimal" for this "\ratm".

\begin{remark}\label{rem:bgen:whyfilt}
  Considering this more general framework has two main benefits. First, working with abstract "\ratms" rather than  the specific ones $\rho_\Lb$ associated to multisets \Lb of languages simplifies the notation. Moreover, it yields more elegant presentations for "covering" algorithms.
\end{remark}

We first define "\ratms". Then, we explain how to use them for measuring the quality of an arbitrary "cover". Given a "\ratm" $\rho$ and a some finite set of languages \Kb, we define the \emph{"$\rho$-\imprint" of~\Kb} which corresponds to this measure. We then use this new notion to define what an "optimal" "\Cs-cover" of a language $L$ is for a given "\ratm" $\rho$. Finally, we connect these definitions with our original goal: solving "\Cs-covering".

\subsection{\kl{\Ratms}}\label{sec:ratms}

In order to define a "\ratm", we first need a \emph{"rating set"}. A \AP""rating set"" is simply a finite commutative and idempotent monoid $(R,+,0_R)$. Recall that being idempotent means that for all $r\in R$, we have $r+r=r$. The binary operation $+$ is called \emph{addition}\footnote{It is often the case to denote by `$+$' a monoid operation when it is commutative. This choice is additionally motivated by the connection with ``language union'' in the definition of "\ratms". Finally, we shall later consider a special class of "rating sets", equipped with another binary operation, which we will denote multiplicatively.}. Given a rating set $R$, we define the relation ``$\leq$'' over $R$ as follows:
\[\text{For all }
  r, s\in R,\quad r\leq s \text{ when } r+s=s.
\]
\begin{fct}\label{fct:bgen:compatible-with-addition}
  The relation $\leq$ is a partial order, which makes $(R,+,0,{\leq})$ an ordered monoid.
\end{fct}

\begin{proof}
  It can be verified that ``$\leq$'' is indeed a partial order (note that it is reflexive because $R$ is idempotent). Let us check that it is compatible with addition. Let $r_1\leq r_2$ and $s_1\leq s_2$, we have to prove that $r_1+s_1\leq r_2+s_2$. By definition of ``$\leq$'',  $r_1\leq r_2$ means that  $r_2=r_1+r_2$. Similarly, $s_2=s_1+s_2$. Therefore, $r_2+s_2=r_1+r_2+s_1+s_2=(r_2+s_2)+(r_1+s_1)$ since addition is commutative. This exactly means that $r_1+s_1\leq r_2+s_2$.
\end{proof}

\begin{example}
  An important idempotent rating set is $(2^{A^*},{\cup},\emptyset)$. Its associated order is inclusion.
\end{example}

Another useful property is that if we consider a morphism between two rating sets, then this morphism is always increasing for the order $\leq$.

\begin{fct}\label{fct:bgen:compatible-with-morphisms}
  Let $(Q,+)$ and $(R,+)$ be "rating sets" and $\delta: Q \to R$ be a monoid morphism. Then, $\delta$ is increasing: for any $q_1,q_2 \in Q$ such that $q_1 \leq q_2$, we have $\delta(q_1) \leq \delta(q_2)$.
\end{fct}

\begin{proof}
  If $q_1 \leq q_2$, then $q_2 = q_2 + q_1$. Since $\delta$ is a morphism, we have $\delta(q_2)= \delta(q_2) + \delta(q_1)$ which exactly says that $\delta(q_1) \leq \delta(q_2)$.
\end{proof}

Once we have a "rating set" $R$, a \AP""\ratm"" for $R$ is a monoid morphism $\rho: (2^{A^*},\cup,\emptyset) \to (R,+,0_R)$, \emph{i.e.}, a map from $2^{A^{*}}$ to $R$ satisfying the following properties:
\begin{enumerate}
\item\label{itm:bgen:fzer} $\rho(\emptyset) = 0_R$.
\item\label{itm:bgen:ford} For all $K_1,K_2 \subseteq A^*$, we have $\rho(K_1\cup K_2)=\rho(K_1)+\rho(K_2)$.
\end{enumerate}

For the sake of improved readability, when applying a "\ratm" $\rho$ to a singleton set $K = \{w\}$, we shall write $\rho(w)$ for $\rho(\{w\})$. Fact~\ref{fct:bgen:compatible-with-morphisms} immediately yields the following result.
\begin{fct}
  \label{fct:bgen:increasing}
  Any "\ratm"  $\rho: 2^{A^*} \to R$ is increasing:
  \[
    \text{For all }K_1,K_2 \subseteq A^*\text{ such that }K_1\subseteq K_2, \text{ we have }\rho(K_1)\leq\rho(K_2).
  \]
\end{fct}

Recall that given a "\ratm" $\rho$, our goal is to define when a "\Cs-cover" is "optimal" for  $\rho$, and to obtain algorithms for specific classes \Cs computing such "optimal" "\Cs-covers" with respect to~$\rho$. In order to carry out this computation, we need some additional properties on~$\rho$. The first is called~``\nice{}ness''.

\medskip\noindent
We say a "\ratm" $\rho$ is \AP""\nice"" when it satisfies the following property:
\begin{equation} \label{eq:bgen:fgen}
  \text{For any language $K \subseteq A^*$,} \quad \rho(K) = \sum_{w \in K} \rho(w).
\end{equation}

\begin{remark}
  Observe that the sum in \eqref{eq:bgen:fgen} is defined even when $K$ is infinite. This is because $R$ is finite, commutative and idempotent, hence the sum boils down to a finite one.
\end{remark}

\begin{remark}
  Not all "\ratms" are "\nice". Consider the "rating set" $R = \{0,1,2\}$ whose addition is defined by $i+j = \mathord{\text{max}}(i,j)$ for $i,j \in R$. We define $\rho: 2^{A^*} \to R$ by $\rho(\emptyset) = 0$ and for any nonempty $K \subseteq A^*$, $\rho(K) = 1$ if $K$ is finite and $\rho(K) = 2$ if $K$ is infinite. One may verify that $\rho$ is a "\ratm" which is not "\nice": for any infinite language $K$, we have $\rho(K) = 2$ while $\sum_{w \in K} \rho(w) = 1$.
\end{remark}

\begin{remark}
  \label{rem:finite-repr-ratm}
  It should be noticed that it is not clear how to finitely represent a "\ratm". However, \eqref{eq:bgen:fgen} shows that any \nice "\ratm" is fully determined by the images of singleton languages $\{w\}$. While this does not yield a finite representation, in our algorithms, we shall work with "\ratms" having stronger properties that make them finitely representable (see Section~\ref{sec:bgen:tame}).
\end{remark}

\medskip
\noindent {\bf Canonical \ratm associated to a finite multiset.} While the above definition is abstract, we are mainly interested in a particular example of "\ratm" which connects the framework presented here to the "covering problem". Given a finite multiset of languages \Lb, observe that $2^\Lb$ is an ordered commutative idempotent monoid with union ``$\cup$'' as the addition. Since addition is union, the order is inclusion. Indeed, $\Hb_1 \subseteq \Hb_2$ if and only if $\Hb_1 \cup \Hb_2 = \Hb_2$.

We use $2^\Lb$ as the "rating set" of a specific {\bf "\nice"} "\ratm" \AP$""\rho_\Lb"": 2^{A^*} \to 2^\Lb$ which we associate to \Lb. We define this "\ratm" as follows:
\[
  \begin{array}{llll}
    \rho_\Lb: & 2^{A^*} & \to     & 2^\Lb                                      \\
              & K       & \mapsto & \{L \in \Lb \mid L \cap K \neq \emptyset\}.
  \end{array}
\]
\begin{fct}\label{fct:bgen:filters}
  For any finite multiset of languages \Lb, the mapping $"\rho_\Lb"$ is a "\nice" "\ratm".
\end{fct}

\begin{proof}
  Let us first verify that $"\rho_\Lb"$ is a "\ratm". Clearly, $"\rho_\Lb"(\emptyset) = \{L \in \Lb \mid L \cap \emptyset \neq \emptyset\} = \emptyset$. Moreover given $K_1,K_2 \subseteq A^*$,
  \[
    \begin{array}{lll}
      "\rho_\Lb"(K_1 \cup K_2) & = & \{L \in \Lb \mid L \cap (K_1 \cup K_2) \neq \emptyset\} \\
                             & = & \{L \in \Lb \mid (L \cap K_1) \cup (L \cap K_2) \neq \emptyset\} \\
                             & = & \{L \in \Lb \mid (L \cap K_1)  \neq \emptyset\} \cup \{L \in \Lb \mid (L \cap K_2) \neq \emptyset\}\\
                             & = & "\rho_\Lb"(K_1) \cup "\rho_\Lb"(K_2).
    \end{array}
  \]
  It remains to show that $"\rho_\Lb"$ is "\nice". Consider $K \subseteq A^*$. We prove that $"\rho_\Lb"(K) = \sum_{w \in K} "\rho_\Lb"(w)$. Since $"\rho_\Lb"$ is a "\ratm", we have $"\rho_\Lb"(w)\subseteq "\rho_\Lb"(K)$ when $w\in K$ by Fact~\ref{fct:bgen:increasing}, whence $\sum_{w \in K} "\rho_\Lb"(w) \subseteq "\rho_\Lb"(K)$. It remains to prove the converse inclusion. By definition, we have $"\rho_\Lb"(K) = \{L \in \Lb \mid L \cap K \neq \emptyset\}$. Thus, for any $L \in "\rho_\Lb"(K)$, there exists a word $w_L \in L \cap K$. In particular, $L\in"\rho_\Lb"(w_L)$. Therefore,
  \[
    "\rho_\Lb"(K)  = \sum_{L \in "\rho_\Lb"(K)} \{L\} \subseteq \sum_{L \in "\rho_\Lb"(K)} "\rho_\Lb"(w_L) \subseteq \sum_{w \in K} "\rho_\Lb"(w).
  \]
  This concludes the proof.
\end{proof}

\subsection{\emph{\kl{\Imprints}}}

Now that we have "\ratms", we turn to "\imprints". Consider a "\ratm" $\rho: 2^{A^*} \to R$. Given any finite set of languages \Kb, we define the "$\rho$-\imprint" of \Kb. Intuitively, when \Kb is a "cover" of some language $L$, this object measures the ``quality'' of \Kb.

\begin{remark} \label{rem:bgen:imprints}
  We are mainly interested in the case when \Kb is a "cover". However, the definition of "$\rho$-\imprint" makes senses regardless of this hypothesis. In fact, it is often convenient in proofs to use it when \Kb is not necessarily a "cover".
\end{remark}

Intuitively, we want to define the "$\rho$-\imprint" of \Kb as the set $\rho(\Kb) \subseteq R$ of all images $\rho(K)$ for $K \in \Kb$. However, it will be convenient to use a slightly different definition which is equivalent for our objective and simplifies the notation. Observe that since $R$ is an ordered set, we may apply a downset operation to subsets of $R$. For any $E \subseteq R$, we write:
\[
  \dclos E = \{r \mid \exists r' \in E \text{ such that } r \leq r'\}.
\]
The \AP""$\rho$-\imprint"" \emph{of \Kb}, denoted by \AP""\prin{\rho}{\Kb}"", is the set,
\[
  \begin{array}{lll}
    \prin{\rho}{\Kb} & = & \dclos \{\rho(K) \mid K \in \Kb\}                                               \subseteq  R \\
                     & = & \{r \in R \mid \text{there exists $K \in
                           \Kb$ such that $r \leq \rho(K)$}\}.
  \end{array}
\]

Before we illustrate this notion with the "\ratms" $"\rho_\Lb"$ associated to finite multisets of languages, let us make a few observations about "$\rho$-\imprints". First observe that since any "$\rho$-\imprint" is a subset of the finite "rating set" $R$ associated to $\rho$, there are finitely many possible "$\rho$-\imprints", even though there are infinitely many sets of languages \Kb. Another simple observation is that all "\imprints" are closed under downset.

\begin{fct} \label{fct:bgen:downset}
  Let $\rho: 2^{A^*} \to R$ be a "\ratm". For any finite set of languages \Kb, the "$\rho$-\imprint" of \Kb is closed under downset:
  \[
    \dclos "\prin{\rho}{\Kb}" = "\prin{\rho}{\Kb}".
  \]
  In other words, for any $r \in "\prin{\rho}{\Kb}"$ and any $r' \leq r$, we have $r' \in "\prin{\rho}{\Kb}"$.
\end{fct}

Observe that when \Kb is the "cover" of some language $L$, the "$\rho$-\imprint" of \Kb always contains some trivial elements within the evaluation set $R$. We define the trivial $\rho$-\imprint on $L$ as follows:
\[
  \AP""\itriv{L,\rho}"" = \dclos \{\rho(w) \mid w \in L\} \subseteq R.
\]
When \Kb is a "cover" of $L$, we know that for any $w \in L$, there exists $K \in \Kb$ such that $w \in K$. Thus, $\rho(w) \leq \rho(K)$ by Fact~\ref{fct:bgen:increasing} and it is immediate by closure under downset that all $r \leq \rho(w)$ belong to "\prin{\rho}{\Kb}". Thus, we deduce the following fact.

\begin{fct} \label{fct:bgen:trivialsets}
  Let $\rho: 2^{A^*} \to R$ be a "\ratm". For any language $L$ and any  "cover" \Kb of $L$, we have $"\itriv{L,\rho}" \subseteq "\prin{\rho}{\Kb}"$.
\end{fct}

\medskip
\noindent
{\bf The special case of the "\ratms" $"\rho_\Lb"$.} We now illustrate \imprints with the special case of "\ratms" $"\rho_\Lb"$ associated to finite multisets of languages \Lb. In particular, we present a property which is specific to these "\ratms" and that we use to connect these definitions to the "covering problem".

Consider a finite multiset of languages \Lb and the associated "\ratm" $"\rho_\Lb": 2^{A^*} \to 2^\Lb$. If we unravel the definitions for this specific case, we get,
\[
  \prin{\rho_{\Lb}}{\Kb} = \{\Hb \subseteq \Lb \mid \text{there exists $K \in
    \Kb$ such that $H \cap K \neq \emptyset$ for all $H \in \Hb$}\} \subseteq 2^\Lb.
\]
We illustrate this special case in Figure~\ref{fig:bgen:finer} below.

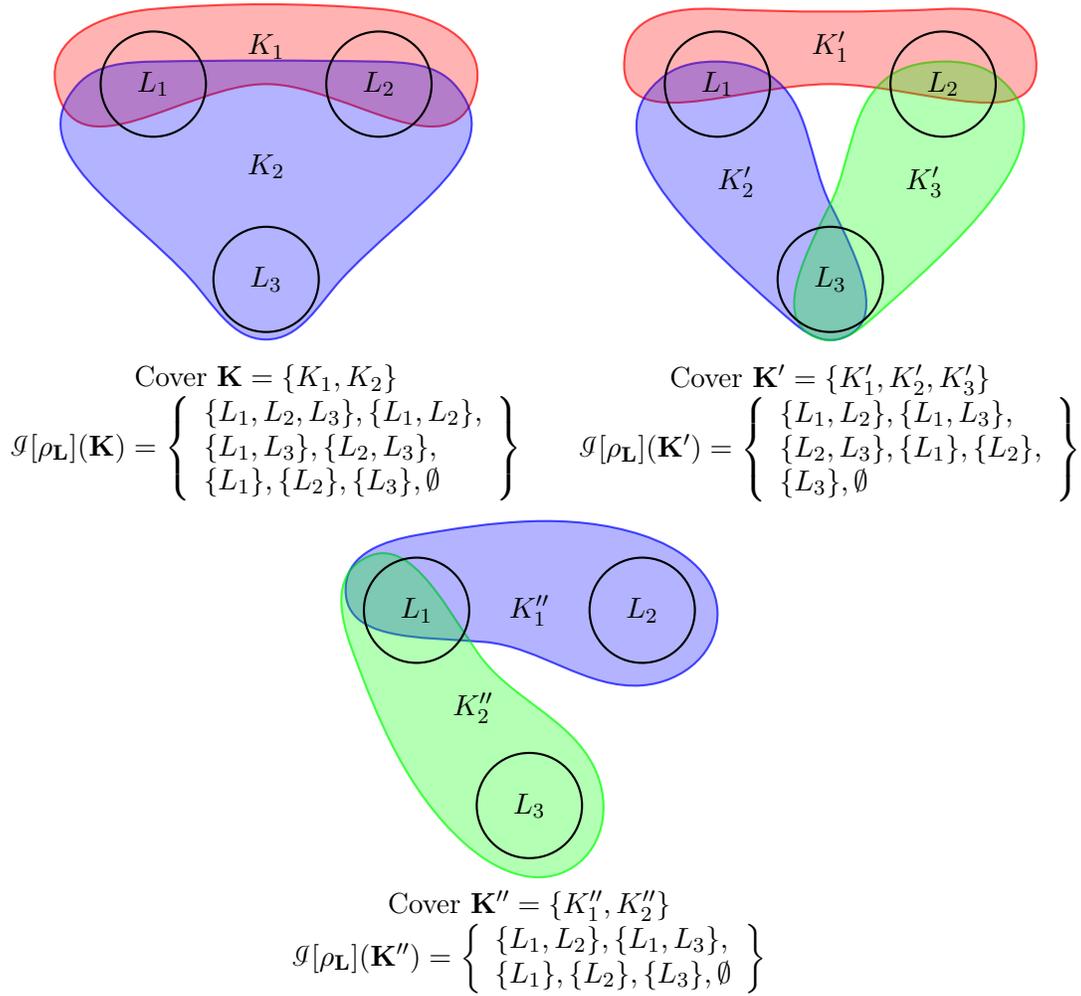
\begin{figure}[!ht]
  \begin{center}
    \begin{tikzpicture}

      \draw[thick] (0,0) circle (0.7cm) node (l1) {$L_1$};
      \draw[thick] (3,0) circle (0.7cm) node (l2) {$L_2$};
      \draw[thick] ($(0,0)!1!-60:(3,0)$) circle (0.7cm) node (l3) {$L_3$};

      \coordinate (m1) at ($(l1)+(1.0,0.0)$);
      \coordinate (m2) at ($(l2)+(1.0,0.0)$);
      \coordinate (m3) at ($(l3)+(1.0,0.0)$);

      \begin{pgfonlayer}{background}
        \def\fstlang{($(l1)!1.3!180:(m1)$) to[closed,curve
          through={
            ($(l1)!1!90:(m1)$)
            .. ($(l2)!1!90:(m2)$)
            .. ($(l2)!1.3!0:(m2)$)
            .. ($(l2)!1/2!(l1)$)
          }]
          ($(l1)!1.3!180:(m1)$)}

        \def\seclang{($(l3)!0.8!-90:(m3)$) to[closed,curve
          through={
            ($(l3)!1!180:(m3)$)
            .. ($(l1)!1!180:(m1)$)
            .. ($(l1)!0.3!90:(m1)$)
            .. ($(l2)!0.3!90:(m2)$)
            .. ($(l2)!1!0:(m2)$)
            .. ($(l3)!1!0:(m3)$)
          }]
          ($(l3)!0.8!-90:(m3)$)}

        \begin{scope}[draw opacity=0.7,fill opacity=0.3]
          \draw[thick,red] \fstlang;
          \draw[thick,blue] \seclang;
          \fill[red] \fstlang;
          \fill[blue] \seclang;
        \end{scope}

        \node at ($($(l2)!1/2!(l1)$)+(0.0,0.5)$) {$K_1$};
        \node at ($(l3)!1/2!30:(l2)$) {$K_2$};

        \node[align=center,anchor=north] at ($(l3)-(0,1.0)$) {"Cover" $\Kb =
          \{K_1,K_2\}$\\$\lprin{\Kb} =
          \left\{\begin{array}{l}\{L_1,L_2,L_3\},\{L_1,L_2\},\\\{L_1,L_3\},\{L_2,L_3\},\\\{L_1\},\{L_2\},\{L_3\},\emptyset\end{array}\right\}$};

      \end{pgfonlayer}

      \begin{scope}[xshift=7.5cm]

        \draw[thick] (0,0) circle (0.7cm) node (l1) {$L_1$};
        \draw[thick] (3,0) circle (0.7cm) node (l2) {$L_2$};
        \draw[thick] ($(0,0)!1!-60:(3,0)$) circle (0.7cm) node (l3) {$L_3$};

        \coordinate (m1) at ($(l1)+(1.0,0.0)$);
        \coordinate (m2) at ($(l2)+(1.0,0.0)$);
        \coordinate (m3) at ($(l3)+(1.0,0.0)$);

        \begin{pgfonlayer}{background}
          \def\fstlang{($(l1)!1.3!160:(m1)$) to[closed,curve
            through={
              ($(l1)!1!90:(m1)$)
              .. ($(l2)!1!90:(m2)$)
              .. ($(l2)!1.3!20:(m2)$)
              .. ($(l2)!1/2!(l1)$)
            }]
            ($(l1)!1.3!160:(m1)$)}
          \def\seclang{($(l3)!0.8!280:(m3)$) to[closed,curve
            through={
              ($(l3)!0.8!210:(m3)$)
              .. ($(l1)!0.7!170:(m1)$)
              .. ($(l1)!0.3!90:(m1)$)
              .. ($(l1)!0.7!10:(m1)$)
              .. ($(l3)!1!90:(m3)$)
            }]
            ($(l3)!0.8!280:(m3)$)}

          \def\trdlang{($(l3)!0.8!-100:(m3)$) to[closed,curve
            through={
              ($(l3)!0.8!-30:(m3)$)
              .. ($(l2)!0.7!10:(m2)$)
              .. ($(l2)!0.3!90:(m2)$)
              .. ($(l2)!0.7!170:(m2)$)
              .. ($(l3)!1!90:(m3)$)
            }]
            ($(l3)!0.8!-100:(m3)$)}

          \begin{scope}[draw opacity=0.7,fill opacity=0.3]
            \draw[thick,red] \fstlang;
            \draw[thick,blue] \seclang;
            \draw[thick,green] \trdlang;
            \fill[red] \fstlang;
            \fill[blue] \seclang;
            \fill[green] \trdlang;
          \end{scope}

          \node at ($($(l2)!1/2!(l1)$)+(0.0,0.5)$) {$K'_1$};
          \node at ($($(l3)!1/2!(l1)$)-(0.5,0.0)$) {$K'_2$};
          \node at ($($(l2)!1/2!(l3)$)+(0.5,0.0)$) {$K'_3$};
        \end{pgfonlayer}

        \node[align=center,anchor=north] at ($(l3)-(0,1.0)$) {"Cover" $\Kb' =
          \{K'_1,K'_2,K'_3\}$\\$\lprin{\Kb'} =
          \left\{\begin{array}{l}\{L_1,L_2\},\{L_1,L_3\},\\\{L_2,L_3\},\{L_1\},\{L_2\},\\\{L_3\},\emptyset\end{array}\right\}$};

      \end{scope}

      \begin{scope}[xshift=3.5cm,yshift=-7cm]

        \draw[thick] (0,0) circle (0.7cm) node (l1) {$L_1$};
        \draw[thick] (3,0) circle (0.7cm) node (l2) {$L_2$};
        \draw[thick] ($(0,0)!1!-60:(3,0)$) circle (0.7cm) node (l3) {$L_3$};

        \coordinate (m1) at ($(l1)+(1.0,0.0)$);
        \coordinate (m2) at ($(l2)+(1.0,0.0)$);
        \coordinate (m3) at ($(l3)+(1.0,0.0)$);

        \begin{pgfonlayer}{background}
          \def\fstlang{($(l1)!1!160:(m1)$) to[closed,curve
            through={
              ($(l1)!1!90:(m1)$)
              .. ($(l2)!1!90:(m2)$)
              .. ($(l2)!1!0:(m2)$)
              .. ($(l2)!1!-90:(m2)$)
              .. ($(l1)!1!-25:(m1)$)
            }]
            ($(l1)!1!160:(m1)$)}

          \def\seclang{($(l1)!1!140:(m1)$) to[closed,curve
            through={
              ($(l1)!1!210:(m1)$)
              .. ($(l3)!1!210:(m3)$)
              .. ($(l3)!1!300:(m3)$)
              .. ($(l3)!1!390:(m3)$)
              .. ($(l1)!1!-35:(m1)$)
            }]
            ($(l1)!1!140:(m1)$)}

          \begin{scope}[draw opacity=0.7,fill opacity=0.3]
            \draw[thick,blue] \fstlang;
            \draw[thick,green] \seclang;
            \fill[blue] \fstlang;
            \fill[green] \seclang;
          \end{scope}

          \node at ($(l2)!1/2!(l1)$) {$K''_1$};
          \node at ($(l3)!1/2!(l1)$) {$K''_2$};
        \end{pgfonlayer}

        \node[align=center,anchor=north] at ($(l3)-(0,1.0)$) {"Cover" $\Kb'' =
          \{K''_1,K''_2\}$\\$\lprin{\Kb''} =
          \left\{\begin{array}{l}\{L_1,L_2\},\{L_1,L_3\},\\\{L_1\},\{L_2\},\{L_3\},\emptyset\end{array}\right\}$};
      \end{scope}
    \end{tikzpicture}
  \end{center}
  \caption{Examples of "covers"  and their
    $"\rho_\Lb"$-\imprints for $\Lb = \{L_1,L_2,L_3\}$}
  \label{fig:bgen:finer}
\end{figure}

We now connect $"\rho_\Lb"$-"\imprints" to the "covering problem". Given a finite set of languages \Kb, \lprin{\Kb} records the subsets of~\Lb for which \Kb is not "separating".

\begin{lemma} \label{lem:bgen:profile}
  Let \Lb be a finite multiset of languages and let \Kb be a finite set of languages. Consider a subset \Hb of\/ \Lb. The following properties are equivalent:
  \begin{enumerate}
  \item \Kb is "separating" for \Hb.
  \item $\Hb \not\in \lprin{\Kb}$.
  \end{enumerate}
\end{lemma}

\begin{proof}
  By definition, $\Hb \not\in \lprin{\Kb}$ if and only if for all $K \in \Kb$, there exists $H \in \Hb$ such that $K \cap H = \emptyset$. This is exactly the definition of \Kb being "separating" for \Hb.
\end{proof}

\subsection{Optimality}

Given a "\ratm" $\rho: 2^{A^*} \to R$ and a language $L$, the main intuition is that the ``best'' "covers" \Kb of $L$ are those with the least possible "$\rho$-\imprint" $"\prin{\rho}{\Kb}" \subseteq R$ (with respect to inclusion). This is validated by the "\ratms" $"\rho_\Lb"$ associated to a finite multiset of languages \Lb. Indeed, in view of Lemma~\ref{lem:bgen:profile}, given a "cover" \Kb, the smaller its $"\rho_\Lb"$-\imprint is, the better it is at being "separating" for subsets of \Lb. For example, observe that in Figure~\ref{fig:bgen:finer}, $\Kb''$ is better than $\Kb'$ which is itself better than $\Kb$. Of course, if we do not put any constraint on the "covers" that we may use, these notions are not very useful. There is always a trivial "cover" of $L$ that is better than any other. Indeed, partitioning the words $w \in L$ according to the value of $"\rho_\Lb"(w) \subseteq \Lb$ yields a "cover" whose $\rho_{\Lb}$-\imprint is \itriv{L,"\rho_\Lb"} and it is impossible to do better by Fact~\ref{fct:bgen:trivialsets}.

However, when we restrict ourselves to \emph{"\Cs-covers"} for a fixed "lattice" \Cs, this ``trivial best cover'' is not necessarily a "\Cs-cover", since it might contain languages that do not belong to \Cs. Hence, it now makes sense to define a notion of ``best'' "\Cs-cover".

\subsubsection*{Definition}

We may now define "optimality" formally. Consider an arbitrary "\ratm" $\rho: 2^{A^*} \to R$ and a "lattice" \Cs. Given a language $L$, \AP an ""optimal"" "\Cs-cover" of $L$ for $\rho$ is a "\Cs-cover" \Kb of $L$ which satisfies the following property:
\[
  "\prin{\rho}{\Kb}" \subseteq "\prin{\rho}{\Kb'}" \quad \text{for any "\Cs-cover" $\Kb'$ of $L$}.
\]
Moreover, as we already announced in Section~\ref{sec:covering}, an important special case is when $L$ is the universal language $A^*$. In this case, we shall speak of \emph{"optimal" "universal \Cs-cover" for $\rho$}.

\begin{example} \label{ex:bgen:alph}
  We use the alphabet $A = \{a,b,c\}$. Consider the class \AP""\at"" of alphabet testable languages (\emph{i.e.}, the Boolean combinations of languages of the form $A^*aA^*$ for some $a \in A$). Let $\Lb = \{(ab)^+,b(ab)^+,c(ac)^+\}$. Consider
  \[
    \Kb = \{A^*bA^*,\; A^* \setminus (A^*bA^*)\}.
  \]
  One may verify that \Kb is an "optimal" universal "\at"-"cover" for \Lb.  Its \Lb-\imprint is
  \[
    \bigl\{\{(ab)^+,b(ab)^+\}, \{(ab)^+\},\{b(ab)^+\},\{c(ac)^+\},\emptyset\bigr\}.
  \]
  Note that it contains $\{(ab)^+,b(ab)^+\}$. Therefore, by the necessary condition stated above, $(ab)^+$ and $b(ab)^+$ cannot be separated by alphabet testable languages. This can be verified directly, and used to prove that the "cover" $\Kb$ is indeed "optimal". Note also that there are other "optimal" universal \at-covers for \Lb. For instance
  \[
    \{A^*cA^*,\; A^* \setminus (A^*cA^*)\}
  \]
  is another universal "cover" with the same \Lb-\imprint as \Kb.
\end{example}

Note that in general, there can be infinitely many "optimal" "\Cs-covers" for a given "\ratm" $\rho$. We now prove that there always exists at least one. In order to prove this, we need our hypothesis that \Cs is a "lattice" (more precisely, we need \Cs to be closed under finite intersection).

\begin{lemma} \label{lem:bgen:opt}
  Let \Cs be a "lattice". Then, for any language $L$ and any "\ratm" $\rho: 2^{A^*} \to R$, there exists an "optimal" "\Cs-cover" of $L$ for $\rho$.
\end{lemma}

\begin{proof}
  We already know that there exists a "\Cs-cover" of $L$ since $\{A^*\}$ is such a "cover". We prove that for any two "\Cs-covers" $\Kb'$ and $\Kb''$ of $L$, there exists a third "\Cs-cover" $\Kb$ of $L$ such that $"\prin{\rho}{\Kb}" \subseteq \prin{\rho}{\Kb'}$ and $"\prin{\rho}{\Kb}" \subseteq \prin{\rho}{\Kb''}$. Since there are only finitely possible "$\rho$-\imprints" (they are all subsets of the finite "rating set" $R$), the lemma will follow.

  We define $\Kb = \{K' \cap K'' \mid K' \in \Kb' \text{ and } K'' \in \Kb''\}$. Since $\Kb'$ and $\Kb''$ are $\Cs$-covers of $L$, the set \Kb is also a "cover" of $L$. Moreover, it is a "\Cs-cover" since \Cs is closed under intersection. Finally, it is immediate from Fact~\ref{fct:bgen:increasing} that $"\prin{\rho}{\Kb}" \subseteq \prin{\rho}{\Kb'}$ and $"\prin{\rho}{\Kb}" \subseteq \prin{\rho}{\Kb''}$.
\end{proof}

An important remark is that the proof of Lemma~\ref{lem:bgen:opt} is \textbf{non-constructive}. Given a "\ratm" $\rho: 2^{A^*} \to R$, computing an actual "optimal" "\Cs-cover" of some language $L$ for $\rho$ is a difficult problem in general. As seen in Theorem~\ref{thm:bgen:main} below, when we work with the "\ratm" $"\rho_\Lb"$ associated to some finite multiset of languages \Lb, this solves "\Cs-covering" for any pair $(L,\Hb)$ where \Hb is a subset of~\Lb: whenever a "\Cs-cover" of $L$ which is "separating" for $\Hb \subseteq \Lb$ exists, any "optimal" "\Cs-cover" of $L$ is one. Before we present this theorem, let us make a key observation about "optimal" "\Cs-covers".

\subsubsection*{"Optimal" \imprint}

By definition, given a "lattice" \Cs, a language $L$ and a "\ratm" $\rho: 2^{A^*} \to R$, all "optimal" "\Cs-covers" of $L$ for $\rho$ have the same "$\rho$-\imprint". Hence, this unique "$\rho$-\imprint" is a \emph{canonical} object for \Cs, $L$ and $\rho$. We say that it is the \emph{\Cs-"optimal" "$\rho$-\imprint" on $L$} and we denote it by \AP$""\opti{\Cs}{L,\rho}""$:
\[
  \opti{\Cs}{L,\rho} = "\prin{\rho}{\Kb}" \quad \text{for any "optimal" "\Cs-cover" \Kb of $L$  for $\rho$}.
\]
Let us complete this definition with a few useful results about "optimal" \imprints. We present two simple facts that one may use to compare "optimal" \imprints for different classes and languages.

\begin{fct} \label{fct:cinclus}
  Let $\rho: 2^{A^*} \to R$ be a "\ratm" and consider two "lattices" \Cs and \Ds such that $\Cs \subseteq \Ds$. Then, for any language $L$, we have $\opti{\Ds}{L,\rho} \subseteq "\opti{\Cs}{L,\rho}"$.
\end{fct}

\begin{proof}
  Consider an "optimal" "\Cs-cover" \Kb of $L$ for $\rho$. By definition, we have $"\opti{\Cs}{L,\rho}" = "\prin{\rho}{\Kb}"$. Moreover, since $\Cs \subseteq \Ds$, the "\Cs-cover" \Kb is also a \Ds-cover of $L$ and we have $\opti{\Ds}{L,\rho} \subseteq "\prin{\rho}{\Kb}"$. Altogether, this yields $\opti{\Ds}{L,\rho} \subseteq "\opti{\Cs}{L,\rho}"$ as desired.
\end{proof}

\begin{fct} \label{fct:linclus}
  Let $\rho: 2^{A^*} \to R$ be a "\ratm" and consider two languages $H,L$ such that $H \subseteq L$. Then, for any "lattice" \Cs, we have $\opti{\Cs}{H,\rho} \subseteq "\opti{\Cs}{L,\rho}"$.
\end{fct}

\begin{proof}
  Consider an "optimal" "\Cs-cover" \Kb of $L$ for $\rho$. By definition, we have $"\opti{\Cs}{L,\rho}" = "\prin{\rho}{\Kb}"$. Moreover, since $H \subseteq L$, we know that \Kb is also a "\Cs-cover" of $H$ and we have $\opti{\Cs}{H,\rho} \subseteq "\prin{\rho}{\Kb}"$. Altogether, this yields $\opti{\Cs}{H,\rho} \subseteq "\opti{\Cs}{L,\rho}"$ as desired.
\end{proof}

\subsection{Connection with the \kl{covering problem}}

We are now ready to connect "optimal" "\Cs-covers" to the "\Cs-covering problem": we express the notion of \emph{"optimal"} "\Cs-cover" with that of \emph{"separating"} "\Cs-cover". The connection is given in the following theorem, through the "\ratm" $"\rho_\Lb"$ canonically associated to a multiset of languages~\Lb.

\begin{theorem} \label{thm:bgen:main}
  Let \Cs be a "lattice". Consider a language $L$ and a finite multiset of languages \Lb. Given any subset $\Hb \subseteq \Lb$, the following properties are equivalent:
  \begin{enumerate}[ref=\arabic*)]
  \item\label{item:cov2sep1} $(L,\Hb)$ is \Cs-coverable.
  \item\label{item:cov2sep2} $\Hb \not\in \opti{\Cs}{L,\rho_\Lb}$.
  \item\label{item:cov2sep3} Any "optimal" "\Cs-cover" of $L$ for \Lb is "separating" for \Hb.
  \end{enumerate}
\end{theorem}

\begin{proof}
  We prove that $\ref{item:cov2sep3} \Rightarrow \ref{item:cov2sep1} \Rightarrow  \ref{item:cov2sep2} \Rightarrow \ref{item:cov2sep3}$. Let us first assume  that $\ref{item:cov2sep3}$ holds, \emph{i.e.}, that any "optimal" "\Cs-cover" of $L$ for \Lb, is "separating" for \Hb. Since there exists at least one "optimal" "\Cs-cover" of $L$ for \Lb (see Lemma~\ref{lem:bgen:opt}), $(L,\Hb)$ is \Cs-coverable, \emph{i.e.}, $\ref{item:cov2sep1}$ holds.

  We now prove that $\ref{item:cov2sep1} \Rightarrow \ref{item:cov2sep2}$. Assume that $\ref{item:cov2sep1}$ holds, \emph{i.e.}, that $(L,\Hb)$ is \Cs-coverable. This means that there exists a "\Cs-cover" \Kb of $L$ which is "separating" for \Hb. It follows from Lemma~\ref{lem:bgen:profile} that $\Hb \not\in \lprin{\Kb}$. Finally, since \Kb is a "\Cs-cover" of $L$, we have $\opti{\Cs}{L,\rho_\Lb} \subseteq \lprin{\Kb}$ by definition and we conclude that $\Hb \not\in \opti{\Cs}{L,\rho_\Lb}$. Therefore, $\ref{item:cov2sep2}$ holds.

  It remains to prove that $\ref{item:cov2sep2} \Rightarrow \ref{item:cov2sep3}$. Assume that $\Hb \not\in \opti{\Cs}{L,\rho_\Lb}$ and let $\Kb$ be an "optimal" "\Cs-cover" of $L$ for \Lb, so that $\opti{\Cs}{L,\rho_\Lb}=\lprin{\Kb}$. Therefore, our hypothesis yields that $\Hb \not\in \lprin{\Kb}$, and it follows from Lemma~\ref{lem:bgen:profile} that \Kb is "separating" for \Hb.
\end{proof}

In view of Theorem~\ref{thm:bgen:main}, both objectives in the "\Cs-covering problem" can now be reformulated with our new terminology. In order to decide "\Cs-covering" for a particular input $(L,\Lb)$, it suffices to compute \opti{\Cs}{L,\rho_\Lb}, the \Cs-"optimal" $\rho_\Lb$-\imprint on $L$ (this is the second item in the theorem). Similarly, if a "\Cs-cover" of $L$ which is "separating" for \Lb exists, it suffices to compute an "optimal" "\Cs-cover" of $L$ for~\Lb to obtain one (this is the third item in the theorem). There are several motivations for using this new formulation:
\begin{enumerate}
\item The \Cs-optimal $"\rho_\Lb"$-"\imprint" on $L$, \opti{\Cs}{L,\rho_\Lb}, is a canonical object associated to \Cs, $L$ and \Lb which always exists, regardless of whether the answer to "\Cs-covering" is ``yes'' or ``no''.
\item This approach yields an abstraction of "\Cs-covering" which enjoys lighter terminology and more elegant presentations for solutions. The key idea is that we do not have to consider finite multisets of languages explicitly: we replace the goal of finding an algorithm deciding whether an input pair $(L,\Lb)$ is \Cs-coverable by the one of finding an algorithm computing "\opti{\Cs}{L,\rho}" from some input \ratm $\rho: 2^{A^*} \to R$. Of course, we shall restrict ourselves to a special class of "\ratms" which can be finitely represented. Otherwise, it would not make sense to speak of algorithms taking a "\ratm" as input.
\item We prove that in order to obtain covering algorithms, it suffices to consider a special class of "\ratms": the "\nice" "\mratms" (defined in Section~\ref{sec:bgen:tame}). This is where where we exploit the fact that the multiset \Lb in our input is made of \emph{regular} languages. These "\ratms" enjoy additional properties which are crucially exploited by our algorithms. In particular, a key point is they can be finitely represented. Consequently, it does make sense to speak of algorithms taking a "\nice" "\mratm" as input.
\end{enumerate}

\subsection{Extension}

Let us finish the section by presenting a natural relation between \ratms: \emph{extension}. It is designed with the following objective in mind. Given two \ratms $\rho$ and $\tau$, if \emph{$\tau$ extends $\rho$}, we want to have the following property for any lattice \Cs: "\opti{\Cs}{L,\rho}" is easily computed from \opti{\Cs}{L,\tau}

\begin{remark}
  We shall use extension to restrict the class of \ratms that one needs to consider for solving covering problems. As seen in Theorem~\ref{thm:bgen:main}, given some lattice \Cs, obtaining an algorithm for \Cs-covering reduces to getting a procedure which computes \opti{\Cs}{L,"\rho_\Lb"} from an input pair $(L,\Lb)$ (\emph{i.e.}, $L$ is a regular language and \Lb a finite multiset of regular languages). Hence, we need to handle all \ratms $"\rho_\Lb"$ associated to some finite multiset of regular languages \Lb. Extension is used to show that we may restrict ourselves to smaller classes of \ratms without loss of generality: we prove that for any finite multiset of regular languages \Lb, one may compute the finite representation of a \ratm extending $\rho_{\Lb}$ in the smaller class. We present the most important restriction that we shall consider in the next section: we always work with \emph{"\tame"} "\ratms".
\end{remark}

Consider two \ratms $\rho:2^{A^*}\to R$ and $\tau: 2^{A^*} \to Q$. By definition, the two rating sets $R$ and $Q$ are finite commutative and idempotent monoids. We say that $\tau$ \emph{extends} $\rho$ when there exists a monoid morphism $\delta: Q \to R$ such that,
\[
  \rho = \delta \circ \tau.
\] 
We call $\delta$ the \emph{extending morphism}. Let us point out that since rating sets must be finite, any extending morphism is clearly finitely representable. Another simple but important observation is that extension is a transitive relation.

\begin{fct} \label{lem:bgen:extrans}
  Let $\rho_1:2^{A^*}\to R_1$, $\rho_2:2^{A^*}\to R_2$ and $\rho_3:2^{A^*}\to R_3$ be \ratms. Assume that $\rho_2$ extends $\rho_1$ and that $\rho_3$ extends $\rho_2$ for the extending morphisms $\delta_1: R_2 \to R_1$ and $\delta_2: R_3 \to R_2$ respectively. Then, $\rho_3$ extends $\rho_1$ for the extending morphism $\delta_1 \circ \delta_2: R_3 \to R_1$.
\end{fct}

We may now present the main property of extension. When $\tau$ extends $\rho$, given any lattice \Cs and any language $L$, "\opti{\Cs}{L,\rho}" is easily computed from \opti{\Cs}{L,\tau} (provided that we have the corresponding extending morphism $\delta$ in hand). Moreover, any "optimal" "\Cs-cover" of $L$ for $\tau$ is also "optimal" $\rho$.

\begin{lemma} \label{lem:bgen:extension}
  Consider two \ratms $\rho:2^{A^*}\to R$ and $\tau: 2^{A^*} \to Q$, and assume that $\tau$ extends $\rho$ with the extending morphism $\delta: Q\to R$. For any lattice \Cs and any language $L$, the two following properties hold:
  \begin{itemize}
  \item $"\opti{\Cs}{L,\rho}" = \dclos \delta(\opti{\Cs}{L,\tau})$.
  \item Any \Cs-cover of $L$ which is optimal for $\tau$ is also optimal for $\rho$.
  \end{itemize}
\end{lemma}

\begin{proof}
  We start with the first item: $"\opti{\Cs}{L,\rho}" = \dclos \delta(\opti{\Cs}{L,\tau})$. We first consider $r \in "\opti{\Cs}{L,\rho}"$ and show that $r \in \dclos \delta(\opti{\Cs}{L,\tau})$. Let \Kb be a \Cs-cover of $L$ which is "optimal" for $\tau$. By definition $"\opti{\Cs}{L,\rho}" \subseteq "\prin{\rho}{\Kb}"$ and we get that $r\in "\prin{\rho}{\Kb}"$. Thus, there exists $K \in \Kb$ such that $r \leq \rho(K)$. Clearly $\tau(K) \in \prin{\tau}{\Kb} = \opti{\Cs}{L,\tau}$ since \Kb is "optimal" for $\tau$. Moreover, by definition of $\delta$, we have $\rho(K) = \delta(\tau(K))$. Thus, since $r \leq \rho(K)$, we have $r \in \dclos \delta(\opti{\Cs}{L,\tau})$.

  Conversely, assume that $r \in \dclos \delta(\opti{\Cs}{L,\tau})$. We show that $r \in "\opti{\Cs}{L,\rho}"$. Let \Kb be a "\Cs-cover" of $L$ which is "optimal" for $\rho$. By definition, $"\prin{\rho}{\Kb}" = "\opti{\Cs}{L,\rho}"$ and it therefore suffices to show that $r \in "\prin{\rho}{\Kb}"$. By definition of $r$, we have $q \in \opti{\Cs}{L,\tau}$ such that $r \leq \delta(q)$. Since $\opti{\Cs}{L,\tau} \subseteq\prin{\tau}{\Kb}$ by definition, this implies $q \in \prin{\tau}{\Kb}$. Hence, we get $K \in \Kb$ such that $q \leq \tau(K)$. Since $\delta$ is a morphism such that $\rho = \delta \circ \tau$, this implies that $\delta(q) \leq \delta(\tau(K)) = \rho(K)$ (see Fact~\ref{fct:bgen:compatible-with-morphisms}). Finally, since $r \leq \delta(q)$, we obtain $r \leq \rho(K)$ which implies that $r \in "\prin{\rho}{\Kb}"$, finishing the proof.

  It remains to prove the second item. Consider a "\Cs-cover" \Kb of $L$ which is "optimal" for $\tau$. We have to show that \Kb is also "optimal" for $\rho$ as well. By hypothesis $\prin{\tau}{\Kb}= \opti{\Cs}{L,\tau}$. Moreover, since $\rho = \delta \circ \tau$, it is simple to verify from the definitions that $"\prin{\rho}{\Kb}" = \dclos \delta(\prin{\tau}{\Kb})$. Thus, we get that $"\prin{\rho}{\Kb}" = \dclos \delta(\opti{\Cs}{L,\tau})$ which yields $"\prin{\rho}{\Kb}" = \copti{L,\rho}$ by the first item. Thus, \Kb is an "optimal" "\Cs-cover" of $L$ for $\rho$.
\end{proof}

\section{\Mratms}
\label{sec:bgen:tame}
We now define a special class of "\ratms": the "\mratms". This class is important for two complementary reasons. First, the "\ratms" which are simultaneously "\nice" and "\tame" are finitely representable. Hence, we are able to speak of computational problems in which the input is a "\nice" "\mratm". For example, given a "lattice" \Cs, the problem of computing \copti{L,\rho} from a regular language $L$ and a "\nice" "\mratm" $\rho$ makes sense. Second, we shall prove that for any finite multiset of regular languages \Lb, we are able to compute a "\nice" "\mratm" extending $"\rho_{\Lb}"$.

Consequently, by Theorem~\ref{thm:bgen:main} and Lemma~\ref{lem:bgen:extension}, we are able to reduce \Cs-covering for any lattice \Cs to the problem of computing \copti{L,\rho} from a regular language $L$ and a "\nice" "\mratm" $\rho$. In other words, considering "\nice" "\mratms" suffices to obtain covering algorithms. This is crucial for our techniques to apply.

\begin{remark}
  We shall actually consider two variants of this reduction. The first one applies to "universal \Cs-covering" only (this allows us to drop the regular language $L$ in the input). We state it precisely in Section~\ref{sec:genba} (see Proposition~\ref{prop:breduc}). The second one applies to the full covering problem and involves a little extra work to handle the regular language $L$. We state it precisely in Section~\ref{sec:genlatts} (see Proposition~\ref{prop:lreduc}).
\end{remark}

We first define "\mratms" and explain how to finitely represent the "\nice" ones. Then, we investigate their properties. Finally, we present a construction for building a "\mratm" extending $"\rho_{\Lb}"$ from an input finite multiset of regular languages \Lb.

\subsection{Definition}

We say that a \ratm $\rho: 2^{A^*} \to R$ is \emph{"\tame"} when its rating set $R$ has more structure: it needs to be an \emph{idempotent "semiring"}. Moreover, $\rho$ has to satisfy an additional property connecting this structure to language concatenation, namely, it has to be a morphism of "semirings". Let us first define "semirings".

\medskip
\noindent
{\bf "Semirings".} A \AP""semiring"" is a tuple $(R,+,\cdot)$ where $R$ is a set and ``$+$'' and ``$\cdot$''  are two binary operations called addition and multiplication, such that the following axioms are satisfied:
\begin{itemize}
\item $(R,+)$ is a commutative monoid whose neutral element is denoted by $0_R$.
\item $(R,\cdot)$ is a monoid whose neutral element is denoted by $1_R$.
\item Multiplication distributes over addition, \emph{i.e.}, for all $r,s,t \in R$ we have:
  \[
    \begin{array}{lll}
      r \cdot (s + t) & = & (r \cdot s) + (r \cdot t), \\
      (r + s) \cdot t & = & (r \cdot t) + (s \cdot t).
    \end{array}
  \]
\item The neutral element ``$0_R$'' of $(R,+)$ is a zero for $(R,\cdot)$, \emph{i.e.}, for any $r \in R$:
  \[
    0_R \cdot r = r \cdot 0_R = 0_R.
  \]
\end{itemize}

\begin{remark}
  Semirings generalize the more standard notion of rings. A ring $(R,+,\cdot)$ is a "semiring" for which $(R,+)$ is a group.
\end{remark}

As usual, for the sake of simplifying the notation, when multiplying elements we shall often write $rr'$ instead of $r\cdot r'$. We say that a "semiring" $R$ is \emph{idempotent} when $r + r = r$ for any $r \in R$, \emph{i.e.}, when the additive monoid $(R,+)$ is idempotent (on the other hand, note that there is no additional constraint on the multiplicative monoid $(R,\cdot)$).

\begin{example}\label{ex:bgen:semiring}
  A simple example of infinite idempotent "semiring" is the set $2^{A^*}$ of all languages over~$A$. Indeed, it suffices to choose union as the addition (with the empty language as neutral element) and concatenation as the multiplication (with the singleton $\{\varepsilon\}$ as neutral element). The ordering is then simply language inclusion. More generally, any class of languages which is closed under union, concatenation and contains the singleton $\{\varepsilon\}$ is a "semiring" as well.
\end{example}

Observe that any finite idempotent "semiring" $R$ is in particular a "rating set" by definition (\emph{i.e.}, $(R,+)$ is an idempotent and commutative monoid). We know from Fact~\ref{fct:bgen:compatible-with-addition} that in an idempotent "semiring", the canonical partial order ``$\leq$'', defined by $r\leq s$ when $r+s=s$, is compatible with addition. Actually, it is also compatible with multiplication.

\begin{fact}\label{fct:bgen:compatible-with-mult}
  Let $R$ be an idempotent "semiring" and let ``$\leq$'' be its induced ordering relation. For all $r_1,r_2,s_1,s_2 \in R$ such that $r_1 \leq r_2$ and $s_1 \leq s_2$, we have $r_1s_1 \leq r_2s_2$.
\end{fact}

\begin{proof}
  Assume that $r_1 \leq r_2$ and $s_1 \leq s_2$. Then by definition of ``$\leq$'', we have $r_2=r_1+r_2$ and $s_2=s_1+s_2$. Therefore, $r_2s_2=(r_1+r_2)(s_1+s_2)=r_1s_1+r_2s_1+r_1s_2+r_2s_2$, whence $r_1s_1+r_2s_2=r_2s_2$ since addition is idempotent. This shows that $r_1s_1\leq r_2s_2$.
\end{proof}

\medskip
\noindent
{\bf "\Mratms".} Now that we have idempotent "semirings", we may define "\mratms": as expected they are "semiring" morphisms. Let $\rho: 2^{A^*} \to R$ be a \ratm. By definition, this means that the rating set $(R,+)$ is an idempotent commutative monoid and that $\rho$ is a monoid morphism from $(2^{A^*},\cup)$ to $(R,+)$. In other words, we have,
\begin{enumerate}[label=(\arabic*)]
\item $\rho(\emptyset) = 0_R$.
\item For all $K_1,K_2 \subseteq A^*$, we have $\rho(K_1\cup K_2)=\rho(K_1)+\rho(K_2)$.
\end{enumerate}
We say that $\rho$ is \AP""\tame"" when the rating set $R$ is equipped with a second binary operation ``$\cdot$'' such that $(R,+,\cdot)$ is an idempotent "semiring" and $\rho$ is also a monoid morphism from $(2^{A^*},\cdot)$ to $(R,\cdot)$. In other words, the two following additional axioms have to be satisfied:
\begin{enumerate}[resume,label=(\arabic*)]
\item\label{itm:bgen:funit} $\rho(\varepsilon) = 1_R$.
\item\label{itm:bgen:fmult} For all $K_1,K_2 \subseteq A^*$, we have $\rho(K_1K_2) = \rho(K_1) \cdot \rho(K_2)$.
\end{enumerate}
Altogether, this exactly says that $\rho$ must be a "semiring" morphism from $(2^{A^*},\cup,\cdot)$ to $(R,+,\cdot)$.

\medskip
\noindent
{\bf A finite representation of "\nice" "\mratms".} It turns out that "\mratms" are finitely representable when they are "\nice". This is a crucial property:  this allows us to speak of algorithms taking a "\nice" "\mratm" as input.

Given any "\mratm" $\rho: 2^{A^*} \to R$ ("\nice" or not), one may associate a morphism $\rho_*: A^* \to R$ between the monoids $(A^*,\cdot)$ and $(R,\cdot)$.  The definition is natural: $\rho_*$ is simply the restriction of $\rho$ to $A^*$:
\[
  \begin{array}{llll}
    \rho_*: & A^* & \to     & R       \\
           & w   & \mapsto & \rho(w).
  \end{array}
\]
Clearly, $\rho_*$ is a monoid morphism by Items~\ref{itm:bgen:funit} and~\ref{itm:bgen:fmult} in the definition of "\mratms". 

The main point here is that when a "\mratm" $\rho: 2^{A^*} \to R$ is "\nice", it is fully determined by the rating set $R$ and by its associated word morphism $\rho_*$. Indeed, as we already noted in Remark~\ref{rem:finite-repr-ratm}, when $\rho$ is a "\nice" \ratm, it is fully defined by the images of singleton languages $\{w\}$ and the addition of the rating set $R$: for any $K \subseteq A^*$, $\rho(K)$ is the sum of all elements $\rho(w)$ for $w \in K$. Moreover, when $\rho$ is "\tame", the images of any singleton language $\{w\}$ is exactly what $\rho_*: A^* \to R$ computes. Both objects are clearly finitely representable: $\rho_*$ is determined by the images $\rho_*(a)$ for $a\in A$ and $R$ is a finite "semiring". Hence, we get a finite representation for "\nice" "\mratms".

\begin{remark}
  From now on, when we speak of a procedure taking a "\nice" "\mratm" $\rho$ as input, we shall implicitly assume that it is given in this form. That is, the procedure takes a finite "semiring" $R$ and a monoid morphism $\beta: A^* \to R$ as input. The "\nice" "\mratm" $\rho$ is the one whose rating set is $R$ and whose associated word morphism is $\rho_*=\beta$.
\end{remark}

Let us complete this definition with an important remark: computing information about an input "\nice" "\mratm" $\rho: 2^{A^*} \to R$ is simple. First, given any regular language $L$, one may compute $\itriv{L,\rho} = \{r \in R \mid \text{$r \leq \rho(w)$ for some $w \in L$}\}$. Indeed, testing whether $r \in \itriv{L,\rho}$ amounts to checking if the (regular) language $\{w \mid r \leq \rho(w)\} \cap L$ is nonempty. Hence, one may compute \itriv{L,\rho} in polynomial time with respect to a recognizer for $L$ and the size $|R|$ of $R$.

\begin{lemma}\label{lem:bgen:evtriv}
  Given as input a "\nice" "\mratm" $\rho: 2^{A^*} \to R$ and a regular language $L\subseteq A^*$, one  may compute the set $\itriv{L,\rho}$ in polynomial time with respect to $|R|$ and the size of some recognizer for $L$ (\emph{i.e.}, an \nfa or a monoid morphism).
\end{lemma}

Furthermore, since $\rho: 2^{A^*} \to R$ is "\nice", we are able to evaluate $\rho(K)$ for any regular language $K$. This amounts to checking whether $K$ intersects regular languages recognized by the canonical word morphism associated to $\rho$. Indeed, since $\rho$ is "\nice", we have $\rho(K) = \sum_{w \in K} \rho(w)$. Therefore, $\rho(K)$ is the sum of all $r \in R$ such that $K \cap \{w \in A^* \mid \rho(w) = r\} \neq \emptyset$. We get the following lemma.

\begin{lemma}\label{lem:bgen:evamramt}
  Given as input a "\nice" "\mratm" $\rho: 2^{A^*} \to R$ and a regular language $K \subseteq A^*$, one may compute the element $\rho(K)$ of $R$ in polynomial time with respect to $|R|$ and the size of some recognizer for $K$ (\emph{i.e.}, an \nfa or a monoid morphism).
\end{lemma}

\subsection{\kl{\Mratms} and \kl{optimal} \kl{\imprints}}

We now present a crucial property of (not necessarily "\nice") "\mratms". It turns out that when the investigated class \Cs is a \pvari of regular languages, given a "\mratm" $\rho: 2^{A^*} \to R$, the structure of the multiplicative semigroup $(R,\cdot)$ is transferred to \emph{\Cs-"optimal" $\rho$-\imprints}. This result is why our framework is meant to be used for classes that are \pvaris of regular languages: it does not hold for arbitrary "lattices".

\begin{lemma} \label{lem:bgen:optsemi}
  Let \Cs be a \pvari of regular languages and let $\rho: 2^{A^*} \to R$ be a "\mratm". Consider two languages $L_1,L_2$. Then, for any $r_1 \in \opti{\Cs}{L_1,\rho}$ and any $r_2 \in \opti{\Cs}{L_2,\rho}$, we have $r_1r_2 \in \opti{\Cs}{L_1L_2,\rho}$.
\end{lemma}

Before proving Lemma~\ref{lem:bgen:optsemi}, let us explain why it is important. Let \Cs be a \pvari and let $\rho: 2^{A^*} \to R$ be a "\mratm". Since $\rho$ is "\tame", it satisfies by definition the following property:
\begin{equation} \label{eq:bgen:optsemiblabla}
  \text{For any $K_1,K_2 \in \Cs$,} \quad \rho(K_1) \cdot \rho(K_2) = \rho(K_1K_2).
\end{equation}
A natural method for building an "optimal" "\Cs-cover" \Kb of some language $L$ for $\rho$ is to start from $\Kb = \emptyset$ and to add new languages $K$ in \Cs to \Kb, until \Kb covers $L$. By definition of $\rho$-\imprints, for \Kb to be "optimal", we need all languages $K$ that we add in \Kb to satisfy $\rho(K) \in \opti{\Cs}{L,\rho}$. It follows from Lemma~\ref{lem:bgen:optsemi} and~\eqref{eq:bgen:optsemiblabla} that when \Cs is a \pvari, we may use concatenation to build new languages $K$. Assume that we have two languages $L_1,L_2$ such that $L_1L_2 \subseteq L$. If we have already built $K_1$ and $K_2$ in \Cs such that $\rho(K_1)\in \opti{\Cs}{L_1,\rho}$ and $\rho(K_2) \in \opti{\Cs}{L_2,\rho}$, then we may add $K_1K_2$ to our "\Cs-cover" of $L$, provided this language belongs to~\Cs, since by Lemma~\ref{lem:bgen:optsemi}, Equation~\eqref{eq:bgen:optsemiblabla} and Fact~\ref{fct:linclus}, we have $\rho(K_1K_2) = \rho(K_1) \cdot \rho(K_2) \in \opti{\Cs}{L_1L_2,\rho} \subseteq \opti{\Cs}{L,\rho}$.

This is central for classes of languages defined through logic (such as "first-order logic"), where language concatenation often plays a central role for building new languages. This means that for such classes, if $K_1$ and $K_2$ belong to~\Cs, then so does $K_1K_2$. Let us now prove Lemma~\ref{lem:bgen:optsemi}.

\begin{proof}[Proof of Lemma~\ref{lem:bgen:optsemi}]
  Let $r_1 \in \opti{\Cs}{L_1,\rho}$ and $r_2 \in \opti{\Cs}{L_2,\rho}$, our objective is to prove that $r_1r_2 \in \opti{\Cs}{L_1L_2,\rho}$. By definition, it suffices to prove that for any "\Cs-cover" \Kb of $L_1L_2$, we have $r_1r_2 \in \prin{\rho}{\Kb}$. Let \Kb be a $\Cs$-cover of $L_1L_2$. Our objective is to find some language $K \in \Kb$ such that $r_1r_2 \leq \rho(K)$. The argument is based on the following claim, which relies on the Myhill-Nerode theorem.

  \begin{claim}
    There exists a language $H \in \Cs$ which satisfies the following two properties:
    \begin{enumerate}
    \item For any $u \in L_1$, there exists $K \in \Kb$ such that $H \subseteq u^{-1}K$.
    \item $r_2 \leq \rho(H)$
    \end{enumerate}
  \end{claim}

  \begin{proof}
    For any $u \in L_1$, consider the set $\Qb_u = \{u^{-1}K \mid K \in \Kb\}$. Clearly, $\Qb_u$ is a "\Cs-cover" of $L_2$ since $\Kb$ is a "cover" of $L_1L_2$ and \Cs is closed under quotients. Moreover, we know by hypothesis on \Cs that all languages in  \Kb are regular. Therefore, it follows from the Myhill-Nerode theorem that they have finitely many left quotients. Thus, while there may be infinitely many $u \in L_1$, there are only finitely many distinct sets $\Qb_u$. It follows that we may use finitely many intersections to build a "\Cs-cover" \Qb of $L_2$ such that for any $Q \in \Qb$ and any $u \in L_1$, there exists $K \in \Kb$ satisfying $Q \subseteq u^{-1}K$. This means that all $Q \in \Qb$ satisfy the first item in the claim, we now pick one which satisfies the second one as~well.

    Since $r_2 \in \opti{\Cs}{L_2,\rho}$, and \Qb is a "\Cs-cover" of $L_2$, we have $r_2 \in \prin{\rho}{\Qb}$. Thus, we get $H \in \Qb$ such that $r_2 \leq \rho(H)$ by definition. This concludes the proof of the claim.
  \end{proof}

  We may now finish the proof of Lemma~\ref{lem:bgen:optsemi}. Let $H \in \Cs$ be defined as in the claim and consider the following set:
  \[
    \Gb = \left\{\bigcap_{v\in H}Kv^{-1} \mid K \in \Kb\right\}.
  \]
  Observe that all languages in \Gb belong to \Cs. Indeed, by hypothesis on \Cs, any $K \in \Kb$ is regular. Thus, it has finitely many right quotients by the Myhill-Nerode theorem and the language $\bigcap_{v\in H}Kv^{-1}$ is the intersection of finitely many quotients of languages in \Cs. By closure under intersection and quotients, it follows that $\bigcap_{v\in H}Kv^{-1} \in \Cs$. Moreover, \Gb is a "\Cs-cover" of $L_1$. Indeed, given $u \in L_1$, we have $K \in \Kb$ such that $H \subseteq u^{-1}K$ by the first item in the claim. Hence, given any $v \in H$, we have $u \in Kv^{-1}$ and we obtain that $u \in \bigcap_{v\in H}Kv^{-1}$, which is an element of \Gb.

  Therefore, since $r_1 \in \opti{\Cs}{L_1,\rho}$ by hypothesis, we have $r_1 \in \prin{\rho}{\Gb}$ and we obtain $G \in \Gb$ such that $r_1 \leq \rho(G)$. Hence, since $r_2 \leq \rho(H)$ by the second item in the claim, we have $r_1r_2 \leq \rho(G) \cdot \rho(H)$. Moreover, since $\rho$ is a "\mratm", it satisfies $ \rho(G) \cdot \rho(H) = \rho(GH)$ by definition, which yields,
  \[
    r_1r_2 \leq \rho(GH).
  \]
  Finally, by definition, $G = \bigcap_{v\in H}Kv^{-1}$ for some $K \in \Kb$. We show that $GH \subseteq K$. Since $\rho$ is increasing by Fact~\ref{fct:bgen:increasing}, this will yield $r_1r_2 \leq \rho(GH) \leq \rho(K)$ which concludes the proof of the lemma. Given $w \in GH$, we have $w = uv$ with $u \in G$ and $v \in H$. Moreover, since $v \in H$, we have $u \in Kv^{-1}$ by definition of~$H$. This exactly says that $w = uv \in K$.
\end{proof}

\subsection{Computing \kl{\nice} \kl{\mratms}}

We now explain why we may replace the finite multiset of regular languages \Lb used in the input pair $(L,\Lb)$ for the "covering problem" by a "\nice" "\mratm". As we already proved in Theorem~\ref{thm:bgen:main}, we may replace \Lb by the associated canonical \ratm $"\rho_\Lb"$: $(L,\Lb)$ is \Cs-coverable if and only if $\Lb \not\in \copti{L,\rho_\Lb}$. However, while $"\rho_\Lb"$ is "\nice", it need not be "\tame". However, we show that from recognizers of the languages in \Lb, we may compute a "\nice" "\mratm" $\rho: 2^{A^*} \to R$ which extends $"\rho_\Lb"$. By Lemma~\ref{lem:bgen:extension}, \copti{L,\rho_\Lb} may be computed from \copti{L,\rho}. Thus, we may replace $"\rho_\Lb"$ by $\rho$. We shall outline this reduction precisely later when we formulate our general approach to "covering". For now, let us explain how the extending "\nice" "\mratm" $\rho$ is built.

\begin{proposition}\label{prop:bgen:mratmeff}
  Given as input a finite multiset of regular languages \Lb, one may compute a "\nice" "\mratm" $\rho: 2^{A^*} \to R$ and a morphism $\delta: R \to 2^\Lb$ such that $\rho$ extends $"\rho_\Lb"$ for the extending morphism $\delta$. Moreover, one can also compute the set $\delta\inv(\Lb)$. Finally, the computation may be achieved in polynomial space with respect to the size of \nfas or monoid morphisms recognizing the languages in \Lb.
\end{proposition}

The remainder of this section is devoted to presenting the construction announced in  Proposition~\ref{prop:bgen:mratmeff}. We show that given as input some finite multiset of regular languages \Lb, one may compute a "\nice" "\mratm" $\rho$ extending \Lb, together with the corresponding extending morphism. We only describe the construction. That is may be achieved in polynomial space is easily verified. It involves two steps:
\begin{enumerate}
\item We first show that for any regular language $L$, one may compute a "\nice" "\mratm" extending $\rho_{\{L\}}$. Let us point out that we actually present \emph{two} constructions: the first one starts from a monoid morphism recognizing $L$ while the other starts from an \nfa.
\item By the first step, we are able to compute a "\nice" "\mratm" extending $\rho_{\{L\}}$ for each $L \in \Lb$. The second step combines all these "\nice" "\mratms" into a single one extending the \ratm $\rho_\Lb$.
\end{enumerate}

\medskip
\noindent
{\bf Canonical "\nice" "\mratm" associated to a morphism.} Consider a finite monoid $M$ and a morphism $\alpha: A^* \to M$. We associate a canonical "\nice" "\mratm" $\rho_\alpha$ to $\alpha$.

Observe that since $M$ is a monoid, the powerset $2^M$ is a "semiring": the addition is union, which makes $2^M$ an idempotent commutative monoid (in particular, the canonical order is inclusion). Moreover, the multiplication is obtained by lifting the one of $M$: given $S,T \in 2^M$, we define $S \cdot T = \{st \mid s \in S \text{ and } t \in T\}$. One may verify that this is indeed a monoid multiplication (its neutral element is $\{1_M\}$) which distributes over union and that $\emptyset$ (the neutral element for union) is a zero for this multiplication. We may now define the canonical "\nice" "\mratm" $\rho_\alpha$ associated to $\alpha$ as follows:
\[
  \begin{array}{llll}
    \rho_\alpha: & 2^{A^*} & \to     & 2^M                        \\
                 & K       & \mapsto & \{\alpha(w) \mid w \in K\}.
  \end{array}
\]
One may verify that $\rho_{\alpha}$ is indeed a "\nice" "\mratm". Moreover, for any language $L$ which is recognized by $\alpha$, the \ratm $\rho_{\alpha}$ extends $\rho_{\{L\}}$. We prove this in the following lemma.

\begin{lemma} \label{lem:bgen:canomorph}
  Let $\alpha: A^* \to M$ be a morphism into a finite monoid. Then, the map $\rho_{\alpha}: 2^{A^*} \to 2^M$ is a "\nice" "\mratm". Moreover, if $L \subseteq A^*$ is a language recognized by $\alpha$ for the accepting set $F \subseteq M$ (\emph{i.e.}, $L = \alpha\inv(F)$), then $\rho_\alpha$ extends $\rho_{\{L\}}$ for the following extending morphism $\delta$:
  \[
    \begin{array}{llll}
      \delta: & 2^M & \to     & 2^{\{L\}}                                             \\
              & S   & \mapsto & \left\{\begin{array}{c                          l}
                                         \{L\} & \text{if $S \cap F \neq \emptyset$} \\
                                         \emptyset & \text{if $S \cap F = \emptyset$}
                                       \end{array}\right.
    \end{array}
  \]
\end{lemma}

\begin{proof}
  That $\rho_{\alpha}$ satisfies the axioms of "\nice" "\mratms" can be verified from the definition. Let us prove that it extends $\rho_{\{L\}}$, \emph{i.e.}, that $\rho_{\{L\}} = \delta \circ \rho_{\alpha}$. Let  $K \subseteq A^*$, we show that $\rho_{\{L\}}(K)= \delta(\rho_{\alpha}(K))$. Observe that since we have $L = \alpha\inv(F)$ and $\rho_{\alpha}(K) = \{\alpha(w) \mid w \in K\}$, it is immediate that,
  \[
    K \cap L \neq \emptyset \quad \text{if and only if} \quad \rho_{\alpha}(K) \cap F \neq \emptyset
  \]
  By definition of $\rho_{\{L\}}$ and $\delta$, we get as desired that $\rho_{\{L\}}(K)= \delta(\rho_{\alpha}(K))$ which concludes the proof.
\end{proof}

\medskip
\noindent
{\bf Canonical "\nice" "\mratm" associated to an \nfa.} Consider an arbitrary \nfa $\As = (A,Q,I,F,\delta)$. We associate a canonical "\nice" "\mratm" $\rho_\As$ to \As.

Observe that the set $2^{Q^2}$ (which consists of sets of pairs of states) is a "semiring". As usual, the addition is union which makes $2^{Q^2}$ and idempotent commutative monoid (in particular, the canonical order is inclusion). Moreover, the multiplication is defined as follows. Given $S,T \in 2^{Q^2}$,
\[
  S \cdot T = \{(q,s) \in Q^2 \mid \text{there exists $r \in Q$ such that $(q,r) \in S$ and $(r,s) \in T$}\}.
\]
One may verify that this is indeed a monoid multiplication (its neutral element the set $\{(q,q) \mid q \in Q\}$) which distributes over union and that $\emptyset$ (the neutral element for union) is a zero for this multiplication. We may now define the canonical "\nice" "\mratm" $\rho_\As$ associated to \As as follows. Recall that given two states $q,r \in Q$ and $w \in A^*$, we write $q \xrightarrow{w} r$ to denote the fact that there exists a run labeled by $w$ from $q$ to $r$ in \As. We define $\rho_\As$ as follows:
\[
  \begin{array}{llll}
    \rho_\As: & 2^{A^*} & \to     & 2^{Q^2}                                                                              \\
              & K       & \mapsto & \{(q,r) \in Q^2 \mid \text{there exists $w \in K$ such that $q \xrightarrow{w} r$}\}
  \end{array}
\]
Note that the definition is independent from the sets $I$ and $F$ of initial and final states of \As. One may verify that $\rho_\As$ is indeed a "\nice" "\mratm". Moreover, the \ratm $\rho_{\{L(\As)\}}$ (where $L(\As)$ is the language recognized by \As) is extended by $\rho_{\As}$. We prove this in the following lemma.

\begin{lemma} \label{lem:bgen:canoauto}
  Let $\As = (A,Q,I,F,\delta)$ be an \nfa. Then the map $\rho_\As: 2^{A^*} \to 2^{Q^2}$ is a "\nice" "\mratm" which extends $\rho_{\{L(\As)\}}$ for the following extending morphism $\gamma$:
  \[
    \begin{array}{llll}
      \gamma: & 2^{Q^2} & \to     & 2^{\{L(\As)\}}                                             \\
              & S   & \mapsto & \left\{\begin{array}{c                          l}
                                         \{L(\As)\} & \text{if $S \cap (I \times F) \neq \emptyset$} \\
                                         \emptyset  & \text{if $S \cap (I \times F) = \emptyset$}
                                       \end{array}\right.
    \end{array}
  \]
\end{lemma}

\begin{proof}
  That $\rho_{\As}$ is a "\nice" "\mratm" can be verified from the definition. Let us prove that it extends $\rho_{\{L(\As)\}}$, \emph{i.e.}, $\rho_{\{L(\As)\}} = \gamma \circ \rho_{\As}$. Let $K\subseteq A^*$, we prove that $\rho_{\{L(\As)\}}(K) = \delta(\rho_{\As}(K))$. By definition, $L(\As)$ is the set of all words labeling a run between a state $q \in I$ and astate $r \in F$. Hence, by definition of $\rho_{\As}(K)$, it is immediate that:
  \[
    K \cap L(\As) \neq \emptyset \quad \text{if and only if} \quad \rho_{\As}(K) \cap (I \times F) \neq \emptyset.
  \]
  By definition of $\rho_{\{L(\As)\}}$ and $\gamma$, we get as desired that $\rho_{\{L(\As)\}}(K)= \delta(\rho_{\As}(K))$ which concludes the proof.
\end{proof}

\medskip
\noindent
{\bf Extending a finite multiset of regular languages.} We may now present the general construction. Let us consider an input multiset of regular languages $\Lb = \{L_1,\dots,L_n\}$. We want to compute a "\nice" "\mratm" $\rho$ extending $\rho_\Lb$. Using the above constructions, we know that for all $i \leq n$, we are able to build a "\nice" "\mratm" $\rho_i: 2^{A^*} \to R_i$ extending $\rho_{\{L_i\}}$ and some extending morphism $\delta_i: R_i \to 2^{\{L_i\}}$.

One may verify that the Cartesian product $R = R_1 \times \cdots \times R_n$ is also an idempotent "semiring" for the componentwise addition and multiplication. Consider the following map $\rho$:
\[
  \begin{array}{llll}
    \rho: & 2^{A^*} & \to     & R                           \\
          & K       & \mapsto & (\rho_1(K),\dots,\rho_n(K)).
  \end{array}
\]
It is straightforward to check that $\rho$ is a "\nice" "\mratm" as well. Moreover, it is clear that $\rho$ may be computed from $\rho_1,\dots,\rho_n$. Finally, $\rho$ extends $\rho_\Lb$ for the following extending~morphism:
\[
  \begin{array}{llll}
    \delta: & R               & \to     & 2^\Lb                                        \\
            & (r_1,\dots,r_n) & \mapsto & \delta_1(r_1) \cup \cdots \cup \delta_n(r_n).
  \end{array}
\]
Clearly, one may compute $\delta$ from $\delta_1,\dots,\delta_n$. This concludes the presentation of our construction. Observe that by construction, we end up with a "rating set" whose size is as described in \figurename~\ref{fig:bgen:consextend} (where we write $|\As|$ for the number of states in the \nfa \As).

\begin{figure}[!htb]
  \begin{center}
    \begin{tikzpicture}
      \matrix (M) [matrix of nodes, column  sep=5mm,row  sep=4mm,draw,very thick,rounded corners=3pt,nodes={align=center,text width = 4cm,anchor=center}]
      {
        {} & \Lb is given by $n$ \nfas $\As_1,\dots,\As_n$ & \Lb is given by $n$ monoids $M_1,\dots,M_n$ \\
        Size of the rating set & $2^{|\As_1|^2 + \cdots + |\As_n|^2}$ & $2^{|M_1| + \cdots + |M_n|}$ \\
      };

      \foreach \col in {2,3} {        \mvline[thick]{M}{\col}
      }

      \foreach \row in {2} {        \mhline[thick]{M}{\row}
      }
    \end{tikzpicture}
  \end{center}
  \caption{Size of the "rating set" for a "\nice" "\mratm" extending the canonical \ratm $"\rho_\Lb"$ associated to a multiset $\Lb = \{L_1,\dots,L_n\}$ of regular languages.}
  \label{fig:bgen:consextend}
\end{figure}

\section{General approach for universal covering}
\label{sec:genba}
We may now start outlining our general methodology for tackling "\Cs-covering". In this section, we start with the special case of \emph{universal} "\Cs-covering", where the language that needs to be covered is~$A^*$ (\emph{i.e.}, the input is of the form $(A^*,\Lb)$). Recall that by Proposition~\ref{prop:bgen:eqpoint}, this restricted problem is equivalent to full "\Cs-covering" when the investigated class \Cs is a "Boolean algebra".

Of course, this methodology will later be subsumed by the one that we shall present for the full "\Cs-covering problem" in Section~\ref{sec:genlatts}. However, it makes sense to introduce a specialized approach for universal "\Cs-covering": this is simpler and "Boolean algebras" account for most of the relevant classes.

\begin{remark}
  Since universal and full "\Cs-covering" are only equivalent when \Cs is a "Boolean algebra", the methodology that we outline here is only meant to be used for such classes. However,  it makes sense for any class \Cs that is a \pvari of regular languages.
\end{remark}

We first recall the notions introduced in the two previous sections and use them to formally reduce universal "\Cs-covering" to another decision problem whose input is a  \nice \mratm. Since the language that needs to be covered in universal covering is always $A^*$, we are able to slightly simplify our terminology. Then, we describe our methodology for solving this new decision problem.

\subsection{\kl{Optimal} universal \kl{\imprints}}

We start by simplifying our terminology on "\ratms" to accommodate "universal \Cs-covering". Specifically, since the language that we want to cover will always be $A^*$, we may omit this parameter when speaking of "\imprints".

Given a "\ratm" $\rho: 2^{A^*} \to R$ and a "lattice" \Cs, we shall say \emph{\Cs-"optimal" universal $\rho$-"\imprint"}, to mean the set \opti{\Cs}{A^*,\rho} (\emph{i.e.}, the \Cs-"optimal" $\rho$-\imprint on $A^*$). Moreover, we simply write \AP""\opti{\Cs}{\rho}"" for this set, omitting the parameter $A^*$  (\emph{i.e.}, $\opti{\Cs}{\rho} = \opti{\Cs}{A^*,\rho}$). Recall that by definition, "\opti{\Cs}{\rho}" is the $\rho$-\imprint of any "optimal" "universal \Cs-cover" for $\rho$.

We first present a few properties of this new object and then formally reduce universal "\Cs-covering" to the problem of computing \opti{\Cs}{\rho} from an input "\nice" "\mratm".

\medskip
\noindent
{\bf Properties.} It turns out that $"\opti{\Cs}{\rho}"= \opti{\Cs}{A^*,\rho}$ has stronger properties than \opti{\Cs}{L,\rho} when $L$ is arbitrary. We present them now. First, the set of trivial elements of "\opti{\Cs}{\rho}" is simpler to describe. Given any "\ratm" $\rho: 2^{A^*} \to R$, we define,
\[
  \itriv{\rho} = \itriv{A^*,\rho}  = \dclos \{\rho(w) \mid w \in A^*\} \subseteq R.
\]
Since "\opti{\Cs}{\rho}" is the $\rho$-\imprint of some universal \Cs-cover (an "optimal" one for $\rho$), the following result is immediate from  Fact~\ref{fct:bgen:trivialsets}.

\begin{fct}\label{fct:bgen:utrivialsets}
  Let $\rho: 2^{A^*} \to R$ be a "\ratm" and let \Cs be a "lattice". Then, $\itriv{\rho} \subseteq "\copti{\rho}"$.
\end{fct}

More importantly, we have the following corollary of Lemma~\ref{lem:bgen:optsemi}. When \Cs is a \pvari of regular languages and $\rho: 2^{A^*} \to R$ is a "\mratm", "\copti{\rho}" is a submonoid of $R$ for multiplication. This property is crucial: all algorithms for universal "\Cs-covering" which are based on our methodology exploit it.

\begin{lemma}\label{lem:boptsemi}
  Let \Cs be a \pvari of regular languages and let $\rho: 2^{A^*} \to R$ be a "\mratm". Then "\copti{\rho}" is a submonoid of $R$ for multiplication:
  \begin{itemize}
  \item $1_R \in "\copti{\rho}"$.
  \item For any $s,t \in "\copti{\rho}"$, we have $st \in "\copti{\rho}"$.
  \end{itemize}
\end{lemma}

\begin{proof}
  That $1_R \in "\opti{\Cs}{\rho}"$ is immediate from Fact~\ref{fct:bgen:utrivialsets}. Indeed, the fact yields that $\itriv{\rho} \subseteq "\opti{\Cs}{\rho}"$. Moreover, since $\rho$ is a \mratm, we have $1_R = \rho(\varepsilon)$ and $\rho(\varepsilon) \in \itriv{\rho}$ by definition. Closure under multiplication is immediate from Lemma~\ref{lem:bgen:optsemi}. Indeed, assume that $s,t \in "\opti{\Cs}{\rho}"$. Since $"\opti{\Cs}{\rho}" = \opti{\Cs}{A^*,\rho}$ by definition, we get from Lemma~\ref{lem:bgen:optsemi} that $st \in \opti{\Cs}{A^*A^*,\rho}$. Finally, since $A^*A^* = A^*$, this yields $st\in "\opti{\Cs}{\rho}"$.
\end{proof}

\medskip
\noindent
{\bf Reduction.} We may now reduce universal "\Cs-covering" to another decision problem whose input is a \nice \mratm. We do so in the following proposition.

\begin{proposition}\label{prop:breduc}
  Let \Cs be a lattice. There is a polynomial space reduction from universal "\Cs-covering" (for input languages given by \nfas or monoid morphisms) to the following decision problem:
  \begin{center}
    \begin{tabular}{ll}
      {\bf Input:}    & A \nice \mratm $\rho: 2^{A^*} \to R$ and a subset $F \subseteq R$. \\
      {\bf Question:} & Do we have $F \cap "\copti{\rho}" = \emptyset$?
    \end{tabular}
  \end{center}
\end{proposition}

\begin{proof}
  The input of universal "\Cs-covering" is a finite multiset of regular languages \Lb: we want to know whether \Lb is \Cs-coverable (\emph{i.e.}, whether the pair $(A^*,\Lb)$ is \Cs-coverable). By Proposition~\ref{prop:bgen:mratmeff}, we may compute in polynomial space a \nice \mratm $\rho: 2^{A^*} \to R$ and an extending morphism $\delta: R \to 2^\Lb$ such that $\rho$ extends the canonical \ratm $"\rho_\Lb"$ associated to \Lb, for the extending morphism $\delta$. Moreover, one can compute, also in polynomial space, the set $F = \delta\inv(\Lb)$. We show that \Lb is \Cs-coverable if and only if $F \cap "\copti{\rho}" = \emptyset$, which will prove Proposition~\ref{prop:breduc}.

  By Theorem~\ref{thm:bgen:main}, \Lb is \Cs-coverable if and only if $\Lb \not\in \copti{\rho_{\Lb}}$. Moreover, by Lemma~\ref{lem:bgen:extension}, we know that $\copti{\rho_{\Lb}} = \dclos \delta("\copti{\rho}")$. Consequently, since $F = \delta\inv(\Lb)$ and \Lb is the maximal element of~$2^\Lb$, we have $\Lb \not\in \copti{\rho_{\Lb}}$ if and only if $F \cap "\copti{\rho}" = \emptyset$, which concludes the proof.
\end{proof}

\begin{remark}\label{rem:computeuniv}
  The reduction of Proposition~\ref{prop:breduc} applies to the first stage of universal "\Cs-covering": designing an algorithm that decides it. However, we also obtain a reduction for the second stage: computing "separating" "universal \Cs-covers", when they exist. Indeed, given some input multiset of regular languages \Lb, we know from Theorem~\ref{thm:bgen:main} that if\/ \Lb is \Cs-coverable, then any "optimal" "universal \Cs-cover" for $"\rho_\Lb"$ is "separating" for \Lb. We may compute a "\nice" "\mratm" $\rho$ extending $\rho_\Lb$ and Lemma~\ref{lem:bgen:extension} states that any "optimal" "universal \Cs-cover" for $\rho$ is also "optimal" for $"\rho_\Lb"$ (and thus, "separating" for \Lb).

  This of particular interest. Indeed, it turns out that in most cases, the correction proofs for algorithms solving the problem presented in Proposition~\ref{prop:breduc} involve describing a generic construction for building "optimal" "universal \Cs-covers". This is the case for all examples that we present.
\end{remark}

Our methodology is designed for handling the decision problem of Proposition~\ref{prop:breduc}. Given a \pvari of regular languages \Cs, we look for an algorithm  computing "\copti{\rho}" from an input "\nice" "\mratm" $\rho$. This clearly yields a procedure for problem of Proposition~\ref{prop:breduc}.

\begin{remark}
  While Proposition~\ref{prop:breduc} holds for any lattice, the methodology requires at least a \pvari of regular languages (this is necessary for applying Lemma~\ref{lem:boptsemi}).
\end{remark}

\subsection{Methodology}

Given a \pvari of regular languages \Cs, our main objective is to compute the \Cs-"optimal" universal $\rho$-"\imprint" "\opti{\Cs}{\rho}" from an input "\nice" "\mratm" $\rho: 2^{A^*}\to R$. A key design principle behind our framework is that our algorithms for computing "optimal" universal "\imprints" are formulated as \emph{elegant characterization theorems}. More precisely, we characterize the set $"\opti{\Cs}{\rho}" \subseteq R$ as the least subset of $R$ such that:
\begin{enumerate}
\item it includes the trivial elements from the set $\itriv{\rho}$, and
\item it is closed under a list of operations.
\end{enumerate}
We speak of a \emph{characterization of \Cs-"optimal" universal "\imprints"}. In practice, such a characterization yields a least fixpoint procedure for computing "\opti{\Cs}{\rho}" from $\rho$: one starts from the set of trivial elements and saturates it with the closure operations in the list (of course, this depends on our ability to implement these operations, but this is straightforward in practice).

Let us present a few examples. All of them are fragments of "\fo". The first fragment is "first-order logic" itself. The characterization that we present is directly adapted from the "separation" algorithm of~\cite{pzfo,pzfoj} and the proof uses essentially the same arguments (see Proposition~4.7 in~\cite{pzfoj}).

\begin{example}[Characterization of "\fo"-"optimal" \imprints]\label{ex:bgen:fo} Consider a {\bf "\nice"} "\mratm" $\rho: 2^{A^*} \to R$ (note the requirement of being  "\nice"). One may show that \foopti is the least subset of~$R$ containing \itriv{\rho} and satisfying the following properties:
  \begin{enumerate}
  \item {\bfseries Downset}: For any $r \in \foopti$ and any $r  \leq r$, we have $r' \in \foopti$.
  \item {\bfseries Multiplication}: For any $s,t \in \foopti$, we have $st \in \foopti$.
  \item {\bfseries "\fo"-closure}: For any $s \in \foopti$, we have $s^\omega + s^{\omega +1} \in \foopti$.
  \end{enumerate}
\end{example}

\begin{remark}
	Let us point out that~\cite{pzfoj} predates the current paper. Consequently, the formulation of Example~\ref{ex:bgen:fo} differs from the one of Proposition~4.7 in~\cite{pzfoj}. The latter does not use the abstract notion of \mratm. Instead, it considers a ``concrete'' object built from a monoid morphism recognizing the two input languages in the separation problem. With respect to the new terminology introduced here, this object corresponds to a \mratm whose rating set is the powerset of some finite monoid.
\end{remark}

\begin{remark}
  The characterization of Example~\ref{ex:bgen:fo} is restricted to \textbf{"\nice"} "\mratms". Of course, this suffices get an algorithm for (universal) \fow-"covering". In fact, algorithms only make sense for this special case as we are not able to finitely represent arbitrary "\ratms". However, it turns out that in most cases, the characterizations themselves are independent from the hypothesis of being "\nice" (the above one is among the few exceptions). This is the case for our other examples.
\end{remark}

Our second example is \bscu, level~1 in the quantifier alternation hierarchy of \fo. The characterization is loosely inspired from the "separation" algorithm of~\cite{pvzmfcs13}. We detail it in Section~\ref{sec:bsigma}.

\begin{example}[Characterization of \bscu-"optimal" \imprints]\label{ex:bscu}
  Recall that for any $B \subseteq A$, we denote by $\fullcont{B}$ the language of words whose alphabet is $B$. Consider a \mratm $\rho: 2^{A^*} \to R$. Then, \bsuopti is the least subset of $R$ containing \itriv{\rho} and satisfying the following properties:
  \begin{enumerate}
  \item {\bfseries Downset}: For any $r \in \bsuopti$ and any $r ' \leq r$, we have $r' \in \bsuopti$.
  \item {\bfseries Multiplication}: For any $s,t \in \bsuopti$, we have $st \in \bsuopti$.
  \item {\bfseries \bscu-operation}: For any $B \subseteq A$, we have $(\rho(\fullcont{B}))^\omega \in \bsuopti$.
  \end{enumerate}
\end{example}

We shall detail a third example in Section~\ref{sec:fo2}: the two-variable restriction of "first-order logic",~\fod.
It is apparent in the two above examples that the characterization of \Cs-"optimal" universal "\imprints" involves two kinds of properties: those generic to all \pvaris and those specific to the one under investigation. Let us conclude the section by listing the generic properties, those satisfied by all \Cs-optimal universal \imprints regardless of the \pvari of regular languages~\Cs.

\begin{lemma}\label{lem:bgen:genalgoba}
  Consider a \mratm $\rho: 2^{A^*} \to R$ and some \pvari of regular languages \Cs. Then, the \Cs-optimal universal $\rho$-\imprint $\opti{\Cs}{\rho} \subseteq R$ contains \itriv{\rho} and satisfies the following closure properties:
  \begin{enumerate}
  \item {\bfseries Downset}: For any $r \in \opti{\Cs}{\rho}$ and any $r ' \leq r$, we have $r' \in \opti{\Cs}{\rho}$.
  \item {\bfseries Multiplication}: For any $s,t \in \opti{\Cs}{\rho}$, we have $st\in \opti{\Cs}{\rho}$.
  \end{enumerate}
\end{lemma}

\begin{proof}
  That \opti{\Cs}{\rho} contains \itriv{\rho} is immediate from Fact~\ref{fct:bgen:utrivialsets}. That \opti{\Cs}{\rho} is closed under downset follows from Fact~\ref{fct:bgen:downset}. Finally closure under multiplication is stated in Lemma~\ref{lem:boptsemi} (this is where we need the hypothesis that \Cs is a \pvari of regular languages).
\end{proof}

\section{\texorpdfstring{Example for universal covering: the logic \bscu}{Example for universal covering: the logic BΣ1}}
\label{sec:bsigma}
In this section, we illustrate our framework by presenting a detailed proof for Example~\ref{ex:bscu}: we present an algorithm for (universal) \bscu-covering using the methodology outlined in the previous section. We actually work with an alternate definition of the class corresponding to \bscu which will be simpler to manipulate: the class of \emph{piecewise testable languages}. We first recall the definition.

\subsection{Piecewise testable languages}

Given two words $u,v \in A^*$, we say that $u$ is a \emph{piece} (or scattered subword) of $v$ if there exist $n \in \nat$, $a_1,\dots,a_n \in A$ and $v_0,\dots,v_n \in A^*$ such that,
\[
  u = a_1 \cdots a_n \quad \text{and} \quad v = v_0a_1v_1 \cdots a_n v_n.
\]
We use pieces to define equivalence relations over $A^*$. Given $k \in \nat$ and $w,w' \in A^*$, we write $w\sim_k w'$ if and only if $w$ and $w'$ contain the same pieces of length at most $k$. Clearly, $\sim_k$ is an equivalence relation of finite index over $A^*$. Given a language $L$ and some integer $k \in \nat$, we say that $L$ is \emph{$k$-piecewise testable} and we write $L \in \kpt$ if $L$ is a union of $\sim_k$-classes. Finally, $L$ is \emph{piecewise testable} if $L \in \kpt$ for some $k$. The following theorem is folklore and simple to establish.

\begin{theorem} \label{thm:bscupt}
  Let $L \subseteq A^*$ be a language. Then, $L$ can be defined by a \bscu sentence if and only if $L$ is piecewise testable.
\end{theorem}

\subsection{Characterization of optimal \imprints} Let us now recall the characterization of \bscu-"optimal" universal "\imprints" from of Example~\ref{ex:bscu}: given a "\mratm" $\rho$, we describe \ptopti.

\begin{remark}
  The statement does not require $\rho$ to be \nice: it holds for any \mratm.
\end{remark}

Recall that for any $B \subseteq A$, we denote by $\fullcont{B} \subseteq B^*$ the set $\fullcont{B} = \{w \in A^* \mid \cont{w} = B\}$. Consider a "\mratm" $\rho: 2^{A^*} \to R$ and a subset $S \subseteq R$. Let $\omega$ denote the idempotent power for the multiplication of the semiring $R$. We say that $S$ is \emph{\bscu-saturated for $\rho$} if it contains \itriv{\rho} and is closed under the following operations:
\begin{enumerate}
\item \emph{\bfseries Downset:} for any $r \in S$ and any $r' \in R$ such that $r' \leq r$, we have $r' \in r$.
\item \emph{\bfseries Multiplication:} For any $s,t \in S$, we have $st \in S$.
\item\label{op:bs1:main} \emph{\bfseries \bscu-operation:} For all $B \subseteq A$, we have $(\rho(\fullcont{B}))^\omega \in S$.
\end{enumerate}

\noindent Note that the operation specific to \bscu is not a closure operator, contrary to the cases of \fo and~\fod.

\smallskip\noindent We now state the main theorem of the section: for any "\mratm" $\rho: 2^{A^*} \to R$, the \bscu-"optimal" universal $\rho$-"\imprint" \ptopti is the least \bscu-saturated subset of $R$ (for inclusion).

\begin{theorem}[Characterization of \bscu-optimal \imprints] \label{thm:bs1:main}
  Let $\rho: 2^{A^*} \to R$ be a \mratm. Then, $\ptopti \subseteq R$ is the least \bscu-saturated subset of $R$.
\end{theorem}

Clearly, Theorem~\ref{thm:bs1:main} yields an algorithm for computing \ptopti from an input "\nice" "\mratm" $\rho: 2^{A^*} \to R$. Indeed, by definition, one may compute the least \bscu-saturated subset of $R$ with a least fixpoint procedure: one starts from the set $\itriv{\rho}$ and saturates it with the three operations in the definition (it is clear that one may implement all of them). Hence, Proposition~\ref{prop:breduc} yields that "universal \bscu-covering" is decidable. Since \bscu is a "Boolean algebra", this is also true for full \bscu-"covering" by Proposition~\ref{prop:bgen:eqpoint}. Altogether, we obtain the following corollary.

\begin{corollary} \label{cor:bs1:main}
  The \bscu-covering problem is decidable.
\end{corollary}

This concludes the presentation of our characterization. We now turn to its proof. We fix an arbitrary \mratm $\rho: 2^{A^*} \to R$ for the remainder of the section. Our objective is to show that \ptopti is the least \bscu-saturated subset of $R$. Our proof argument follows two steps that we shall use for all examples presented in the paper:
\begin{enumerate}
\item  First, we show that \ptopti is \bscu-saturated. This corresponds to soundness of the least fixpoint algorithm: all elements it computes belong to \ptopti.
\item Then, we show that \ptopti is smaller than all \bscu-saturated sets. This corresponds to completeness of the least fixpoint algorithm: it computes all elements of \ptopti.
\end{enumerate}

\subsection{Soundness}

We prove that \ptopti is \bscu-saturated.  Since \bscu is a \vari of regular languages, we already know from Lemma~\ref{lem:bgen:genalgoba} that \ptopti contains \itriv{\rho} and is closed under downset and multiplication. Thus, we may focus on \emph{\bscu-operation}. Given $B \subseteq A$, we have to show that $(\rho(\fullcont{B}))^\omega \in \ptopti$.

\medskip

By definition, this amounts to proving that for any universal \bscu-cover \Kb, we have $(\rho(\fullcont{B}))^\omega \in \prin{\rho}{\Kb}$. Since \Kb is a \bscu-cover, it is immediate from Theorem~\ref{thm:bs1:main} that there exists $k \in \nat$ such that all $K\in \Kb$ belong to \kpt. Consider the following language $H$:
\[
  H = (\fullcont{B})^{k\omega} \quad \text{where $\omega \geq 1$ is the idempotent power for the multiplication of $R$}
\]
Clearly, we have $\rho(H) = (\rho(\fullcont{B}))^\omega$. Moreover, since \Kb is a "universal cover", there exists $K \in \Kb$ such that $K \cap H \neq \emptyset$. We show that $H \subseteq K$. It will follow that $(\rho(\fullcont{B}))^\omega = \rho(H) \leq \rho(K)$ which implies that $(\rho(\fullcont{B}))^\omega \in \prin{\rho}{\Kb}$ by definition. By hypothesis, $K \in \kpt$ which means that $K$ is a union of $\sim_k$-classes. Hence, since $H \cap K \neq \emptyset$, it suffices to show that all words in $H$ are $\sim_k$ equivalent to obtain $H \subseteq K$. This is what we do.

Let $w,w' \in H$, we prove that $w \sim_k w'$, \emph{i.e.}, $w$ and $w'$ contain the same pieces of length at most $k$. By symmetry, we consider $v \in A^*$ of length $|v| \leq k$ which is a piece of $w$ and show that $v$ is a piece of $w'$. By definition $H \subseteq B^*$ which yields $w \in B^*$ and $v \in B^*$. Moreover, by definition of $H$, $w' \in H$ implies that $w'$ admits a decomposition $w' = w'_1 \cdots w'_{k}$ such that $\cont{w'_{i}} = B$ for all $i \leq k$. Since $|v| \leq k$ and $v \in B^*$, this implies that $v$ is a piece of $w'$, finishing the proof.

\subsection{Completeness}

We turn to the difficult direction. Recall that a \mratm $\rho: 2^{A^*} \to R$ is fixed. We fix an arbitrary \bscu-saturated set $S \subseteq R$, and show that $\ptopti \subseteq S$. The proof is a generic construction, which builds a universal \bscu-cover \Kb such that $\prin{\rho}{\Kb} \subseteq S$. Since $\ptopti \subseteq \prin{\rho}{\Kb}$ for any universal \bscu-cover \Kb, this proves $\ptopti \subseteq \prin{\rho}{\Kb} \subseteq S$ as desired.

\begin{remark}
  Since we showed that \ptopti is \bscu-saturated, one may apply our construction in the special case when $S = \ptopti$. This builds a universal \bscu-cover \Kb such that $\prin{\rho}{\Kb} = \ptopti$, \emph{i.e.}, one that is optimal for $\rho$.
\end{remark}

The languages contained in the universal \bscu-cover \Kb that we build are defined according to a new notion called \emph{template}, which we first define.

\newcommand{\unit}{unit\xspace}
\newcommand{\units}{units\xspace}

\medskip
\noindent
{\bf Templates.} We call \emph{\unit} an element $t$ which either a single letter $a \in A$ or a triple $(b,B,c)$ where $B \subseteq A$ is a sub-alphabet and $b,c \in B$. Given $\ell \in \nat$, a \emph{template of length $\ell$} is a tuple $T = (t_1,\dots,t_\ell)$ (empty when $\ell = 0$) where every $t_i$ is a \unit.

Moreover, we say that a template $T = (t_1,\dots,t_\ell)$ is \emph{unambiguous} if all pairs of consecutive \units $t_i,t_{i+1}$ in the template are either:
\begin{enumerate}
\item Two single letters or,
\item A single letter $a$ and a triple $(b,B,c)$ such that $a \not\in B$ or,
\item Two triples $(b_i,B_{i},c_i)$ and $(b_{i+1},B_{i+1},c_{i+1})$ such that $c_i \not\in B_{i+1}$ and $b_{i+1} \not\in B_i$.
\end{enumerate}

\begin{example}
  $T = (a,(b,\{b,c\},c),d,(a,\{a\},a))$ is an unambiguous template of length $4$. $T'=({\bf b},(c,\{{\bf b},c\},b),d,a)$ and $T'' = (a,(b,\{b,{\bf c}\},b),({\bf c},\{c\},c),a)$ are ambiguous templates of length $4$.
\end{example}

To any template $T$ and any natural number $n \geq 1$, we associate a language $K_{n,T}$ (definable in \bscu when $T$ is unambiguous). These are the languages that we shall use in our universal \bscu-cover. When $T$ is empty, we let $K_{n,T} = \{\varepsilon\}$ (the definition does not depend on $n$ in this case). Otherwise, we first consider \units. Let $t$ be a \unit. There are two cases:
\begin{itemize}
\item If $t = a \in A$, then $K_{n,t}  = \{a\}$ (again, the definition does not depend on $n$).
\item If $t = (b,B,c)$ with $B \subseteq A$ and $b,c \in B$, then $K_{n,t}  = B^* b (\fullcont{B})^n c B^*$.
\end{itemize}
Finally, if $T = (t_1,\dots,t_\ell)$ is a template of length $\ell \geq 1$ we define,
\[
  K_{n,T} = K_{n,t_1} \cdots K_{n,t_\ell}
\]
We now show that when $T$ is unambiguous, the language $K_{n,T}$ may be defined in \bscu. Hence, we may use it in our universal \bscu-cover \Kb.

\begin{lemma} \label{lem:ispiecewise}
  For any $n \geq 1$ and any unambiguous template $T$, $K_{n,T}$ is definable in \bscu.
\end{lemma}

\begin{proof}
  When $T$ is empty, $K_{n,T} = \{\varepsilon\}$ is clearly definable in \bscu. Otherwise, let $\ell \geq 1$ be the length of $T$. Using induction on $\ell$, we show that the following implication holds for $k = n|A|+2\ell$:
  \begin{equation} \label{eq:piecewise}
    \text{For any $w,w' \in A^*$,} \quad w \sim_k w' \Rightarrow \text{$w \in K_{n,T}$ if and only if $w' \in K_{n,T}$}.
  \end{equation}
  This exactly says that $K_{n,T}$ is a union of $\sim_k$-classes. Consequently, we get $K_{n,T} \in \kpt$ which yields that $K_{n,T}$ is definable in \bscu by Theorem~\ref{thm:bscupt}. It remains to prove~\eqref{eq:piecewise} by induction on~$\ell$.

  \smallskip

  In the base case, $\ell = 1$. Hence, $T$ is a single \unit $t$ and $k = n|A|+2$. Consider $w,w' \in A^*$ such that $w \sim_k w'$. Assuming that $w \in K_{n,t}$, we show that $w' \in K_{n,t}$ (the converse implication is symmetrical). There are two subcases depending on $t$. If $t = a \in A$, then $K_{n,t} = \{a\}$ and $w=a$. Since $k \geq 2$, $w \sim_k w'$ implies $w' = a \in K_{n,T}$. Otherwise, $t = (b,B,c)$ with $B \subseteq A$ and $b,c \in B$. In that case, $K_{n,t}  = B^* b (\fullcont{B})^n c B^*$. Since $w \in K_{n,t}$, it is simple to verify that $w \in B^*$ and there exists $v \in (\fullcont{B})^n$ of length $n|B|$ such that $bvc$ is a piece of $w$. Hence, since $k \geq n|B|+2$ and $w \sim_k w'$, we get that $w' \in B^*$ and $bvc$ is a piece of $w'$ as well. This implies $w' \in K_{n,T}$, as desired.

  \smallskip

  We now assume that $\ell \geq 1$ and let $T = (t_1,\dots,t_\ell)$. Recall that $k = n|A|+2\ell$. Consider $w,w' \in A^*$ such that $w \sim_k w'$. Assuming that $w \in K_{n,T}$, we show that $w' \in K_{n,T}$ (again, the converse implication is symmetrical). We distinguish two subcases depending on the first \unit $t_1$ of $T$. Since both subcases are similar, we treat the one when $t_{1}$ is a triple $(b,B,c)$ with $B \subseteq A$ and $b,c \in B$ (the case when $t_1 = a \in A$ is simpler and left to the reader). We need to prove that
  \[
    w' \in K_{n,T} = K_{n,t_1} \cdots K_{n,t_\ell}.
  \]
  We first show that $w'$ has a prefix in $K_{n,t_1} = B^* b (\fullcont{B})^n c B^*$ using the hypothesis that $T$ is unambiguous.

  For all $i \leq \ell$, we define $u_i \in A^*$ as the following word. If $t_i = a_i \in A$, then $u_i = a_i$ and if $t_i = (b_i,B_i,c_i)$, then $u_i \in b_ic_i$. Recall that $w \in K_{n,T} = K_{n,t_1} \cdots K_{n,t_\ell}$. It is simple to verify from this hypothesis that there exists $v \in (\fullcont{B})^n$ of length $|v| = n|B|$ such that $bvcu_2 \cdots u_\ell$ is a piece of $w$. Since $bvcu_2 \cdots u_\ell$ has length at most $n|B|+2 + 2(\ell-1) \leq k$ by definition of $k$, and since $w \sim_k w'$, it follows that $bvcu_2 \cdots u_\ell$ is a piece of $w'$ as well. Consequently $w'$ admits a decomposition $w' = x'cy'$ such that $bv$ is a piece of $x'$ and $u_2 \cdots u_\ell$ is a piece of $y'$. We show that $x'c \in B^*$ which proves as desired that the prefix $x'c$ of $w'$ belongs to $B^* b (\fullcont{B})^n c B^*$. By contradiction assume that $x'c$ contains some letter $d \not\in B$. Since $c \in B$, it follows that $dcu_2 \cdots u_\ell$ is a piece of $w'$. Hence, using again the hypothesis that $w \sim_k w'$, this implies that $dcu_2 \cdots u_\ell$ is a piece of $w$ as well. Since $T$ is unambiguous, one may verify that this contradicts the hypothesis that $w \in K_{n,T}$: we have shown that $w'$ has a prefix in $K_{n,t_1}$.

  We now finish the proof that $w' \in K_{n,T}$. Let $x'$ be the largest prefix of $w'$ belonging to~$B^*$. Since we just showed that $w'$ has a prefix in $B^* b (\fullcont{B})^n c B^* \subseteq B^*$, it is immediate that $x' \in B^* b (\fullcont{B})^n c B^* = K_{n,t_1}$. Moreover, $w' = x' d' y'$ with $d' \in A \setminus B$ and $y' \in A^*$. It remains to show that $d'y' \in K_{n,t_2} \cdots K_{n,t_\ell}$. We use induction. Since $w \sim_k w'$, $w$ also contains letters outside $B$ and it follows that $w = x d y$ with $x \in B^*$, $d \in A \setminus B$ and $y \in A^*$. By hypothesis, $w \in K_{n,T} =  K_{n,t_1} \cdots K_{n,t_\ell}$. Since, $T$ is unambiguous $x \in B^*$ and $d \in A \setminus B$, one may verify that this implies $dy \in K_{n,t_2} \cdots K_{n,t_\ell}$. Moreover, since $w \sim_k w'$, $w = xdy$ and $w' = x'd'y'$ with $x,x' \in B^*$ and $d,d' \not\in B$, it is simple to verify that $dy \sim_{k-1} d'y'$. In particular, $dy \sim_{k-2} d'y'$, and $k-2=n|A|+2(\ell-1)$. As the template $(t_2,\dots,t_\ell)$ is clearly unambiguous, the induction yields that $d'y' \in K_{n,t_2} \cdots K_{n,t_\ell}$. Altogether, this shows that $w' \in K_{n,T}$.
\end{proof}

\smallskip
\noindent
{\bf Main construction.} We now are ready to build a universal \bscu-cover \Kb satisfying $\prin{\rho}{\Kb} \subseteq S$ (where $S$ is any fixed \bscu-saturated set) to finish the proof of Theorem~\ref{thm:bs1:main}. As announced, \Kb consists only of languages chosen with the two following propositions. Recall that $R$ is the "rating set" of~$\rho$.

\begin{proposition} \label{prop:bs1:imprint}
  Let $n \geq |R|$ and $T$ a template. Then, $\rho(K_{n,T}) \in S$.
\end{proposition}

\begin{proposition} \label{prop:bs1:cover}
  Let $n \geq 1$. Then, for any word $w \in A^*$ there exists an unambiguous template $T$ of length at most $(n+2)^{|A|}-1$ such that $w \in K_{n,T}$.
\end{proposition}

Before proving these propositions, let us use them to build \Kb and finish the completeness proof. We fix $n = |R|$ and $\ell = (n+2)^{|A|}-1$. We then define \Kb as the set of all languages $K_{n,T}$ where $T$ is an \emph{unambiguous} template of length at most $\ell$. By definition, \Kb is finite since there are only finitely many unambiguous templates of length at most $\ell$. Moreover, any $K \in \Kb$ is definable in \bscu by Lemma~\ref{lem:ispiecewise}. It remains to show that \Kb is a universal cover and that $\prin{\rho}{\Kb} \subseteq S$.

We first show that \Kb is a universal cover: given $w \in A^*$, we exhibit a language in \Kb containing $w$. By Proposition~\ref{prop:bs1:cover}, $w \in K_{n,T}$ where $T$ is an unambiguous template of length at most $\ell = (n+2)^{|A|}-1$. It is then immediate that $K_{n,T} \in \Kb$ by definition. It only remains to show that $\prin{\rho}{\Kb} \subseteq S$. Let $r \in \prin{\rho}{\Kb}$. By definition of \Kb, we have $r \leq \rho(K_{n,T})$ where $T$ is some template. Since $n = |R|$, Proposition~\ref{prop:bs1:imprint} yields that $\rho(K_{n,T}) \in S$. Since $S$ is \bscu-saturated, it is closed under downset, which yields $r \in S$. This concludes the main argument. It remains to prove the two propositions.

\subsection*{Proof of Proposition~\ref{prop:bs1:imprint}}

Let $n = |R|$. Consider a template $T$. We show that $\rho(K_{n,T}) \in S$. If $T$ is empty, then $K_{n,T} = \{\varepsilon\}$ by definition. Consequently, $\rho(K_{n,T}) = \rho(\varepsilon) \in \itriv{\rho}$ and we get $\rho(K_{n,T}) \in S$ since $S$ is \bscu-saturated. Otherwise, $T = (t_1,\dots,t_\ell)$ for $\ell \geq 1$ and,
\[
  K_{n,T} = K_{n,t_1} \cdots K_{n,t_\ell}
\]
Hence, $\rho(K_{n,T}) = \rho(K_{n,t_1}) \cdots \rho(K_{n,t_\ell})$. Since $S$ is closed under multiplication by definition of \bscu-saturated subsets, it suffices to prove the following lemma.

\begin{lemma} \label{lem:bins}
  For any \unit $t$, we have $\rho(K_{n,t}) \in S$.
\end{lemma}

\begin{proof}
  If $t = a \in A$, then $\rho(K_{n,t}) = \rho(a) \in \itriv{\rho} \subseteq S$ by definition of \bscu-saturated subsets. Otherwise $t = (b,B,c)$ with $B \subseteq A$ and $b,c \in B$. In that case, $K_{n,t} = B^* b (\fullcont{B})^n c B^*$. Since $n = |R|$ by definition, a standard fact on finite monoids yields a number $p \leq n$ (which depends on $B$) such that $(\rho(\fullcont{B}))^p = (\rho(\fullcont{B}))^\omega$. Clearly, $K_{n,T} \subseteq (\fullcont{B})^p$.  Since $\rho$ is a \mratm, we get
  \[
    \rho(K_{n,t}) \subseteq \rho((\fullcont{B})^p)= (\rho(\fullcont{B}))^p = (\rho(\fullcont{B}))^\omega.
  \]
  Since $S$ is \bscu-saturated, we have $(\rho(\fullcont{B}))^\omega \in S$. Finally, since \bscu-saturated subsets are closed under downset, the above inclusions yield $\rho(K_{n,t}) \in S$, as desired.
\end{proof}

\subsection*{Proof of Proposition~\ref{prop:bs1:cover}}

Let $n \geq 1$. We have to show that for any word $w \in A^*$ there exists an unambiguous template  $T$ of length at most $(n+2)^{|A|}-1$ such that $w \in K_{n,T}$. This requires a few additional definitions. For any \unit $t$, we define $P_{n,t} \subseteq K_{n,t}$ as the following language:
\begin{itemize}
\item If $t = a \in A$, then $P_{n,t}  = K_{n,T} =  \{a\}$.
\item If $t = (b,B,c)$ with $B \subseteq A$ and $b,c \in B$, then $P_{n,t}  = (\fullcont{B})^{n+2} \subseteq K_{n,t} = B^* b (\fullcont{B})^{n} cB^*$.
\end{itemize}
Finally, given a template $T$, we define $P_{n,T} = \{\varepsilon\}$ when $T$ is empty and if $T = (t_1,\dots,t_\ell)$ has length $\ell \geq 1$ we define,
\[
  P_{n,T} = P_{n,t_1} \cdots P_{n,t_\ell}
\]
Clearly, we have $P_{n,T} \subseteq K_{n,T}$ for any template $T$. Therefore, it suffices to show the following lemma.

\begin{lemma} \label{lem:bs1last}
  For any word $w \in A^*$ there exists an unambiguous template  $T$ of length at most $(n+2)^{|A|}-1$ such that $w \in P_{n,T}$.
\end{lemma}

We now focus on proving Lemma~\ref{lem:bs1last}. We first get rid of the unambiguity condition.

\begin{fct} \label{fct:ridofunambig}
  Let $T$ be a (possibly ambiguous) template and let $\ell$ be its length. Then, there exists an unambiguous template $T'$ of length $\ell' \leq \ell$ and such that $P_{n,T} \subseteq P_{n,T'}$.
\end{fct}

\begin{proof}
  If $T$ is already unambiguous, then it suffices to choose $T = T'$. Otherwise, one may easily construct another (possibly ambiguous) template $T''$ of length $\ell'' < \ell$ such that $P_{n,T} \subseteq P_{n,T''}$. The key idea is that two consecutive \units witnessing ambiguity of $T$ are merged into a single one in $T''$. Since any template of length $0$ or $1$ is unambiguous, repeating this construction eventually yields the desired template $T'$.
\end{proof}

By Fact~\ref{fct:ridofunambig}, it suffices to prove that for any word $w \in A^*$, there exists a template $T$ of length at most $(n+2)^{|A|}-1$ such that $w \in P_{n,T}$ (the fact then yields the desired unambiguous template).

\medskip

The proof is by induction on the alphabet of $w$. Let $B = \cont{w}$. We prove by induction on $|B|$ that there exists a template $T$ of length at most $(n + 2)^{|B|}-1$ such that $w \in P_{n,T}$. If $|B| = 0$, then $w = \varepsilon$ and we choose $T$ as the empty template.

Assume now that $|B| \geq 1$. We use induction to prove that if there exists a template $U$ of length $\ell \geq (n + 2)^{|B|}$ such that $w \in P_{n,U}$, then there exists another template $T$ with strictly smaller length and such that $w \in P_{n,T}$. Since $w \in P_{n,U}$ where $U$ is the template of length $|w|$ whose elements are the letters of $w$, it suffices to apply this argument repeatedly to get the desired template $T$ of length at most than $(n + 2)^{|B|} - 1$.

Assume therefore that we have $\ell \geq (n + 2)^{|B|}$ and a template $U = (t_1,\dots,t_\ell)$ such that $w \in P_{n,U}$. This means that $w$ admits a decomposition $w = w_1 \cdots w_{\ell}$ such that for all $i$, one of the two following properties is satisfied:
\begin{enumerate}
\item $t_i = a_i \in A$ and $w_i = a_i$ or,
\item $t_i = (b_i,B_i,c_i)$ and $w_i  \in (\fullcont{B_i})^{n+2}$.
\end{enumerate}
Let $\ell' = (n+2)^{|B|-1}$. We consider two distinct cases. First, we assume that there exists $i \leq \ell - \ell'$ such that $\cont{w_{i+1}w_{i+2}\cdots w_{i+\ell'}} = C \subsetneq B$. By induction hypothesis, $w_{i+1}w_{i+2}\cdots w_{i+\ell'} \in P_{n,V}$ where $V$ is a template of length at most $(n+2)^{|C|}-1 < \ell'$. We may now replace the \units $t_{i+1},t_{i+2}\cdots t_{i+\ell'}$ in $U$ with those of $V$ which yields the desired template $T$ of length strictly smaller than $\ell$ (we replaced $\ell'$ \units with $(n+2)^{|C|}-1 < \ell'$ \units).

Otherwise, we know that for all $i$, $\cont{w_{i+1}w_{i+2}\dots w_{i+\ell'}} = B$. Since $\ell = (n+2)\ell'$, it follows that $w$ can be redecomposed into $n+2$ factors $w = w'_1 \cdots w'_{n+2}$ such that for all $i$, $\cont{w'_i} = B$. We conclude that $w \in (\fullcont{B})^{n+2}$. In other words, $w \in P_{n,T}$ where $T$ is a template of length $1$ whose unique \unit is a triple $(b,B,c)$ where $b,c$ are arbitrarily chosen letters in $B$. This concludes the proof of Lemma~\ref{lem:bs1last} and of the completeness part of Theorem~\ref{thm:bs1:main}.\qed

\section{Example for universal covering: two-variable first-order logic}
\label{sec:fo2}
In this section, we present a second example illustrating our framework. We detail and prove a characterization of \fod-optimal universal \imprints which yields an algorithm for \fod-covering. Here, \fod denotes two-variable "first-order logic". We first briefly recall the definition of \fod and then present the characterization.

\subsection{Definition} The two-variable fragment of "first-order logic" (denoted \fod) is obtained by restricting the number of  allowed variables within a single sentence to at most \emph{two}. That is, we define \fod sentences as the "first-order" sentences containing at most two distinct variable names. Let us point out that while one is restricted to two variables only, it is possible to \emph{reuse them}. For example, consider the following sentence:
\[
  \exists x\exists y\quad a(x) \quad\wedge\quad a(y) \quad\wedge\quad (x < y) \quad\wedge\quad \exists \underline{\red{x}}\ (a(\underline{\red{x}}) \wedge y < \underline{\red{x}})
\]
This sentence defines the language $A^*aA^*aA^*aA^*$. Observe that in order to quantify the third position whose label is ``$a$'',  we reused the variable $x$: this is allowed in \fod. It is folklore that the class of languages defined by \fod is a \vari.

We shall need the following classification of languages in \fod when proving our characterization. It is standard to classify "first-order" sentences (and therefore those in \fod)  using the notion of quantifier rank. Consider a sentence $\varphi$. We define the \emph{rank} of $\varphi$ as the longest sequence of nested quantifiers within its parse tree. It is well-known and simple to verify that for any $k \in \nat$, there are only finitely many non-equivalent \fod sentences of rank $k$.

We now consider the canonical equivalences over $A^*$ associated to each possible quantifier rank and state standard properties. Given $k \in \nat$, the canonical equivalence \keqfod is defined as follows. For any two words $w,w' \in A^*$,
\[
  \text{$w \keqfod w'$\qquad if\qquad for any \fod sentence $\varphi$ of rank at most $k$,} \quad w \models \varphi \Longleftrightarrow w' \models \varphi.
\]
It is classical and simple to verify that for any $k \in \nat$, the equivalence relation \keqfod has finite index and that all equivalence classes are definable by \fod sentences of rank $k$. We now present two classical properties of the equivalences \keqfod, which we shall need when proving our characterization. Since they are folklore, the proof is left to the reader (they are proved using \efgame arguments). The first property states that \keqfod is a congruence for word concatenation.

\begin{lemma} \label{lem:fo2:fo2cong}
  For all $k \in \nat$, \keqfod is a congruence for concatenation. That is, given $w^{}_1,w^{}_2,w'_1,w'_2 \in A^*$ such that $w^{}_1 \keqfod w'_1$ and $w^{}_2 \keqfod w'_2$, we have $w_1w_2 \keqfod w'_1w'_2$.
\end{lemma}

The second lemma states the characteristic property of two-variable "first-order logic". In fact, our characterization of \fod-"optimal" "\imprints" is based on this property.

\begin{lemma} \label{lem:fo2:efprop}
  Let $B \subseteq A$ and let $u,v \in A^*$ be words such that $\cont{u} = \cont{v} = B$. Let $k,\ell \in \nat$ be integers such that $\ell \geq k$. Then for any $w,w' \in B^*$, we have,
  \[
    u^{\ell} w v^{\ell} \keqfod u^{\ell} w' v^{\ell}.
  \]
\end{lemma}

\subsection{Alphabet compatible \ratms} Our characterization of \fod-"optimal" universal "\imprints" applies to a class of "\mratms" fulfilling additional properties. Let us define this class and then explain why we may further restrict ourselves to it without loss of generality.

Observe that the set $2^{2^A}$ of all sets of sub-alphabets of $A$ is an idempotent commutative monoid when equipped with union. Consider a \ratm $\rho: 2^{A^*} \to R$. We say that $\rho$ is \emph{alphabet compatible} when there exists a monoid morphism $r \mapsto \cont{r}$ from $(R,+)$ to $(2^{2^A},\cup)$ such that:
\begin{equation} \label{eq:fo2:compat}
  \text{For every $K \subseteq A^*$, we have $\cont{\rho(K)} = \{\cont{w} \mid w \in K\}$}.
\end{equation}
The key idea behind this definition is that for every language $K$, the image $\rho(K)$ of $K$ in $R$ records all sub-alphabets $B \subseteq A$ such that $K \cap B^\circledast \neq \emptyset$ (recall that for any sub-alphabet $B \subseteq A$, we write $\fullcont{B} = \{w \in A^* \mid \cont{w} = B\}$ for the set of words whose alphabet is exactly $B$).

\medskip
When dealing with an alphabet compatible "\ratm", we shall implicitly assume that the morphism $r \mapsto \cont{r}$ is fixed and that we have it in hand. We now use extension to show that we may restrict ourselves to "\nice" alphabet compatible "\mratms" without loss~of~generality.

\begin{lemma} \label{lem:compalph}
  Given as input a "\nice" "\mratm" $\rho: 2^{A^*} \to R$, one may compute a "\nice" alphabet compatible \mratm $\tau: 2^{A^*} \to Q$ extends $\rho$ together with the associated extending morphism $\delta: Q \to R$.
\end{lemma}

\begin{proof}
  The set $2^{2^{A}}$ is clearly an idempotent semiring for the following addition and multiplication:
  \begin{itemize}
  \item The addition is union and its neutral element is $\emptyset$.
  \item The multiplication is as follows. Given $S,T \in 2^{2^{A}}$ (\emph{i.e.}, $S$ and $T$ are sets of sub-alphabets):
    \[
      S \cdot T = \{B \cup C \mid B \in S \text{ and } C \in T\}.
    \]
    The neutral element of the multiplication is $\{\emptyset\}$.
  \end{itemize}
  Consequently, the Cartesian product $R \times 2^{2^{A}}$ is also an idempotent semiring for the componentwise addition and multiplication. We define $\tau: 2^{A^*} \to R \times 2^{2^{A}}$ as follows:
  \[
    \begin{array}{llll}
      \tau: & 2^{A^*} & \to     & R \times 2^{2^{A}} \\
            & K       & \mapsto & (\rho(K),\{\cont{w} \mid w \in K\}).
    \end{array}
  \]
  Clearly, $\tau$ is a "\nice" "\mratm" extending $\rho$: the extending morphism $\delta: R \times 2^{2^{A}} \to R$ is simply the projection on the first component. Moreover, it is simple to verify that $\tau$ is alphabet compatible: for every $(r,S) \in R \times 2^{2^{A}}$, it suffices to define $\cont{(r,S)} = S$. This concludes the proof.
\end{proof}

\begin{remark}\label{rem:fo2:alphcompblowup}
  For a \mratm $\rho: 2^{A^*} \to R$ to be alphabet compatible, its rating set $R$ needs to have at least $2^{2^{|A|}}$ elements, \emph{i.e.}, at least as many elements as there are sets of sub-alphabets. This can be observed in the proof of Lemma~\ref{lem:compalph}: when building an alphabet compatible \ratm out of an arbitrary one, the new rating set is the Cartesian product of the old one with $2^{2^A}$.

  However, when computing \fodopti (which is our objective in this section) from a possibly non alphabet compatible \mratm $\rho$, the blowup is only singly exponential with respect to $|A|$. Indeed, one may verify that in the least fixpoint algorithm given in the next subsection, the second component of the
  will always be a singleton.
\end{remark}

\subsection{Characterization of optimal \imprints} We are now ready to present our characterization of \fod-"optimal" universal \imprints. That is, given an alphabet compatible \mratm $\rho$, we describe the set \fodopti.

\begin{remark}
  In this case as well, our characterization does not require $\rho$ to be "\nice": it holds for any alphabet compatible "\mratm".
\end{remark}

Consider an alphabet compatible \mratm $\rho: 2^{A^*} \to R$ and a subset $S \subseteq R$. We say that $S$ is \emph{\fod-saturated} (for $\rho$) if it contains \itriv{\rho} and is closed under the following operations:
\begin{enumerate}
\item \emph{\bfseries Downset:} For any $r \in S$ and any $r' \leq r$, we have $r' \in S$.
\item \emph{\bfseries Multiplication:} For any $s,t \in S$, we have $st \in S$.
\item\label{op:fo2:main} \emph{\bfseries \fod-closure:} For any $B \subseteq A$ and any (multiplicative) idempotents $e,f \in S$ such that $\cont{e} = \cont{f} = \{B\}$, we have:
  \[
    e \cdot \rho(B^*) \cdot f \in S.
  \]
\end{enumerate}

We now state the main theorem of this section: for any alphabet compatible "\mratm" $\rho: 2^{A^*} \to R$, the set \fodopti is the least \fod-saturated subset of $R$ (for inclusion).

\begin{theorem}[Characterization of \fod-"optimal" \imprints] \label{thm:fo2:main}
  Let $\rho: 2^{A^*} \to R$ be an alphabet compatible \mratm. Then, \fodopti is the least \fod-saturated subset of $R$.
\end{theorem}

Clearly, Theorem~\ref{thm:fo2:main} yields an algorithm for computing \fodopti from any input \nice \mratm $\rho: 2^{A^*} \to R$. Indeed, by Lemma~\ref{lem:compalph}, we may compute a \nice \emph{alphabet compatible} \mratm $\tau: 2^{A^*} \to Q$ which extends $\rho$ together with the associated extending morphism $\delta: Q \to R$. It is then simple to compute the least \fod-saturated subset of $Q$ (for $\tau$) with a least fixpoint algorithm, as it is immediate that \fod-closure may be implemented. By Theorem~\ref{thm:fo2:main}, this set is exactly \opti{\fod}{\tau}. Finally, we use Lemma~\ref{lem:bgen:extension} to obtain $\fodopti = \dclos \delta(\opti{\fod}{\tau})$.

Consequently, Proposition~\ref{prop:breduc} yields an algorithm for universal "\fod-covering". Moreover, we obtain an algorithm for full "\fod-covering" as well by Proposition~\ref{prop:bgen:eqpoint} since \fod is a Boolean algebra. Thus, we get the following corollary.

\begin{corollary} \label{cor:fo2:main}
  The "\fod-covering problem" is decidable.
\end{corollary}

We now turn to the proof of Theorem~\ref{thm:fo2:main}. Consider an alphabet compatible \mratm $\rho: 2^{A^*} \to R$. Our objective is to show that \fodopti is the least \fod-saturated subset of $R$. As before, we follow two independent steps corresponding respectively to soundness and completeness of the least fixpoint procedure.

\subsection{Soundness}

We show that \fodopti is \fod-saturated. Since \fod is a \vari of regular languages, we already know from Lemma~\ref{lem:bgen:genalgoba} that \fodopti contains \itriv{\rho} and is closed under downset and multiplication. Thus, we may focus on \emph{\fod-closure}. Given $B \subseteq A$ and any idempotents $e,f \in \fodopti$ such that $\cont{e} = \cont{f} = \{B\}$, we have to show that:
\[
  e \cdot \rho(B^*) \cdot f \in \fodopti.
\]
By definition, we have to prove that for any universal \fod-cover \Kb, we have $e \cdot \rho(B^*) \cdot f \in \prin{\rho}{\Kb}$.

We use the equivalences \keqfod associated to \fod. Since \Kb is a finite set, there exists $k \geq 1$ such that any $K \in \Kb$ is defined by an \fod sentence of rank at most $k$. We let \Hb be the partition of $A^*$ into \keqfod-classes. By definition, \Hb is a universal \fod-cover. Consequently, since $e,f \in \fodopti$, we have $e,f \in \prin{\rho}{\Hb}$, which yields two \keqfod-classes $H_e,H_f \in \Hb$ such that $e \leq \rho(H_e)$ and $f \leq \rho(H_f)$. Let us define $L$ as the following language:
\[
  L = (H_e)^{k}B^*(H_f)^{k}.
\]
Observe that since $e,f$ are idempotents, we have $e \cdot \rho(B^*) \cdot f \leq \rho(L)$. Moreover, since \Kb is a universal cover and $L$ is nonempty (as $H_e,H_f$ cannot be empty since they are \keqfod-classes), there exists a language $K \in \Kb$ such that $L \cap K \neq \emptyset$. Using our choice of $k \in \nat$, we prove that $L \subseteq K$. This will imply that $e \cdot \rho(B^*) \cdot f \leq \rho(L) \leq \rho(K)$, which will yield $e \cdot \rho(B^*) \cdot f \in \prin{\rho}{\Kb}$, finishing the soundness proof.

\medskip

We now concentrate on proving that $L \subseteq K$. Recall that by hypothesis, $K \in \Kb$ is defined by an \fod sentence of rank at most $k$. Moreover, we have $L \cap K \neq \emptyset$. Hence, by definition of \keqfod, it suffices to show that all words in $L$ are \keqfod-equivalent to conclude as desired that $L \subseteq K$. We use the following fact, which is based on the definition of alphabet compatible \ratms.

\begin{fct} \label{fct:fo2:soundfact}
  For any $u \in H_e$ and $v \in H_f$, we have $\cont{u} = \cont{v} = B$.
\end{fct}

\begin{proof}
  We present a proof for $u \in H_e$. The argument is symmetrical for $v \in H_f$. It simple to verify that $\fullcont{(\cont{u})}$ may be defined by an \fod sentence of rank $1$. Hence, since $k \geq 1$ and $H_e$ is a \keqfod-class containing $u \in \fullcont{(\cont{u})}$, it follows that $H_e \subseteq \fullcont{(\cont{u})}$ by definition of \keqfod. By definition of alphabet compatible \ratms, this implies that $\cont{\rho(H_e)} = \{\cont{u}\}$. Finally, since $e \leq \rho(H_e)$ and $r\mapsto \cont{r}$ is a monoid morphism from $R$ to $2^{2^A}$, we obtain $\{B\}=\cont{e}\subseteq\cont{\rho(H_e)}=\{\cont{u}\}$. Therefore, $\cont{u}=B$, as desired.
\end{proof}

We now consider $u \in H_e$ and $v \in H_f$. By Fact~\ref{fct:fo2:soundfact}, $\cont{u} = \cont{v} = B$. We show that all words in $L$ are $\keqfod$-equivalent to the word $u^kv^k$, which concludes the proof.

Recall that by definition, $H_e,H_f$ are \keqfod-classes. Thus, all words in $H_e$ (resp. $H_f$) are equivalent to $u \in H_e$ (resp $v \in H_f$).  Moreover, we have $L = (H_e)^{k}B^*(H_f)^{k}$ by definition. Since \keqfod is compatible with concatenation (this is Lemma~\ref{lem:fo2:fo2cong}), we get that for any word $w \in L$, there exists $x \in B^*$ such that,
\[
  w \keqfod u^k x v^k
\]
Moreover, since $\cont{u} = \cont{v} = B$ and $x \in B^*$, we obtain from Lemma~\ref{lem:fo2:efprop} that $u^k x v^k \keqfod u^kv^k$. Therefore, transitivity yields that all $w \in L$ are \keqfod-equivalent to $u^kv^k$.

\subsection{Completeness}

We turn to the difficult direction in Theorem~\ref{thm:fo2:main}. Recall that an alphabet compatible \mratm $\rho: 2^{A^*} \to R$ is fixed. Given an \fod-saturated subset $S \subseteq R$, we show $\fodopti \subseteq S$. Our proof is a generic construction which builds an \fod-cover \Kb of $A^*$ such that $\prin{\rho}{\Kb} \subseteq S$. Since $\fodopti \subseteq \prin{\rho}{\Kb}$ for any \fod-cover \Kb of $A^*$, this proves the desired result.

\begin{remark}
  As before, since we showed that \fodopti is \fod-saturated, one may apply our construction in the special case when $S = \fodopti$. This builds a universal \fod-cover \Kb such that $\prin{\rho}{\Kb} = \fodopti$, \emph{i.e.}, one that is optimal for $\rho$.
\end{remark}

The construction of \Kb is achieved by induction on three parameters. The most important one is some sub-alphabet $B \subseteq A^*$: we reduce the construction of \Kb to that of "\fod-covers" of $B^*$ for some $B \subsetneq A$. The induction is stated in Proposition~\ref{prop:fo2:construction}.

\begin{proposition} \label{prop:fo2:construction}
  Consider an \fod-saturated set $S \subseteq R$. For any $B \subseteq A$ and any $t_\ell,t_r \in	S$, there exists an \fod-cover \Kb of $B^*$ such that for all $K \in \Kb$:
  \begin{equation} \label{eq:goalfod}
    K \subseteq B^* \quad \text{and} \quad t_\ell \cdot \rho(K) \cdot t_r \in S.
  \end{equation}
\end{proposition}

Before proving Proposition~\ref{prop:fo2:construction}, let us use it to finish the proof of Theorem~\ref{thm:fo2:main}. Consider an arbitrary \fod-saturated set $S \subseteq R$. We apply Proposition~\ref{prop:fo2:construction} in the case when $B = A$ and $t_\ell,t_r = 1_R = \rho(\varepsilon)$ (which belongs to \itriv{\rho} and therefore to $S$ since $S$ is \fod-saturated). We obtain a universal \fod-cover \Kb such that $\rho(K) \in S$ for any $K \in \Kb$. Since $S$ is closed under downset, this implies $\prin{\rho}{\Kb} \subseteq S$ by definition, which concludes the completeness proof.

\medskip

It remains to prove Proposition~\ref{prop:fo2:construction}. We let $B \subseteq A$ and $t_\ell,t_r \in S$ be as in the statement of the proposition. Our objective is to construct an \fod-cover \Kb of $B^*$ such that for all $K \in \Kb$, we have $K \subseteq B^*$ and $t_\ell \cdot \rho(K) \cdot t_r \subseteq S$. The proof is by induction on three parameters. To present them, we need an additional definition. We define $S_B \subseteq S$ as the following subset of~$S$:
\[
  S_B = \{s \in S \mid \text{$s = \rho(K)$ for some nonempty language $K \subseteq B^*$}\}.
\]
The following fact is a direct consequence of $S$ being \fod-saturated.

\begin{fct} \label{fct:sbmono}
  $S_B$ is a submonoid of $R$ for multiplication.
\end{fct}

\begin{proof}
  Clearly, $1_R \in S_B$ since $1_R = \rho(\varepsilon) \in \itriv{\rho} \subseteq S$ and $\{\varepsilon\} \subseteq B^*$. For closure under multiplication, consider $q,r \in S_B$. By definition, $q,r \in S$ and we have nonempty $K_q,K_r \subseteq B^*$ such that $q = \rho(K_q)$ and $r = \rho(K_r)$. Since $S$ is closed under multiplication, we get $qr \in S$ and clearly, $qr = \rho(K_qK_r)$ with $K_qK_r \subseteq B^*$. Thus, $qr \in S_B$.
\end{proof}

We may now define our three induction parameters. The first and most important one is simply the size of the sub-alphabet $B$. The other two depend on $t_\ell$ and $t_r$ respectively. Given $s_1,s_2 \in S$, we say that,
\begin{itemize}
\item $s_2$ is \emph{{\bf left} $B$-reachable} from $s_1$ when there exists $x \in S_B$ such that $s_2 = x \cdot s_1$.
\item $s_2$ is \emph{{\bf right} $B$-reachable} from $s_1$ when there exists $x \in S_B$ such that $s_2 = s_1 \cdot x$.
\end{itemize}
Since $S_B$ is a monoid for multiplication by Fact~\ref{fct:sbmono}, it is immediate that left and right $B$-reachability are both preorder relations on the set $S$. For each $s \in S$, we define the \emph{left $B$-index of $s$} (resp.\ the \emph{right $B$-index of $s$}) as the number of elements that are left $B$-reachable (resp.\ right $B$-reachable) from $s$. We proceed by induction on the following parameters listed by order of importance:
\begin{enumerate}
\item The size of $B$.
\item The right $B$-index of $t_\ell$.
\item The left $B$-index of $t_r$.
\end{enumerate}

\medskip

We consider three cases depending on properties of $B$, $t_\ell$ and $t_r$ that we define below. First observe that since $S$ contains \itriv{\rho} by definition of \fod-saturated sets, it is immediate that for all $b \in B$, we have $\rho(b) \in S_B$. We may now present the three cases of our construction. They depend on the two following properties of $t_\ell$ and $t_r$ respectively.
\begin{itemize}
\item We say that $t_\ell$ is \emph{right saturated} (by $B$) if for all $b \in B$, there exists $x \in S_B$ and $t_\ell$ is right $B$-reachable from $t_\ell \cdot x \cdot \rho(b)$.
\item We say that $t_r$ is \emph{left saturated} (by $B$) if for all $b \in B$, there exists $x \in S_B$ and $t_r$ is left $B$-reachable from $\rho(b) \cdot x \cdot t_r$.
\end{itemize}

The base case happens when $t_\ell$ and $t_r$ are respectively right and left saturated by $B$. When $t_\ell$ is not right saturated, we use induction on $|B|$ and the right $B$-index of $t_\ell$. Finally, when $t_r$ is not left saturated, we use induction on $|B|$ and the left $B$-index of $t_r$. Let us start with the base case.

\subsubsection*{Base case: $t_\ell$ is right saturated by $B$ and $t_r$ is left
  saturated by $B$}

In that case, we simply choose $\Kb = \{B^*\}$, which is clearly an \fod-cover of $B^*$. We have to use our hypothesis to show that this choice satisfies the conditions of Proposition~\ref{prop:fo2:construction}. Clearly, $K\subseteq B^*$ for all $K\in\Kb$. It remains to prove that~\eqref{eq:goalfod} holds:
\[
  t_\ell \cdot \rho(B^*) \cdot t_r \in S.
\]
This is what we do now using the following lemma.

\begin{lemma} \label{lem:fo2:basecase}
  There exist idempotents $e,f \in S$ such that $\cont{e} = \cont{f} = \{B\}$, $t_\ell = t_\ell e$ and $t_r = f t_r$.
\end{lemma}

Before proving the lemma, let us use it to conclude this case. Let $e,f$ be as defined in the lemma. This yields,
\[
  t_\ell \cdot \rho(B^*) \cdot t_r = t_\ell \cdot e \cdot \rho(B^*) \cdot f \cdot t_r.
\]
Moreover, since $e,f \in S$ are idempotents such that $\cont{e} = \cont{f} = B$, we know from \fod-closure in the definition of \fod-saturated sets that,
\[
  e \cdot \rho(B^*) \cdot f \in S.
\]
Since $t_\ell,t_r \in S$, one may now use closure under multiplication in the definition of \fod-saturated sets to obtain that $t_\ell \cdot \rho(B^*) \cdot t_r \in S$ as desired. This concludes the proof of the base case. It remains to show Lemma~\ref{lem:fo2:basecase} in order to finish this case.

\begin{proof}[Proof of Lemma~\ref{lem:fo2:basecase}]
  We only prove the existence of $e$ using the fact that $t_\ell$ is right saturated (the existence of $f$ is proved symmetrically using the fact that $t_r$ is left saturated). By definition of right saturation, we know that for each $b \in B$, there exists $x_b,y_b \in S_B$ such that,
  \[
    t_\ell = t_\ell x_b \cdot \rho(b) \cdot y_b.
  \]
  Let $B = \{b_1,\dots,b_n\}$. We define $e$ as the following element (where $\omega\in\mathbb{N}$ denotes the idempotent power for the multiplication of $R$):
  \[
    e =  (x_{b_1} \cdot \rho(b_1) \cdot y_{b_1} x_{b_2} \cdot \rho(b_2) \cdot y_{b_2}\cdots x_{b_n} \cdot \rho(b_n) \cdot y_{b_n})^\omega.
  \]
  By definition, $e$ is idempotent and we have $t_\ell = t_\ell e$. Moreover, we have $e \in S$ since it is a multiplication of elements in $S$ and \fod-saturated sets are closed under multiplication. It remains to prove that $\cont{e}$ is equal to $\{B\}$.

  By hypothesis, for every $b \in B$, we have $x_b,y_b \in S_B$ which yields nonempty languages $H_b,K_b \subseteq B^*$ such that $x_b = \rho(H_b)$ and $y_b = \rho(K_b)$. Consequently,
  \[
    e = \rho((H_{b_1}b_1K_{b_1}H_{b_2}b_2K_{b_2} \cdots H_{b_n}b_nK_{b_n})^\omega).
  \]
  Clearly, the language $(H_{b_1}b_1K_{b_1}H_{b_2}b_2K_{b_2} \cdots H_{b_n}b_nK_{b_n})^\omega$ is nonempty and included in $\fullcont{B}$ since $b_1,\dots,b_n$ account for all letters in $B$. By definition of alphabet compatible \ratms, this yields $\cont{e} = \{B\}$.
\end{proof}

It remains to treat the inductive case, when either $t_\ell$ is not right saturated or $t_r$ is not left saturated. Since these two cases are symmetrical, we only consider the one when $t_\ell$ is not right saturated.

\subsubsection*{Inductive case: $t_\ell$ is not right saturated}

the fact that $t_\ell$ is not right saturated means that there exists some $b \in B$ such that for all $q \in S_B$, the element $t_\ell$ is {\bf not} right $B$-reachable from $t_\ell q \cdot \rho(b)$. Let us fix such a letter $b$. We use it to construct our \fod-cover \Kb of $B^*$.

Let $C = B \setminus \{b\}$. Using induction on the size of the alphabet we may apply Proposition~\ref{prop:fo2:construction} in the case $t_\ell = t_r = 1_R$ (recall that the alphabet is the most important induction parameter), to get an \fod-cover \Hb of $C^*$ such that for all $H \in \Hb$:
\begin{equation} \label{eq:fo2:hprop1}
  H \subseteq C^* \quad \text{and} \quad \rho(H) \in S.
\end{equation}
For all $H \in \Hb$, we define $t_H = t_\ell \cdot \rho(H) \cdot \rho(b)$. By definition, we have the following fact.

\begin{fct} \label{fct:thesubcover}
  For all $H \in \Hb$, we have $t_H \in S$.
\end{fct}

\begin{proof}
  Immediate by closure under multiplication since $t_\ell, \rho(H)$ and $\rho(b)$ all belong to $S$.
\end{proof}
Moreover, for any $H \in \Hb$, we have $H \subseteq C^* \subseteq B^*$ and $\rho(H) \in S$, which yields $\rho(H) \in S_B$. By choice of the letter $b$, we know that $t_\ell$ is not right $B$-reachable from $t_H$. In particular, it follows that the right $B$-index of $t_H$ is strictly smaller than the right $B$-index of $t_\ell$. Hence, for any $H \in \Hb$, we may use induction in Proposition~\ref{prop:fo2:construction} to construct an \fod-cover $\Kb_H$ of $B^*$ such that for any $K \in \Kb_H$,
\begin{equation} \label{eq:fo2:hprop2}
  K \subseteq B^* \quad \text{and} \quad t_H \cdot \rho(K) \cdot t_r \in S.
\end{equation}
We are now ready to construct our \fod-cover \Kb of $B^*$ which satisfies the properties described in Proposition~\ref{prop:fo2:construction}.  We define,
\[
  \Kb = \Hb \cup \bigcup_{H \in \Hb} \{HbK \mid K \in \Kb_H\}.
\]
It remains to prove that \Kb is an \fod-cover of $B^*$ and that~\eqref{eq:goalfod} holds. We begin by proving that all languages $K \in \Kb$ are \fod-definable. This is immediate if $K \in \Hb$ (by construction of \Hb). Otherwise, $K = HbK'$ with $H \in \Hb$ and $K' \in \Kb_H$. Observe that by definition, $H \subseteq C^*$, $b \in B \setminus C$. Thus, a word $w$ is in $K = HbK'$ if and only if it satisfies the following three properties:
\begin{enumerate}
\item $w$ contains the letter ``$b$''.
\item The prefix of $w$ obtained by keeping all positions that come strictly before the leftmost ``$b$'' belongs to $H$.
\item The suffix of $w$ obtained by keeping all positions that come strictly after the leftmost ``$b$'' belongs to $K'$.
\end{enumerate}
Hence, it suffices to show that these three properties can be expressed in \fod. The first one is defined by ``$\exists x\ P_b(x)$''. For the second and third properties, observe that one can select the position carrying the leftmost ``$b$'' with the following \fod formula:
\[
  \varphi(x)\ :=\  P_b(x) \wedge \neg \exists y\ (P_b(y) \wedge y < x).
\]
Therefore a sentence for the second property above can be obtained from a sentence $\Psi_H$ that defines $H$ (which exists by hypothesis on $H$) by restricting all quantifications to positions $x$ that satisfy the formula $\exists y\ x < y \wedge \varphi(y)$. Similarly, a sentence for the third property can be obtained from a sentence $\Psi_{K'}$ that defines $K'$ (which exists by hypothesis on $K'$) by restricting all quantifications to positions $x$ that satisfy the formula $\exists y\ y < x \wedge \varphi(y)$.

\medskip

We now prove that $\Kb$ is a "cover" of $B^*$. Let $w \in B^*$, we exhibit $K \in \Kb$ such that $w \in K$. If $w \in C^*$, it belongs to some $H \in \Hb \subseteq \Kb$ since $\Hb$ is a "cover" of $C^*$. Otherwise, $b \in \cont{w}$ and $w$ can be decomposed as $w = ubv$ with $u \in C^*$ and $v \in B^*$ (the highlighted ``$b$'' is the leftmost one in $w$). Since $\Hb$ is a "cover" of $C^*$, there exists $H \in \Hb$ such that $u \in H$. Moreover, since $\Kb_H$ is a "cover" of $B^*$, we get $K \in \Kb_H$ such that $v \in K$. It follows that $w \in HbK$, and the language $HbK$ belongs to \Kb by definition.

\medskip

It remains to prove that~\eqref{eq:goalfod} holds: for any $K \in \Kb$, $K \subseteq B^*$ and $t_\ell \cdot \rho(K) \cdot t_r \in S$. Let $K \in \Kb$. That $K \subseteq B^*$ is immediate: it is a concatenation of languages included in $B^*$ by definition. We show that $t_\ell \cdot \rho(K) \cdot t_r \in S$. If $K \in \Hb$, then $\rho(K) \in S$ by~\eqref{eq:fo2:hprop1}. Since $t_\ell,t_r \in S$ it then follows from closure under multiplication that $t_\ell \cdot \rho(K) \cdot t_r \in S$. Assume now that $K =HbK'$ for $H \in \Hb$ and $K' \in \Kb_H$. Since $\rho$ is a \mratm, we have,
\[
  \rho(K) = \rho(H) \cdot \rho(b) \cdot \rho(K').
\]
Recall that $t_H = t_\ell \cdot \rho(H) \cdot \rho(b)$. Therefore, $t_\ell \cdot \rho(K) \cdot t_r = t_H \cdot \rho(K') \cdot t_r$. Finally, since $K' \in \Kb_H$, we have $t_H \cdot \rho(K') \cdot t_r \in S$ by~\eqref{eq:fo2:hprop2}. This concludes the proof of Proposition~\ref{prop:fo2:construction}.

\section{General approach for full covering}
\label{sec:genlatts}
We now turn to the full "\Cs-covering problem". We generalize the methodology presented in Section~\ref{sec:genba}.

\begin{remark}
  This generalized methodology applies to any \pvari of regular languages~\Cs. However, it is more involved than the one which we already presented for universal "\Cs-covering" in Section~\ref{sec:genba}. Hence, it is only meant to be used for lattices that are \emph{not} closed under complement. When dealing with a \vari of regular languages, universal and full covering are equivalent by Proposition~\ref{prop:bgen:eqpoint}. Hence, one may just use the simple methodology of Section~\ref{sec:genba}.
\end{remark}

We begin by presenting a formal reduction from full "\Cs-covering" to a problem whose input is a "\nice" "\mratm". For technical reasons, this requires to slightly generalize our terminology on "\ratms".

\subsection{Optimal pointed \kl{\imprints}}
\label{sec:reform-cover-probl}

As expected, the key idea behind our approach is to reduce "\Cs-covering" to the problem of computing \copti{L,\rho} from a regular language $L$ and a "\nice" "\mratm" $\rho$ (this is exactly what we did for the special case $L = A^*$ in Section~\ref{sec:genba}). However, the situation is slightly more complicated here. While getting such a reduction is fairly simple (given the results that we already have), we shall need to consider an object that is more involved than \copti{L,\rho}.

More precisely, our approach does not work directly with \copti{L,\rho} as this set does not have strong enough properties. Recall that in the case $L = A^*$, what made $"\copti{\rho}" = \copti{A^*,\rho}$ the ``right object'' is Lemma~\ref{lem:boptsemi}: \copti{\rho} is a submonoid of $R$. This is not the case for \copti{L,\rho} in general. We cope with this problem by considering a monoid morphism $\alpha: A^* \to M$ recognizing the language $L$. Our approach considers all sets \opti{\Cs}{\alpha\inv(s),\rho} for $s \in M$ \emph{simultaneously}. As we show in the following proposition, this is enough information to recover \copti{L,\rho}.

\begin{proposition} \label{prop:isreg}
  Let \Cs be a lattice of regular languages and let $L$ be a language recognized by a morphism $\alpha: A^* \to M$ for the accepting set $F$ (\emph{i.e.}, $L = \alpha\inv(F)$). Let $\rho: 2^{A^*} \to R$ be a "\ratm". Then, the following properties hold:
  \begin{itemize}
  \item We have $\opti{\Cs}{L,\rho} = \bigcup_{s \in F} \opti{\Cs}{\alpha\inv(s),\rho}$.
  \item If for any $s \in M$, we have an optimal "\Cs-cover" $\Kb_s$ of $\alpha\inv(s)$ for $\rho$, then $\Kb = \bigcup_{s \in F} \Kb_s$ is an optimal "\Cs-cover" of $L$ for $\rho$.
  \end{itemize}
\end{proposition}

\begin{proof}
  We start with the second item. For any $s \in M$, we let $\Kb_s$ be an optimal "\Cs-cover" of $\alpha\inv(s)$ for~$\rho$. Since $L = \alpha\inv(F)$ and $\Kb_s$ is a "\Cs-cover" of $\alpha\inv(s)$ for any $s \in M$, we know that $\Kb = \bigcup_{s \in F} \Kb_s$ is a "\Cs-cover" of $L$. We need to show that it is optimal. By definition, this amounts to proving that for any arbitrary "\Cs-cover" $\Kb'$ of $L$, we have $\prin{\rho}{\Kb} \subseteq \prin{\rho}{\Kb'}$. Let $r \in \prin{\rho}{\Kb}$. By definition, $r \leq \rho(K)$ for some $K \in \Kb$. Moreover, $K \in \Kb_s$ for some $s \in F$. Since $\Kb_s$ is an optimal "\Cs-cover" of $\alpha\inv(s)$ for $\rho$, it follows that $r \in \opti{\Cs}{\alpha\inv(s),\rho}$. Finally, since $\Kb'$ is a "\Cs-cover" of $L$, it is also a "\Cs-cover" of $\alpha\inv(s) \subseteq L$.   Hence, $\opti{\Cs}{\alpha\inv(s),\rho} \subseteq \prin{\rho}{\Kb'}$ which yields $r \in \prin{\rho}{\Kb'}$, as desired.

  We finish with the first item. We just proved that $\Kb = \bigcup_{s \in F} \Kb_s$ is an optimal "\Cs-cover" of~$L$. Therefore, $\opti{\Cs}{L,\rho} = \prin{\rho}{\Kb}$. Moreover, it is immediate from the definition of "\imprints" that $\prin{\rho}{\Kb}  = \bigcup_{s \in F} \prin{\rho}{\Kb_s}$. Since $\Kb_s$ is an optimal "\Cs-cover" of $\alpha\inv(s)$ for $\rho$, which means that $\prin{\rho}{\Kb_s} = \opti{\Cs}{\alpha\inv(s),\rho}$, we obtain $\opti{\Cs}{L,\rho} = \bigcup_{s \in F} \opti{\Cs}{\alpha\inv(s),\rho}$.
\end{proof}

Altogether, this means that we shall be looking for algorithms which take a morphism $\alpha: A^* \to M$ and a "\nice" "\mratm" $\rho: 2^{A^*} \to R$ as input and compute all \Cs-optimal $\rho$-"\imprints" $\opti{\Cs}{\alpha\inv(s),\rho}$ for $s \in M$. It will be convenient to have a single notation recording all these objects. This leads to the notion of optimal \emph{pointed} "\imprints".

\medskip

Given a lattice \Cs, a morphism $\alpha: A^* \to M$ and a "\ratm" $\rho: 2^{A^*} \to R$, we define the following subset of $M \times R$:
\[
  \popti{\Cs}{\alpha,\rho} = \{(s,r) \in M \times R \mid r \in \opti{\Cs}{\alpha\inv(s),\rho}\}
\]
We call \popti{\Cs}{\alpha,\rho} the \emph{\Cs-optimal $\alpha$-pointed $\rho$-\imprint}. Clearly, the single set \popti{\Cs}{\alpha,\rho} encodes all the sets \copti{\alpha\inv(s),\rho} for $s \in M$. We first present a few properties of this new object and then formally reduce full "\Cs-covering" to the problem of computing \popti{\Cs}{\alpha,\rho}.

\medskip
\noindent
{\bf Properties.} As expected, a first property is that \popti{\Cs}{\alpha,\rho} always contains trivial elements. Recall that given any language $L$ and any "\ratm" $\rho: 2^{A^*} \to R$, we defined $\itriv{L,\rho} \subseteq R$ as the set,
\[
  \itriv{L,\rho} = \{r \in R \mid \text{$\exists\,w \in L$ such that $r \leq \rho(w)$}\}.
\]
We may extend this notation to make it match our new terminology, ``pointed \imprints''. Given a morphism $\alpha: A^* \to M$ and a "\ratm" $\rho: 2^{A^*} \to R$, we define,
\[
  \begin{array}{lll}
    \ptriv{\alpha,\rho} & = & \{(s,r) \in M \times R \mid r \in \itriv{\alpha\inv(s),\rho}\}\\
                        & = & \{(s,r) \in M \times R \mid \text{$\exists\,w \in \alpha\inv(s)$ such that $r \leq \rho(w)$}\}.
  \end{array}
\]
Clearly, when $\rho$ is a "\nice" "\mratm", one may compute \ptriv{\alpha,\rho} from $\alpha$ and $\rho$ by Lemma~\ref{lem:bgen:evtriv}. Moreover, the next fact is immediate from the definition of \popti{\Cs}{\alpha,\rho} and Fact~\ref{fct:bgen:trivialsets}.

\begin{fct} \label{fct:ptriv}
  Let \Cs be a lattice. Consider a  morphism $\alpha: A^* \to M$ and a "\ratm" $\rho: 2^{A^*} \to R$. Then, we have $\ptriv{\alpha,\rho} \subseteq \popti{\Cs}{\alpha,\rho}$.
\end{fct}

More importantly, we recover the key property of optimal universal "\imprints": when \Cs is a \pvari of regular languages and $\rho: 2^{A^*} \to R$ is a "\mratm", \pcopti{\alpha,\rho} is a submonoid of $M \times R$ for the componentwise multiplication. This is a simple corollary of Lemma~\ref{lem:bgen:optsemi}.

\begin{lemma} \label{lem:pointedmono}
  Let \Cs be a \pvari of regular languages. Consider a morphism $\alpha: A^* \to M$ and a "\mratm" $\rho: 2^{A^*} \to R$. Then \popti{\Cs}{\alpha,\rho} is a submonoid of $M \times R$:
  \begin{itemize}
  \item We have $(1_M,1_R) \in \popti{\Cs}{\alpha,\rho}$.
  \item For any $(s,q),(t,r) \in \popti{\Cs}{\alpha,\rho}$, we have $(st,qr) \in \popti{\Cs}{\alpha,\rho}$.
  \end{itemize}
\end{lemma}

\begin{proof}
  That $(1_M,1_R) \in \popti{\Cs}{\alpha,\rho}$ is immediate from Fact~\ref{fct:ptriv}: $\ptriv{\alpha,\rho} \subseteq \popti{\Cs}{\alpha,\rho}$. Indeed, by definition of "\mratms", we have $(1_M,1_R) = (\alpha(\varepsilon),\rho(\varepsilon))$. Moreover, it is immediate by definition that $(\alpha(\varepsilon),\rho(\varepsilon))  \in \ptriv{\alpha,\rho}$.

  Closure under multiplication follows from Lemma~\ref{lem:bgen:optsemi}. Assume that $(s,q),(t,r) \in \popti{\Cs}{\alpha,\rho}$. By definition, this means that $q \in \opti{\Cs}{\alpha\inv(s),\rho}$ and $r \in \opti{\Cs}{\alpha\inv(t),\rho}$. Thus, we get from Lemma~\ref{lem:bgen:optsemi} that $qr \in \opti{\Cs}{\alpha\inv(s)\alpha\inv(t),\rho}$. Moreover, since $\alpha$ is a morphism, we have $\alpha\inv(s)\alpha\inv(t) \subseteq \alpha\inv(st)$. Thus, Fact~\ref{fct:linclus} yields $qr \in \opti{\Cs}{\alpha\inv(st),\rho}$ which exactly says that $(st,qr) \in \popti{\Cs}{\alpha,\rho}$.
\end{proof}

\medskip
\noindent
{\bf Reduction.} Let us finally connect optimal pointed \imprints with the full covering problem. We do so with the following proposition, which states a reduction generalizing the one of Proposition~\ref{prop:breduc}.

\begin{proposition}\label{prop:lreduc}
  Consider a lattice \Cs. There exists a polynomial space reduction from "\Cs-covering" (for input languages given by \nfas or monoid morphisms) to the following decision problem:
  \begin{center}
    \begin{tabular}{ll}
      {\bf Input:}    & A morphism $\alpha: A^* \to M$ and a subset $F_M \subseteq M$. \\ & A "\nice" "\mratm" $\rho: 2^{A^*} \to R$ and a subset $F_R \subseteq R$. \\
      {\bf Question:} & Do we have $(F_M \times F_R) \cap \pcopti{\alpha,\rho} = \emptyset$?
    \end{tabular}
  \end{center}
\end{proposition}

\begin{proof}
  The input of "\Cs-covering" is a pair $(L,\Lb)$ where $L$ is a regular language and \Lb is a finite multiset of regular languages: we want to know whether $(L,\Lb)$ is \Cs-coverable. One may compute in polynomial space a morphism $\alpha: A^* \to M$ recognizing $L$ (this is trivial is $L$ if already given by a morphism and standard if $L$ is given by an \nfa). We let $F_M \subseteq M$ as the corresponding accepting set: $L = \alpha\inv(F_M)$. Moreover, by Proposition~\ref{prop:bgen:mratmeff}, we are able to compute in polynomial space a "\nice" "\mratm" $\rho: 2^{A^*} \to R$, a morphism $\delta: R \to 2^\Lb$ and a subset $F_R$ of $R$ such that:
  \begin{itemize}
  \item $\rho$ extends $\rho_\Lb$ (the canonical "\ratm" associated to \Lb) for the extending morphism $\delta$.
  \item $F_R = \delta\inv(\Lb)$.
  \end{itemize}
  The reduction outputs $\alpha,F_M,\rho$ and $F_R$. It remains to show that this is indeed a reduction from "\Cs-covering" to the problem described in the proposition. That is, $(L,\Lb)$ is \Cs-coverable if and only if $(F_M \times F_R) \cap \pcopti{\alpha,\rho} = \emptyset$.

  By Theorem~\ref{thm:bgen:main} $(L,\Lb)$ being \Cs-coverable is equivalent to $\Lb \not\in \copti{L,\rho_{\Lb}}$. Moreover, by Lemma~\ref{lem:bgen:extension}, we know that $\copti{L,\rho_{\Lb}} = \dclos \delta(\copti{L,\rho})$. Consequently, since $F_R = \delta\inv(\Lb)$ and \Lb is the maximal element of $2^\Lb$, $(L,\Lb)$ is \Cs-coverable if and only if $F_R \cap \copti{L,\rho} = \emptyset$. Finally, by Proposition~\ref{prop:isreg}, we have $\opti{\Cs}{L,\rho} = \bigcup_{s \in F_M} \opti{\Cs}{\alpha\inv(s),\rho}$.  Hence, $(L,\Lb)$ is \Cs-coverable if and only if for any $s \in F_M$, we have $F_R \cap \opti{\Cs}{\alpha\inv(s),\rho} = \emptyset$. By definition of \pcopti{\alpha,\rho}, this last condition is equivalent to $(F_M \times F_R) \cap \pcopti{\alpha,\rho} = \emptyset$, finishing the proof.
\end{proof}

\begin{remark}
  Similarly to what happened for universal covering in Section~\ref{sec:genba}, we also get a reduction for the second stage in "\Cs-covering": computing separating "\Cs-covers" when they exist. Indeed, given some input pair $(L,\Lb)$, we know from Theorem~\ref{thm:bgen:main} that if $(L,\Lb)$ is \Cs-coverable, then any optimal "\Cs-cover" of $L$ for $\rho_\Lb$ is separating for \Lb. We may compute a "\nice" "\mratm" $\rho$ extending $\rho_\Lb$. By Lemma~\ref{lem:bgen:extension}, any optimal "\Cs-cover" of $L$ for $\rho$ is also optimal for $\rho_{\Lb}$. Moreover, we may compute a morphism $\alpha: A^* \to M$ recognizing $L$ (\emph{i.e.}, $L = \alpha\inv(F)$ for some $F \subseteq M$). By Proposition~\ref{prop:isreg}, we know that $\Kb = \bigcup_{s \in F} \Kb_s$ is an optimal "\Cs-cover" of $L$ for $\rho$ when the $\Kb_s$ are optimal "\Cs-covers" of $\alpha\inv(s)$ for $\rho$.
\end{remark}

In view of Proposition~\ref{prop:lreduc}, we may now focus on the problem of computing "optimal" pointed \imprints. We explain how to tackle this new problem now. 

\begin{remark}
  While Proposition~\ref{prop:lreduc} holds for any "lattice", using the methodology that we present now requires at least a \pvari of regular languages. This is mandatory for applying Lemma~\ref{lem:pointedmono}, which is crucial in our approach.
\end{remark}

\subsection{Methodology}

In this generalized setting as well, a key design principle behind the framework is that our algorithms for computing "optimal" pointed \imprints are formulated as \emph{elegant characterization theorems}. Given a \pvari \Cs, a morphism $\alpha: A^* \to M$ and a "\nice" "\mratm" $\rho:2^{A^*} \to R$, we characterize the \Cs-"optimal" $\alpha$-pointed $\rho$-imprint \popti{\Cs}{\alpha,\rho} as the least subset of $M \times R$ which:
\begin{enumerate}
\item includes the trivial elements from \ptriv{\alpha,\rho}, and
\item is closed under a list of operations.
\end{enumerate}
We speak of a \emph{characterization of \Cs-"optimal" pointed \imprints}. As before, the key idea is that such a result yields a least fixpoint procedure for computing \popti{\Cs}{\alpha,\rho}, thus solving the \Cs-"covering problem" by Proposition~\ref{prop:lreduc}. Indeed, one starts from the set of trivial elements and saturates it with the operations in the list.

\begin{remark}
  As before, while we have to restrict ourselves to "\nice" "\mratms" for the computation, the characterizations themselves often hold for all "\mratms". This is the case for our two examples.
\end{remark}

Let us present two examples. Both of them come from the quantifier alternation hierarchy of "first-order logic" over words: the levels \sicu and \sicd. We start with a characterization for \sicu, which we shall detail and prove in the next section.

\begin{example}[Characterization of \sicu-"optimal" pointed \imprints]
  Consider a morphism $\alpha: A^* \to M$ and a "\mratm" $\rho: 2^{A^*} \to R$. Then, \suopti is the least subset of $M \times R$ containing \ptriv{\alpha,\rho} and satisfying the following properties:
  \begin{enumerate}
  \item {\bfseries Downset}: For any $(s,r) \in \suopti$ and any $r ' \leq r$, we have $(s,r') \in \suopti$.
  \item {\bfseries Multiplication}: For any $(s,q),(t,r) \in \suopti$, we have $(st,qr) \in \suopti$.
  \item {\bfseries \sicu-operation}: We have $(1_M,\rho(A^*)) \in \suopti$.
  \end{enumerate}
\end{example}

Our second example is the level \sicd. Note that we shall not prove this result in the paper. It is essentially adapted from the original "separation" algorithm of~\cite{pzqalt}. A (difficult) proof for the formulation that we use below is available in~\cite[Theorem~6.5]{pseps3j}, which solves "covering" for a family of classes that includes \sicd.

\begin{example}[Characterization of \sicd-"optimal" pointed \imprints]   Let $\alpha: A^* \to M$ be a morphism and let $\rho: 2^{A^*} \to R$ be an alphabet compatible "\mratm". Then, \sdopti is the least subset of $M \times R$ containing \ptriv{\alpha,\rho} and closed under the following operations:
  \begin{enumerate}
  \item {\bfseries Downset}: For any $(s,r) \in \sdopti$ and any $r ' \leq r$, we have $(s,r') \in \sdopti$.
  \item {\bfseries Multiplication}: For any $(s,q),(t,r) \in \sdopti$, we have $(st,qr) \in \sdopti$.
  \item {\bfseries \sicd-closure}: For any idempotent $(e,f) \in \sdopti$ (\emph{i.e.}, $e$ is an idempotent of $M$ and $f$ is a multiplicative idempotent of $R$) and for each $B\in\cont{f}$, we have $(e,f \cdot \rho(B^*) \cdot f) \in \sdopti$.
  \end{enumerate}
\end{example}

As expected, most of the properties used in the above example of characterizations are generic: they are satisfied by all \Cs-"optimal" pointed \imprints \popti{\Cs}{\alpha,\rho} (provided that \Cs is a \pvari of regular languages). Let us present these generic properties properly in the following lemma.

\begin{lemma} \label{lem:genalgolatt}
  Let \Cs be a \pvari of regular languages \Cs. Consider a morphism $\alpha: A^* \to M$ and a "\mratm" $\rho: 2^{A^*} \to R$. Then, the \Cs-"optimal" $\alpha$-pointed $\rho$-\imprint $\popti{\Cs}{\alpha,\rho} \subseteq M \times R$ contains \ptriv{\alpha,\rho} and satisfies the following closure properties:
  \begin{enumerate}
  \item {\bfseries Downset}: For any $(s,r) \in \popti{\Cs}{\alpha,\rho}$ and any $r ' \leq r$, we have $(s,r') \in \popti{\Cs}{\alpha,\rho}$.
  \item {\bfseries Multiplication}: For any $(s,q),(t,r) \in \popti{\Cs}{\alpha,\rho}$, we have $(st,qr) \in \popti{\Cs}{\alpha,\rho}$.
  \end{enumerate}
\end{lemma}

\begin{proof}
  That \popti{\Cs}{\Lb,\rho} contains \ptriv{\Lb,\rho} is immediate from Fact~\ref{fct:bgen:trivialsets}. Closure under downset follows from Fact~\ref{fct:bgen:downset}.  Finally, closure under multiplication is a consequence of Lemma~\ref{lem:bgen:optsemi}. 
\end{proof}

\section{\texorpdfstring{Example for full covering: the logic \sicu}{Example for full covering: the logic Σ1}}
\label{sec:sigma1}
In this section, we illustrate the general approach outlined in Section~\ref{sec:genlatts} with a simple example taken from logic. Specifically, we use it to obtain a covering algorithm for the level \sicu within the quantifier alternation hierarchy of "first-order logic". As expected, this algorithm is based on a characterization of \sicu-"optimal" pointed \imprints.

\begin{remark}
  Since this section is about illustrating our framework, we need a simple example and \sicu is an ideal candidate for this. However, it is known that "covering" is decidable for the levels \sicu,\bscu,\sicd,\bscd and \sict in the alternation hierarchy of\/ "\fo" and all the algorithms may be formulated within our framework. The cases of \sicu and \bscu are detailed in the paper. We refer the reader to \cite{pseps3j} for \sicd and \sict and to \cite{pzbpol} for \bscd.
\end{remark}

We shall actually work with an alternate definition of the class corresponding to \sicu. It is folklore and simple to verify that we have the following theorem.

\begin{theorem} \label{thm:sicu}
  Consider a language $L \subseteq A^*$. Then $L$ can be defined by a \sicu sentence if and only if it is a finite union of languages of the form $A^*a_1A^*a_2A^* \cdots A^*a_nA^*$ with $a_1,\dots,a_n \in A$.
\end{theorem}

\subsection{\texorpdfstring{Characterization of \sicu-optimal \imprints}{Characterization of Σ1-optimal \imprints}}

We start by describing the property characterizing \sicu-"optimal" pointed \imprints. Consider a morphism $\alpha: A^* \to M$ into a finite monoid $M$ and a "\mratm" $\rho: 2^{A^*} \to R$.

\begin{remark}
  There is no additional constraint on $\alpha$ and $\rho$. In particular, there is no need for $\rho$ to be "\nice" for the characterization to hold.
\end{remark}

We say that a subset $S \subseteq M \times R$ is \emph{\sicu-saturated} for $\rho$ if it contains \ptriv{\alpha,\rho} and is closed under the following operations: 
\begin{enumerate}
\item \emph{\bfseries Downset:} For any $(s,r) \in S$ and any $r' \leq r$, we have $(s,r') \in S$.
\item \emph{\bfseries Multiplication:} For any $(s,q),(t,r) \in S$, we have $(st,qr) \in S$.
\item\label{op:sig1} \emph{\bfseries \sicu-operation:} We have $(1_M,\rho(A^*)) \in S$.
\end{enumerate}

\smallskip\noindent We are now ready to state the main theorem of this section.

\begin{theorem}[Characterization of \sicu-"optimal" \imprints] \label{thm:sig1}
  Let $\alpha: A^* \to M$ be a morphism into a finite monoid $M$ and let $\rho: 2^{A^*} \to R$ a "\mratm". Then, \suopti is the least \sicu-saturated subset of $M \times R$.
\end{theorem}

Clearly, Theorem~\ref{thm:sig1} yields a least fixpoint procedure for computing \suopti from any morphism $\alpha$ and any \emph{"\nice"} "\mratm" $\rho$ given as input. Indeed, we are able to compute \ptriv{\alpha,\rho} and all operations in the definition of \sicu-saturated subsets are clearly implementable (for \sicu-operation, this is immediate since $\rho(A^*)$ may be computed by Lemma~\ref{lem:bgen:evamramt}). Altogether, we get the desired corollary from Proposition~\ref{prop:lreduc}: "\sicu-covering" is decidable.

\begin{corollary} \label{cor:sig1}
  The "\sicu-covering problem" is decidable.
\end{corollary}

We now turn to the proof of Theorem~\ref{thm:sig1}. The argument is organized around the approach outlined in Section~\ref{sec:genlatts}. Let $\alpha: A^* \to M$ be a morphism into a finite monoid $M$ and let $\rho: 2^{A^*} \to R$ be a "\mratm". We start by proving that the set \suopti is \sicu-saturated (from an algorithmic point of view, this corresponds to soundness of the least fixpoint procedure: it only computes elements of \suopti).

\subsection{Soundness}

We show that $\suopti \subseteq M \times R$ is \sicu-saturated. Since \sicu is a \pvari, we already know from Lemma~\ref{lem:genalgolatt} that \suopti contains \ptriv{\alpha,\rho} and is closed under downset and multiplication. Therefore, we may concentrate on the \sicu-operation.

We need to show that $(1_M,\rho(A^*)) \in \suopti$. By definition, this amounts to proving that $\rho(A^*) \in \opti{\sicu}{\alpha\inv(1_M),\rho}$. Let \Kb be a $\rho$-"optimal" \sicu-cover \Kb of $\alpha\inv(1_M)$. It suffices to prove  that $\rho(A^*) \in \prin{\rho}{\Kb}$ (by definition of \opti{\sicu}{\alpha\inv(1_M),\rho}). This follows from the next fact.

\begin{fact} \label{fct:sound}
  The only language $K \in \sicu$ containing $\varepsilon$ is $A^*$.
\end{fact}

\begin{proof}
  By Theorem~\ref{thm:sicu}, every language in \sicu is a finite union of languages which are of the form $A^*a_1A^*a_2 \cdots A^*a_nA^*$ with $n \in \nat$ and $a_1,\dots,a_n \in A$. The only such language containing $\varepsilon$ is $A^*$.
\end{proof}

We may now finish the proof. Since $\varepsilon \in \alpha\inv(1_M)$ and \Kb is a "cover" of $L$, there exists $K \in \Kb$ containing $\varepsilon$. Moreover, since $K \in \sicu$, it follows from Fact~\ref{fct:sound} that $K = A^*$. Thus, $A^* \in \Kb$, which means that $\rho(A^*) \in \prin{\rho}{\Kb}$, as desired.\qed

\subsection{Completeness}

We now know that the set \suopti is \sicu-saturated. It remains to prove that it is the least such set. Consider an arbitrary \sicu-saturated subset $S \subseteq M \times R$. We show that $\suopti \subseteq S$. Note that from an algorithmic point of view, this direction of the proof corresponds to completeness of the least fixpoint procedure: we show that it computes \emph{all} elements of \suopti.

We proceed as follows. For each $s \in M$, we build a \sicu-cover $\Kb_s$ of $\alpha\inv(s)$ such that for any $r \in \prin{\rho}{\Kb_s}$, we have $(s,r) \in S$. By definition of \suopti, it will then follow that,
\[
  \suopti \subseteq \{(s,r) \in M \times R \mid r \in \prin{\rho}{\Kb_s}\} \subseteq S.
\]

\begin{remark}
  Since we already showed that \suopti itself is \sicu-saturated, the special case when $S = \suopti$ yields, for any $s \in M$, a \sicu-cover $\Kb_s$ of $\alpha\inv(s)$ satisfying:
  \[
    \prin{\rho}{\Kb_s} = \opti{\sicu}{\alpha\inv(s),\rho}.
  \]
  In other words, we are able to build "optimal" "\sicu-covers".
\end{remark}

We may now start the construction. We first associate to any word $w \in A^*$ a language $K_w \in \sicu$ and then use these languages to construct our "covers" $\Kb_s$. For a word $w = a_1 \cdots a_n \in A^*$, define:
\[
  K_w = A^*a_1A^*a_2A^* \cdots A^*a_nA^*.
\]
Clearly, all languages $K_w$ belong to \sicu by Theorem~\ref{thm:sicu}. We may now define our "\sicu-covers" $\Kb_s$. For $s \in M$, define:
\[
  \Kb_s = \{K_w \mid \text{$w \in \alpha\inv(s)$ and $|w| \leq |M|$}\}.
\]
Clearly, all sets $\Kb_s$ are finite. It now remains to prove that they satisfy the desired properties: for all $s \in M$, $\Kb_s$ is a \sicu-cover of $\alpha\inv(s)$ such that for any $r \in \prin{\rho}{\Kb_s}$, we have $(s,r) \in S$.

\medskip

We begin by proving that $\Kb_s$ is a \sicu-cover of $\alpha\inv(s)$ for all $s \in M$. Since we already know that all languages $K_w$ belong to \sicu, it suffices to prove that $\Kb_s$ is a "cover" of $\alpha\inv(s)$. Let $w \in \alpha\inv(s)$. We need to find $K \in \Kb_s$ such that $w \in K$, which is what we prove in the following lemma.

\begin{lemma} \label{lem:pump}
  For any $w \in \alpha\inv(s)$, there exists $K \in \Kb_s$ such that $w \in K$.
\end{lemma} 

\begin{proof}
  We proceed by induction on $|w|$. When $|w| \leq |M|$, we have $K_w \in \Kb_s$ and it is immediate by definition that $w \in K_w$. We now assume that $|w| > |M|$. We prove that there exists $w_1,w_2 \in A^*$ and $v \in A^+$ such that $w = w_1vw_2$ and $w_1w_2 \in \alpha\inv(s)$. Since $|w_1w_2| < w$, this yields $K \in \Kb_s$ such that $w_1w_2 \in K$. Moreover, since $K$ is of the form $A^*a_1A^*a_2A^* \cdots A^*a_nA^*$ by definition, it will follow that $w = w_1vw_2$ belongs to $K$ as well, finishing the proof.

  Let $w = b_1 \cdots b_{|w|}$. Since $w \in \alpha\inv(s)$, we have $\alpha(b_1) \cdots \alpha(b_{|w|}) = \alpha(w) = s$. Since $|w| > |M|$, it follows from the pigeonhole principle that there exist $i,j$ with $i < j$ such that $\alpha(b_1) \cdots \alpha(b_{i}) = \alpha(b_1) \cdots \alpha(b_{j})$. We let $w_1 = b_1 \cdots b_i$, $v = b_{i+1} \cdots b_j$ and $w_2 = b_{j+1} \cdots b_{|w|}$. By definition, we have $w = w_1vw_2$ and $v$ is nonempty since $i < j$. Moreover,
  \[
    \begin{array}{lll}
      \alpha(w_1) & = & \alpha(b_1) \cdots \alpha(b_{i}) = \alpha(b_1) \cdots \alpha(b_{j}),\\
      \alpha(w_2) & = & \alpha(b_{j+1}) \cdots \alpha(b_{|w|})
    \end{array}
  \]
  Therefore, $\alpha(w_1w_2) = \alpha(b_1) \cdots \alpha(b_{|w|}) = \alpha(w) = s$, which concludes the proof.
\end{proof}

It remains to prove that for any $s \in M$ and $r \in \prin{\rho}{\Kb_s}$, we have $(s,r) \in S$. Recall that $S$ is an arbitrary $\sicu$-saturated set. Since $r \in \prin{\rho}{\Kb_s}$,  we have $K \in \Kb_s$ such that $r \leq \rho(K)$. Since $S$ is \sicu-saturated, it is closed under downset, so that it suffices to show that $(s,\rho(K)) \in S$.

By definition of $\Kb_s$, there exists $w \in \alpha\inv(s)$ such that $K = K_w$. Therefore, $K = A^*a_1A^* \cdots A^*a_nA^*$ with $w = a_1 \cdots a_n$. By definition, we have $(\alpha(a_i),\rho(a_i)) \in \ptriv{\alpha,\rho}$ for all $i \leq n$. Consequently, we obtain $(\alpha(a_i),\rho(a_i)) \in S$ for all $i \leq n$ since $S$ is \sicu-saturated. Moreover, we know from the \sicu-specific operation that $(1_M,\rho(A^*)) \in S$. It then follows from closure under multiplication that,
\[
  \left(1_M \cdot \prod_{1 \leq i \leq n} (\alpha(a_i) \cdot 1_M),\quad\rho(A^*) \cdot \prod_{1 \leq i \leq n}(\rho(a_i) \cdot \rho(A^*))\right) \in S
\]
Since $\alpha$ is a morphism and $\rho$ is a "\mratm" (and therefore a morphism for multiplication), this exactly says that $(\alpha(w),\rho(K)) \in S$. Finally, $\alpha(w)=s$ by definition, and we obtain as desired that $(s,\rho(K)) \in S$, finishing the proof.

\section{Conclusion}
\label{sec:conc}
In this paper, we introduced the "covering problem", which we designed as a replacement of the two problems that are commonly used to investigate classes of languages: "membership" and "separation". With "covering", we get the best of both worlds: like for separation, the problem is flexible enough, so that we are able to obtain covering algorithms for classes that seem to be beyond reach when dealing with "membership" alone. Like for "membership", we have a solid methodology and we recover what was missing in the case of "separation": constructiveness. Moreover, the framework is adapted not only to "Boolean algebras", but also to "lattices".

As an application, we have presented covering algorithms for five fragments of first-order logic: \fo itself, its two-variable fragment \fod, as well as levels $\frac12$, 1 and $\frac 32$ in the Straubing-Thérien hierarchy (denoted \sicu, \bscu and \sicd, respectively), and we have proved these algorithms for \sicu, \bscu and \fod.

There are many natural questions in this line of research. One is to push further the investigation of this problem to other classes of languages, in particular for those of the quantifier alternation hierarchy. It would also be interesting to get precise complexity bounds for this problem, starting from automata as inputs, or from monoids. Finally, whether such problems are meaningful for other structures than words remains to be explored.

\bibliographystyle{abbrv}

\end{document}